\documentclass[11pt, draft, onecolumn]{IEEEtran}

\usepackage{color}
\usepackage{epsf,psfrag,amssymb,amsfonts,latexsym,cite,verbatim}
\usepackage{enumerate,times,array, float}
\usepackage{amsmath,cases,graphicx}
\usepackage[mathscr]{eucal}

\def\psfancypar#1#2{\begingroup\def\par{\endgraf\endgroup\lineskiplimit=0pt}
               \setbox2=\hbox{\large\sc #2}
               \newdimen\tmpht \tmpht \ht2 \advance\tmpht by \baselineskip
               \font\hhuge=Times-Bold at \tmpht
               \setbox1=\hbox{{\hhuge #1}}
               \count7=\tmpht \count8=\ht1
               \divide\count8 by 1000 \divide\count7 by \count8 
               \tmpht=.001\tmpht\multiply\tmpht by \count7 
               \font\hhuge=Times-Bold at \tmpht
               \setbox1=\hbox{{\hhuge #1}}
               \noindent
                \hangindent1.05\wd1
               \hangafter=-2 {\hskip-\hangindent
               \lower1\ht1\hbox{\raise1.0\ht2\copy1}%
                \kern-0\wd1}\copy2\lineskiplimit=-1000pt}

\newcommand{\Psibf}{\mbox{${\bf \Psi}$}}

\newcommand{\E}{\mbox{{\rm E}}}

\newcommand{\abf}{\mbox{${\bf a}$}}

\def\boxit#1{\vbox{\hrule\hbox{\vrule\kern3pt
        \vbox{\kern3pt#1\kern3pt}\kern3pt\vrule}\hrule}}

\def\reals{ { {\rm  I \kern-0.15em R }  } }
\def\complex{ {\,{{\rm C} \kern-0.50em \raise0.20ex {  |}}\, }}

\def\mubf{\hbox{\boldmath$\mu$\unboldmath}}

\def\chibf{\hbox{\boldmath$\chi$\unboldmath}}

\def\Sigmabf{\hbox{$\bf \Sigma$}}

\def\Lambdabf{\mbox{$ \bf \Lambda $}}

\def\abf{{\bf a}}

\def\nbf{{\bf n}}

\def\sbf{{\bf s}}

\def\ubf{{\bf u}}
\def\vbf{{\bf v}}

\def\xbf{{\bf x}}
\def\ybf{{\bf y}}

\def\xbf{{\bf x}}
\def\ybf{{\bf y}}
\def\Abf{{\bf A}}
\def\Bbf{{\bf B}}
\def\Cbf{{\bf C}}
\def\Dbf{{\bf D}}
\def\Ebf{{\bf E}}

\def\Hbf{{\bf H}}
\def\Ibf{{\bf I}}

\def\Qbf{{\bf Q}}
\def\Rbf{{\bf R}}

\def\Ubf{{\bf U}}
\def\Vbf{{\bf V}}

\def\Xbf{{\bf X}}

\def\Zbf{{\bf Z}}

\def\Cc{{\cal C}}

\def\Hc{{\cal H}}

\def\Kc{{\cal K}}

\def\Nc{{\cal N}}

\def\be{\vskip .3cm \begin{equation}}
\def\ee{\end{equation} \vskip .4cm \noindent}
\newcommand{\R}{\mbox{$\hat {\bf R}_{N}$}}

\def\Rxx{\Rbf_{\ssstyle X\kern-.1em X}}

\let\ssstyle=\scriptscriptstyle

\def\Kout{\setbox1=\hbox{\Huge\bf K}\hbox to
1.05\wd1{\hspace{.05\wd1}
}}
\def\Sout{\setbox1=\hbox{\Huge\bf S}\hbox to 1.05\wd1{\hspace{.05\wd1}
}}

  \ifx\LabelFigloaded\MYundefined\relax
  \else
   \fi

  \def\LabelFigloaded{\relax}

  \chardef\LabelFigCatAt\the\catcode`\@
  \catcode`\@=11

 \let\LabelFigwlog@ld\wlog
 \def\wlog#1{\relax}

 \ifx\\\MYundefined@
    \let\\\relax
 \fi

  \def\ms@g{\immediate\write16}

 \def\N@wif{\csname newif\endcsname }
 \def\Temp@ {\N@wif\ifIN@}
 \ifx\INN@\MYundefined@
    \else \let\Temp@\relax
 \fi
 \Temp@

  \def\IN@{\expandafter\INN@\expandafter}
  \long\def\INN@0#1@#2@{\long\def\NI@##1#1##2##3\ENDNI@
    {\ifx\m@rker##2\IN@false\else\IN@true\fi}%
     \expandafter\NI@#2@@#1\m@rker\ENDNI@}
  \def\m@rker{\m@@rker}
 
  \newtoks\Initialtoks@  \newtoks\Terminaltoks@
  \def\SPLIT@{\expandafter\SPLITT@\expandafter}
  \def\SPLITT@0#1@#2@{\def\TTILPS@##1#1##2@{%
     \Initialtoks@{##1}\Terminaltoks@{##2}}\expandafter\TTILPS@#2@}

 \def\Shifted@@#1#2#3{\setbox0=\hbox{#3}%
   \raise -\dp0\vbox {\kern-#2%
       \hbox {\kern#1\unhbox0\kern-#1}%
           \kern#2}}

 \newcount\gridcount
 \newbox\auxGridbox@ \newbox\hGridbox@ \newbox\vGridbox@
 \newbox\Labelbox@ \newbox\auxLabelbox@
 \newbox\Coordinatebox@
 \newtoks\Labeltoks@
 \newdimen\Wdd@ \newdimen\Htt@
 \newdimen\Wddd@ \newdimen\Httt@
 
 \def\Wr@{\immediate\write16}

 \newdimen\GL@wd
 \def\GridLineWidth#1{\GL@wd=#1}

 \def\gobble#1{}
 \def\EdgeErr@{\Wr@{}%
      \Wr@{\string\Edges\space argument
      1, 10, 100 or 1000 please\string!}%
      }

 \newcount\Edgect@

 \def\Sweepup#1\endSweepup{}

 \def\SetEdges@{%
    \edef\Zr@@s{\expandafter\gobble\number\Edgect@\empty}%
        \count255=0\Zr@@s\relax
        \ifnum\count255=\z@\else\EdgeErr@\show\tailtest\fi
        \count255=1\Zr@@s\relax
        \ifnum\count255=\Edgect@\relax\else\EdgeErr@\show\leadtest\fi
    \EdgGl@b\edef\Zr@s{\expandafter\gobble\Zr@@s\empty}
    \ifnum\Edgect@>\@ne\relax\EdgGl@b\let\L@Dc\empty
        \else\EdgGl@b\edef\L@Dc{\string.}\fi
    \ifnum\Edgect@>\@ne\relax
        \EdgGl@b\edef\Edgescale@##1{\divide##1 by \Edgect@}%
        \else\EdgGl@b\edef\Edgescale@##1{}\fi
    }

 \def\Edges#1{\Edgect@=#1\relax
     \let\EdgGl@b\global \SetEdges@}

 \Edges{1}

 \def\hhrule{\hrule height \GL@wd\vskip-.\GL@wd}

 \def\hRule@{%
   \advance\gridcount -2%
   \vfil\hhrule\vfil
   \llap{\smash{\raise -2.5pt
     \hbox{\L@Dc\number\gridcount\Zr@s\kern2pt}}}%
   \hhrule
   }

\def\vvrule{\vrule width \GL@wd \kern-\GL@wd}

 \def\vRule@{\advance\gridcount 2%
   \hfil\vvrule\hfil
   \setbox\auxGridbox@=\vbox to 0pt
      {\vskip \Htt@\vskip 2pt
        \hbox to 0pt{\hss\L@Dc\number\gridcount\Zr@s\hss}\vss}%
      \wd\auxGridbox@=0pt \box\auxGridbox@
   \vvrule
   }

 \def\PlaceGrid@@{\gridcount=10 
  \setbox\hGridbox@=\hbox{%
        \hbox{%
             \hskip-.4pt\vrule
             \vbox to \Htt@{%
               \offinterlineskip\parindent=\z@\relax
               \hbox to \Wdd@{\hfil}
               \hRule@\hRule@\hRule@\hRule@
               \vfil\hhrule\vfil}%
             \vrule\hskip-.4pt}
    }%
  \gridcount=0%
  \setbox\vGridbox@=\hbox{%
      \vbox{\offinterlineskip\parindent=0pt\hsize=0pt
         \vskip-.4pt\hrule%
         \hbox to \Wdd@{%
                 \vtop to \Htt@{\vfil}%
                 \vRule@\vRule@\vRule@\vRule@
                 \hfil\vvrule\hfil}%
         \hrule\vskip-.4pt}}%
  \wd\hGridbox@=0pt\ht\hGridbox@=0pt
  \wd\vGridbox@=0pt\ht\vGridbox@=0pt
  \hbox{\box\hGridbox@\box\vGridbox@}%
  }

 \def\LabelsGlobal{\def\LabGl@b{\global}}
 \def\LabelsLocal{\def\LabGl@b{}}
 \LabelsGlobal 

 \def\SetLabels#1\endSetLabels{%
   \LabGl@b\Labeltoks@={#1()\\}%
   }

 \LabGl@b\Labeltoks@={()\\}

 \def\ShowGrid{\LabGl@b\let\PlaceGrid@\PlaceGrid@@}
 \def\HideGrid{\LabGl@b\let\PlaceGrid@\relax}
 \def\Grids{\ShowGrid\LabGl@b\let\GridSwitch@\ShowGrid}
 \def\noGrids{\HideGrid\LabGl@b\let\GridSwitch@\HideGrid}

 \noGrids

 \def\bAdjust@@{%
     \setbox\auxLabelbox@=\hbox{\raise \dp\auxLabelbox@
            \box\auxLabelbox@}}
 \def\bAdjust@{\let\vAdjust@\bAdjust@@}

 \def\eAdjust@@{\dimen0=-.5\ht\auxLabelbox@
     \advance\dimen0 by .5\dp\auxLabelbox@
     \setbox\auxLabelbox@=
            \hbox{\raise\dimen0\box\auxLabelbox@}}
 \def\eAdjust@{\let\vAdjust@\eAdjust@@}

 \def\tAdjust@@{%
     \setbox\auxLabelbox@=\hbox{\raise-\ht\auxLabelbox@
            \box\auxLabelbox@}}
 \def\tAdjust@{\let\vAdjust@\tAdjust@@}

 \let\vAdjust@\relax

 \def\lAdjust@{\let\hAdjust@\rlap}
 \def\rAdjust@{\let\hAdjust@\llap}

 \let\hAdjust@\relax\let\vAdjust@\relax

 \def\FetchLabel@#1(#2)#3\\{%
     \IN@0#2@@\ifIN@
        \setbox0=\hbox{\ignorespaces#1#3\unskip}%
        \ifdim\wd0>0pt
           \ms@g{}%
           \ms@g{ !!! Bad label(s)? !!!}%
           \message{ #1(#2)#3}%
        \fi
        \def\LabelMole@##1\endFetchLabel@{%
            \IN@0()\\@##1@%
            \ifIN@\def\Temp@{\FetchLabel@##1\endFetchLabel@}%
            \else\def\Temp@{}%
            \fi
            \Temp@
           }%
     \else
       \ignorespaces#1\unskip
       \setbox\auxLabelbox@=%
         \hbox to 0pt{\hss\ignorespaces\hAdjust@
          {\ignorespaces#3\unskip}\hss}%
       \vAdjust@
       \let\hAdjust@\relax\let\vAdjust@\relax
       \AugmentLabelBox@@{#2}%
       \ht\Labelbox@=0pt\dp\Labelbox@=0pt
       \let\LabelMole@\FetchLabel@%
     \fi\LabelMole@}

 \newtoks\XYSep@ 
 \def\SetXYSeparator#1{%
     \IN@0#1@@\ifIN@\XYSep@{*}%
     \else
     \XYSep@{#1}%
     \fi
     }

 \SetXYSeparator*

 \def\AugmentLabelBox@@#1{%
     \IN@0\the\XYSep@ @#1@\ifIN@
       \SPLIT@0\the\XYSep@ @#1@%
       \setbox\Labelbox@=\hbox to 0pt{%
         \unhbox\Labelbox@
         \Shifted@@{\the\Initialtoks@\Wddd@}%
         {\the\Terminaltoks@\Httt@}%
         {\box\auxLabelbox@}}%
     \else
         \ms@g{}%
         \ms@g{ !!! Bad insertion point. !!!}%
         \message{ (#1\ this point was rejected.)}%
     \fi
    }

 \def\FetchOption@#1[#2]#3\endFetchOption@{%
    \def\temp{#1}
    \ifx\temp\empty
       \Edgect@=#2\relax
       \let\EdgGl@b\relax
       \SetEdges@
       \Cleaner@#3%
    \fi}

 \def\Cleaner@#1[@]{\Labeltoks@{#1}}
     
 \def\PlaceLabels@@{\mathsurround=0pt
     \def\Cr@{\\}%
     \let\L\lAdjust@\let\R\rAdjust@
     \let\B\bAdjust@\let\E\eAdjust@\let\T\tAdjust@
     \expandafter\FetchOption@\the\Labeltoks@[@]\endFetchOption@
     \Wddd@=\Wdd@ \Edgescale@\Wddd@ 
     \Httt@=\Htt@ \Edgescale@\Httt@
     \expandafter\FetchLabel@\the\Labeltoks@\endFetchLabel@
     \box\Labelbox@
     }%

 \let \PlaceLabels@\PlaceLabels@@

 \def\AffixLabels#1{\setbox\Coordinatebox@=\hbox{#1}%
      \Wdd@=\wd\Coordinatebox@ \Htt@=\ht\Coordinatebox@
      \advance\Htt@ \dp\Coordinatebox@
      \hbox{\copy\Coordinatebox@\kern-\Wdd@ 
           \Shifted@@{0pt}{-\dp\Coordinatebox@}%
           {\PlaceLabels@\PlaceGrid@}%
           \kern\Wdd@}%
      \GridSwitch@ 
      \LabGl@b\Labeltoks@{()\\}%
      }
 
   \let\wlog\LabelFigwlog@ld   
   \catcode`\@=\LabelFigCatAt  

                                By

              Raymond S\'eroul <A18645@FRCCSC21.BITNET>
                                and 
              Laurent Siebenmann <lcs@topo.math.u-psud.fr>
    
              VERSIONS: July 1991, Oct 1991, Jan 1992, July 1992

INTRODUCTION

      This labelling package is intended for TeX users who
rely on non-TeX sources for for their graphics inserts.  It
provides means for adding TeX labels to such inserts with a
minimum of fuss. 

       For most labels, TeX users have in the past found it
reasonably convenient to rely on non-TeX sources. Typical
occasions when an inescapable need for TeX labels seemed to
arise are

 (a) when the graphics program lacks certain exotic or complex
mathematical symbols

 (b) when the very highest typographical quality is wanted for the
labels

 (c) when labels included with the graphics fail to print, 
 and you cannot figure out why (cf. boxedeps.doc).  The labels
 provided by labelfig.tex are 100

       Since this package first appeared, many users, who in the
past scarcely dreamed of using TeX labels, have come to use
nothing but.  So it is now appropriate to add

Intoxication Warning:  TeX labels may be addictive and expensive. 

     If you have a fast preview you may disagree, and even find
that this package provides an agreeable paste-up environment; see
extra applications at end.

     Note to publishers: It is possible and convenient to ultimately
export the TeX labels produced by labelfig.tex to become an integral
part of the EPS file. This is often desired by a publisher who typically
uses an "upmarket" graphics or page layout program, with which the
staff is skilled in perfecting figures.  See Appendix I for
a recipe.

     The authors are grateful to Patrick Ion of Math Reviews for
helpful comments and encouragement.

BASIC INSTRUCTIONS

    After reading in the macro file using

preview or proof your figure with a coordinate grid printed on
top, by typing the following:

    \ShowGrid  
    \AffixLabels{<the graphics insertion>}

Here <the graphics insertion> is what you would type to insert
the graphics object alone without the grid.  This must provide
for the space around it. For example <the graphics insertion>
might well be \BoxedEPSF{MyFigure scaled 700} using the
boxedeps.tex macro package (from same source); this provides a
TeX box containing the encapsulated PostScript insert specified by
the file MyFigure. \AffixLabels{...} provides the grid (supposing
\ShowGrid is present) and later, once you have specified labels
using the grid, it will "tack on" the labels.

     The grid is a sort of (usually elongated) checkerboard of
ten rows and ten columns and its (internal) partitions are by
default numbered  .1, ... ,.9  both horizontally (X-coordinate
running left to right) and vertically (Y-coordinate running bottom
to top).  Thus the points enclosed by the grid correspond to the
points of the unit square in the cartesian "X-Y" plane, the lower
left corner corresponding to the origin (0,0).  By extrapolation,
the full page corresponds to a larger rectangle in the plane.

     These coordinates serve to position labels as follows.
Before the \AffixLabels{...} command type label specifications:

  \SetLabels
   (<X-coordinate>*<Y-coordinate>) <first label> \\
   .
   .
   .
   (<X-coordinate>*<Y-coordinate>)  <last label> \\
  \endSetLabels

Each row specifies one label and is terminated by \\.  In each
row, the position indicator comes first; it is written as a
standard cartesian point except that the X- and Y- coordinates
are separated by * rather than a comma because TeX allows a
comma as decimal point. There are no dimension units to specify
as the unit is the grid itself.

     By default, this cartesian point specifies where the middle
of the baseline of the label will be located.  However if you precede
the point by \L [or \R] the left [or right] edge of the baseline will
be located there. Similarly you may also precede the point by \T, \E,
or \B to vertically align the top equator or bottom of the label box
at the specified point.  This gives nine standard positions of
the label with respect to the insertion point --- corresponding to
the eight principle points of the compas and the center

                     \L\T     \T      \R\T

                     \L\E     \E      \R\E

                     \L\B     \B      \R\B

But this neglects the default "baseline" level of TeX,
giving potentially three more positions

                     \L    <no tag>   \R

For text, the baseline level is often the preferred. Its relation to
the others is variable. It will often coincide with the bottom level,
as happens for "X".  But it is often distinct, as for "g", in which
case you have in all 12 distinct positions rather than 9.

     It is convenient to think of this specification of label
position as attaching the label by a thumb-tack to the coordinate
grid. There are up to twelve positions of the thumb-tack on the
label, while the position of the thumb-tack on the coordinate grid is
arbitrary.  Normally, one choses the position of the thumb-tack on
the label to be the one that is the closest to the item being
labeled.  There are good reasons for this "rule of thumb":

   (a)  It facilitates correct positioning at first try.

   (b)  If the scale of the figure must be altered after labels
have been affixed, the labels have a good chance of remaining well
positioned.

   (c)  The visible grid need not extend beyond the "bounding box"
for the figure, because the best preferred position is always
(at least almost) within the bounding box .

The second reason is particularly important. Indeed it often
happens that scale has to be altered after labelling begins, in
order to either provide space for the labels, or to adjust
proportions between the labels and the figure.  (The size of labels
is unaffected by scaling.)

     Here is an artificial but self-contained test which uses
TeX rules to make a graphics object.

TEST

    Do not skip this!

 \def\FrameIt#1{\hbox{\vrule$\vcenter {\hrule\kern3pt%
             \hbox {\kern3pt #1\kern3pt}%
               \kern3pt\hrule}$\relax\vrule}}

 \def\Caption#1#2{\FrameIt{%
       \vtop {\hsize=#1\relax \parindent=0pt
         \leftskip=0pt \rightskip=0pt plus15pt
         \parfillskip=0pt
         \lineskip=1pt\baselineskip=0pt
         #2}}}

 \def\FirstQuadrant{\hbox to 100pt{\vrule\vbox to 100pt{%
        \hbox to 100pt{\hfil}\vfil\hrule}\hss}}

  \SetLabels
    \R(.5*.2) $\zeta\,\cdot$\\
    (.9*-.10) $\xi$\\
    \R(-.03*.9) $\eta$\\
    \T(.5*.9) \Caption{70pt}{%
          \it The norm of
          $g(\xi+i\eta)$ is indicated on
          contours of this invisible surface.}\\
  \endSetLabels

  \AffixLabels{\FirstQuadrant}

  \end

  Note that the coordinates to use for labels are indicated on the
edges of the grid (when visible) corresponding to the conventional
x- and y- axes of the Cartesian plane. By default the grid is
1-by-1. However, by the command \Edges{100}, you can change this
to 100-by-100 and many users find this alternative most
convenient. Place the command \Edges{...} in your style file (or
header) since its effect is is global. Other possible edge values
are 10 and 1000.

  If you use the command \Edges{...} at all, do so with care.  For
if you accidentally delete an \Edges{...} command your labels will
abruptly be badly misplaced and may logically but mysteriously
generate "dimension too big" errors under TeX and "off page" errors
under your driver.  

  You can dictate the edgescale for an individual figure by giving
the scale in brackets immediately after \SetLabels.  Thus, to
import into an article using say \Edge{100} a figure labelled using
another edgescale, say the original 1-by-1 default, you can use
\SetLabels[1]...\endSetLabels.

GETTING IT DOWN PAT

     Complicated labeling deserves the same respect as
complicated mathematics.  Do not expect it to come out perfect the
first time!  What is needed in either case is a mechanism to
repeatedly typeset troublesome pieces.

     One mechanism is always available.  One does complicated
labelling in a separate "test" file involving just the figure being
labelled;  a texpert will know how to \dump TeX's current state as
a temporary format that restarts rapidly at each retry.  Usually,
one then pastes the completed labelled figure back into the main
TeX file, but, of course, one can also \input it as an auxiliary
file.

     If you do not have a TeXpert at handy, here is a first
approximation to an efficient setup. By deletions reduce a copy
of your article to just a few lines before and after the figure.
Now label the figure, and finally, copy and paste the labelled
figure to the original article. Then copy the next figure to label
into this testbed and repeat. The TeXpert can improve the  speed
at which TeX starts up, by compiling a format specifically for
your article; just one caution: best NOT include in the format
ephemeral details of setup like \Set<mydriver>ArtSpecials (from
boxedeps.tex because this reads  figure dimensions which you may
change during your work session.

     An improved mechanism to repeatedly typeset troublesome
pieces is now available on the Macintosh; it is called LinoTeX;
see the same ftp sources.  It could be set up on many types
of computer.

     Before using labelfig.tex to attach labels to a graphics
object inserted using boxedeps.tex or BoxedArt.tex, make it a
firm rule to carefully adjust the bounding box using the trimming
commands of these packages, and also at least tentatively scale
and position the object. Beware of changing the grid inadvertently
after the labels have been positioned.  For example, correcting
the bounding box of a PostScript graphics object can foul up the
labels by changing the coordinate grid to which the labels are
attached. This is particularly true for the trimming  commands of
boxedeps.tex and BoxedArt.tex. However, as noted already, change
of scale is much less disruptive, and modest adjustments should be
well tolerated.

     Sometimes the labels protrude so far from the bounding box
of a figure that the figure has to be repositioned.  Best do this
by ad hoc spacing, say using \hglue and \vglue; altering the
bounding box would create a vicious circle.

     Remember that you are responsible for preventing labels
from overlapping. You are responsible for all label typography
including size and style. A label is really just about anything
that can be put in a TeX box. Note that spaces at the beginning
and end of labels will normally be suppressed; if you really want
them you must protect them with TeX braces.

     This package temporarily sets the \mathsurround parameter
of TeX to zero  while the labels are being affixed. This is done
because nonzero \mathsurround space would influence the position
of left and right aligned labels; then, when a texpert or printer
modifies mathsurround, diagram labeling might be disastrously
altered. There is a small price to pay involving labels that are
formatted as caption boxes including mathematics: you  may want or
need to specify an explicit mathsurround space within the caption
box; it will not influence anything outside.

     Those hostile to the use of * as separator between
the X and Y coordinates of label insertion points, are free to
impose another using \SetXYSeparator{<the new separator>}.  
Americans may prefer "," to "*" since they never use a 
comma as a decimal point; on the other hand, * may be more visible.

APPENDIX (I)  MERGING labelfig.tex LABELS INTO AN EPSF GRAPHICS OBJECT.

     As promised in the introduction, here is a recipe useful for
publishers. It works at least on Macintosh and at least for vectorized
graphics and Adobe type1 fonts.  (There is surely a similar recipe for
PCs under MSWindows.)

 (a)  Use boxedeps.tex utility to integrate the figure given by the eps
file, "x.eps" say, with a visible frame around it.  See
\ShowDisplacementBoxes command in boxedeps.tex.  To get precise results
automatically it is important to use the \Trim... commands of
boxedeps.tex making the "DisplacementBox" neatly fit the figure.

 (b)  Use the TeX printer driver and LaserWriter (versions >= 8.1.1) to
export to an EPSF the DVI page containing the integrated, labelled
figure. You now have an EPS file  "xx.eps"  that contains too much, and at
the wrong scale, and at wrong position.

 (c)  Convert the EPSF to an Adode Illustrator format EPSF using
the shareware utility called epsConvert by Sam Weiss
1993-- (currently $25).

 (d)  In Illustrator (or a compatible program), group the labels and the
"DisplacementBox"; copy them to the clipboard and paste them into "x.ps".
This step requires that all the label fonts be "visible to the Macintosh.

 (e)  Translate and scale the pasted group consisting of the labels plus
the "DisplacementBox" so as to make the "DisplacementBox" the bounding
box of (labelless) figure represented by "x.eps".  At this point the
labels will be correctly placed on the figure "x.eps".

 (f)  Ungroup and delete the "DisplacementBox".  The result is the
desired single EPS file, "x+.eps" say, It contains the original figure
plus its labels.  

     Using grouping and ungrouping appropriately in "x+.eps", a
publisher's staff can very efficiently improve label positions etc.

APPENDIX II)  SOME EXOTIC APPLICATIONS

     The grid of labelfig.tex is analogous to a light-table in
classical page makeup with wax or latex glue.  In principle, you
can use it to compose any page from its indivisible parts.  This
even has some of the artisanal charm of classical paste-up
provided you have a fast screen preview to make the process
"interactive".

     In practice labelfig.tex is a tool for nonstandard jobs.
Here are a few going beyond the labelling already discussed.

(I)  GRAPHICS INTEGRATION.

     This is accomplished by treating the imported graphics
objects as labels.  The underlying graphics object is then
typically an empty  \vbox to <dimension>{\vfill} in a TeX
\midinsert...\endinsert construction.  A label line
might be of the form

   (.1*.1) \\

The exact form of the special command varies from driver to
driver.  However, in the case of encapsulated PostScript graphics
(EPSF norm), by relying on boxedeps.tex, one can have the
following standard syntax (independant of driver  (see
boxedeps.doc for details.
  
  (.1*.1) \BoxedEPSF{MyFigure scaled <scale in mils>}\\

This may be slow since it requires TeX to read the PostScript
file to read bounding box using many complex macros.  So you
may want to try

  (.1*.1) \EPSFSpecial{MyFigure}{<scale in mils>}\\

which is fast and driver independant, but it squashes the
bounding box, normally to its lower left corner.

     Similarly for graphics of the Macintosh PICT norm ---
using BoxedArt.tex (same sources) in place of boxedeps.tex.

     This approach to integration is to be recommended when
one is assembling a composite graphics object.

 (II)  COMMUTATIVE DIAGRAM ENHANCEMENT

     Commutative diagrams or arrays of mathematical objects
connected by arrows of various sorts are common in mathematics.
The mathematical objects require the use of TeX.  Recently TeX
acquired a good collection of arrows of all slopes --- that of
LamSTeX --- plus pwerful macros to build the diagrams.

     However, even the LamSTeX collection is often
inadequate; it lacks for example double shafted arrows, dotted
arrows and curved arrows. Fortunately it is possible to produce
such arrows on an individual basis using sophisticated graphics
programs such as Illustrator and AldusFreehand (both serving
the EPSF norm) or using Metafont (with its public domain norm).
Since the creation of each new arrow is a work of love, you
probably want to limit the number of arrows by using LamSTeX
for most arrows. The 40K commutative diagram module of LamSTeX
has been adapted to work with AmSTeX and a copy may be posted
with LabelFig and related files. Unfortunately no one has yet
offered a version that works with Plain TeX or LaTeX.

       Suffice it here to say that when the exotic arrow has
been somehow imported into TeX, labelfig.tex treats it as a
label that one affixes to the commutative diagram.  Two other
steps will be treated in separate notes, namely the matter of
extracting the dimension specifications for the arrow and the
construction of the arrow --- for these steps are far from
unique and often depend intimately on your computer environment. 
Notes for the Macintosh-Textures-Illustrator combination are
found in the file ExoticArrows.doc.

 (III) NESTING 

Ingenuity pays off in exploiting labelfig.tex. One can
mix graphics and typography quite freely.  labelfig.tex is good
for freeform or overlapping arrangements, while boxedeps.tex (or
BoxedArt.tex) is best for regimented non-overlapping
arrangements --- and the two can be combined.

     The default behavior of labelfig.tex is not ideal 
for nesting objects, because to prevent trouble for beginners
the register for labels is globally cleared when \AffixLabels
concludes.  But there are switches available

      \LabelsGlobal      \LabelsLocal

which change this.  To understand this, extend the above test 
by something like:

 \LabelsLocal

 \SetLabels
    (.5*.5) AAA\\
 \endSetLabels

 {
 \SetLabels
    (.5*.5) ZZZ\\
 \endSetLabels
   \AffixLabels{\FirstQuadrant}
 }

   \AffixLabels{\FirstQuadrant}

     There are however potential pitfalls.  Neither
labelfig.tex nor boxedeps.tex has been tested under extreme
conditions. Problems may occur if their procedures are
indiscriminately nested. For boxedeps.tex (not labelfig.tex)
there is a precise cause for worry, namely many of its
variables are "global", which means that TeX braces will not
provide the protection one might expect.

COMMAND SUMMARY FOR labelfig.tex

  Here [...] means optional (one or zero)
       [...]* means any number of such constructs

  \SetLabels
    [[<P>](<X><Sep><Y>) <label> \\]*
  \endSetLabels
  \ShowGrid  
  \AffixLabels{<the figure>}

   --- <P> is tack position, one of eleven or empty
              order irrelevant

                   \L\T      \T      \R\T

                   \L\E      \E      \R\E

                     \L               \R

                   \L\B      \B      \R\B

   --- (<X><Sep><Y>) insertion point;
  <Sep> is separator, = * by default;
  \SetXYSeparator{<Sep>} changes it.
   <X> and <Y> are real numbers

  --- <label> a label to attach 

  --- <the figure> the figure to label 

  \GlobalLabels (default)     
  \LocalLabels  setting for nested constructs.

 \Grids makes ALL grids appear; \HideGrid then makes just next disappear.
 \noGrids returns to default.  The commands are always global.

 \GridLineWidth{<dimension>} adjusts width of grid lines. Default is very
small, to give "hairline" effect. If your grid lines are missing try
setting \GridLineWidth{1pt}.

 \Edges#1 globally changes the edge size of all grids to the numerical 
value #1, which must be 1, 10, 100, or 1000.  The default is 1.

VERSION HISTORY.
 --- Jan 1993: \Edges#1 and [??] option after \SetLabels
 --- July 1992: \Grids, \noGrids, \HideGrid;
       Gridlines become hairlines; \GridLineWidth{<dimension>}.
 --- Oct 1991, Jan 1992: \SetXYSeparator{<Sep>},  \LabelsGlobal,
       \LabelsLocal.
 --- July 1991: first release

Address for bugs and other feedback:

        Raymond S\'eroul
        IREM and Lab. de Typographie Informatise
        Univ. Rene Descartes
        Strasbourg

    Tel 33-88-41-63-45
    Email:  A18645@FRCCSC21.BITNET

        Laurent Siebenmann
        Mathematique, Bat. 425,
        Univ de Paris-Sud,
        91405-Orsay,
        France

    Tel 33-1-6941-7949; 
    Email: lcs@topo.math.u-psud.fr

\def\scalefig#1{\epsfxsize #1\textwidth}

\newcommand {\Ebb}{{\mathbb{E}}}

\newtheorem{theorem}{Theorem}
\newtheorem{lemma}{Lemma}

\newtheorem{corollary}{Corollary}

\newtheorem{remark}{Remark}

\title{{\LARGE Outage Probability and Outage-Based Robust Beamforming  for
MIMO Interference Channels with Imperfect Channel State
Information}}

\author{
Juho Park, {\em Student~Member, IEEE}, Youngchul
Sung$^\dagger$\thanks{$^\dagger$Corresponding author}, {\em
Senior~Member, IEEE}, Donggun Kim, {\em Student~Member, IEEE},
and H. V. Poor {\em Fellow, IEEE}
\thanks{J. Park, Y. Sung and D. Kim are with the Dept. of Electrical Engineering, KAIST, Daejeon 305-701, South Korea.
E-mail:\{jhp@, ysung@ee. and dg.kim@\}kaist.ac.kr and H. V. Poor
is with Dept. of Electrical Engineering, Princeton University,
Princeton, NJ 08544, E-mail: poor@princeton.edu This work was
supported by the Korea Research Foundation Grant funded by the
Korean Government (KRF-2008-220-D00079). A preliminary version of
this work is presented at IEEE Globecom 2011.
\cite{Park&Kim&Sung:11Globecom}.}
}

\begin{document}

\maketitle

\begin{abstract}
 In this paper, the outage probability and outage-based beam design for
 multiple-input multiple-output (MIMO) interference channels are considered.
First, closed-form expressions for the outage probability in MIMO
interference channels are derived under the assumption of
Gaussian-distributed  channel state information (CSI) error, and
the asymptotic behavior of the outage probability as a function of
several system parameters is examined by using the Chernoff bound.
It is shown that the outage probability decreases exponentially
with respect to the quality of CSI measured by the inverse of the
mean square error of CSI. Second, based on the derived outage
probability expressions, an iterative beam design algorithm for
maximizing the sum outage rate is proposed. Numerical results show
that the proposed beam design algorithm yields
better sum outage rate performance than conventional algorithms
such as interference alignment developed under the
assumption of perfect CSI.
\end{abstract}


\section{Introduction}

Due to their importance in current and future wireless
communication systems, multiple-input multiple-output (MIMO)
interference channels have gained much attention from the research
community in recent years. Since Cadambe and Jafar showed that
interference alignment (IA) achieved the maximum number of degrees
of freedom  in MIMO interference channels
\cite{Cadambe&Jafar:08IT}, there has been extensive research in
devising good beam design algorithms for MIMO interference
channels. Now, there are many available beam design algorithms for
MIMO interference channels such as IA-based algorithms
\cite{Gomadam&Jafar:11IT, Peters&Heath:11VT, Yu&Sung:10SP} and
sum-rate targeted algorithms
\cite{Gomadam&Jafar:11IT,Peters&Heath:11VT,
Santamaria&Gonzalez&Heath&Peter:10Globecom,Schmidt&Honig:09ASILOMAR,
Negro&Slock:11PIMRC,Shin&Moon:11Globecom}. However, most of
these algorithms assume perfect channel state information (CSI) at
transmitters and receivers, whereas the assumption of perfect CSI
is unrealistic in practical wireless communication systems since
perfect CSI is unavailable in practical wireless communication
systems due to channel estimation error, limited feedback or other
limitations \cite{Biglieri:book}. Thus, the CSI error should be
incorporated into the beam design to yield better performance, and
this is typically done under robust beam design frameworks.

There are many robust beam design studies in the conventional
single-user MIMO case  and also in the multiple-input and
single-output (MISO) multi-user case. In the MISO multi-user case,
the problem is more tractable than in the MIMO multi-user case, and
extensive research results are available on MISO broadcast and
interference channels with imperfect CSI
\cite{Lindblom&Karipidis&Larsson:11Arxiv,
Lindblom&Larsson&Jorswieck:10WCOM,Li&Chang&Lin&Chi:11Arxiv};
 the outage rate region is defined for MISO interference
channels in \cite{Lindblom&Karipidis&Larsson:11Arxiv}, and the
optimal beam structure that achieves a Pareto-optimal point of the
outage rate region is given in
\cite{Lindblom&Larsson&Jorswieck:10WCOM}.  For more complicated
MIMO interference channels, there are several pioneering works on
robust beam design under CSI uncertainty
\cite{Chiu&Lau&Huang&Wu&Liu:10WCOM, Shen&Li&Tao&Wang:10WCOM,
Ghosh&Rao&Zeidler:10SP}. In \cite{Chiu&Lau&Huang&Wu&Liu:10WCOM},
the authors solved  the problem based on a worst-case approach.
In their work, the CSI error is modelled as a random variable
under a Frobenius norm constraint, and a semi-definite relaxation
method is used to obtain the  beam vectors that maximize the
minimum signal-to-interference-plus-noise ratio (SINR) over all
users and all possible CSI error. In
\cite{Shen&Li&Tao&Wang:10WCOM}, on the other hand, the CSI error
is modelled as an independent Gaussian random variable, and the
beam is designed to minimize the mean square error (MSE) between the transmitted signal and the reconstructed signal at the receiver  with given
imperfect CSI at the transmitter (CSIT).

In this paper, we consider the robust beam design in MIMO
interference channels based on a different criterion. Here, we
consider the rate outage due to channel uncertainty and the
problem of sum rate maximization under an outage constraint in
MIMO interference channels. This formulation is practically
meaningful since an outage probability is assigned to each user
and the supportable rate with the given outage probability is
maximized in practical systems. Here, we assume that the
transmitters and receivers have imperfect CSI and the CSI error is
circularly-symmetric complex Gaussian distributed. Under this
assumption, we first derive closed-form expressions for the outage
probability in MIMO interference channels for an arbitrarily
given set of transmit and receive beamforming vectors, and then
derive the asymptotic behavior of the outage probability as a
function of several system parameters by using the Chernoff bound.
It is shown that {\em the outage probability decreases
exponentially with respect to (w.r.t.) the quality of CSI measured
by the inverse of the MSE of CSI, typically called the channel $K$
factor \cite{Biglieri:book} or interpreted as the Fisher
information} \cite{Poor:book} in statistical estimation theory.
In particular, it is shown that in the case of interference
alignment, the outage probability can be made arbitrarily small by
improving the CSI quality if the target rate is strictly less than
the rate obtained by using the estimated as the nominal channel.
Next, based on the derived outage probability expressions, we
propose an iterative beam design algorithm for maximizing the
weighted sum rate under the constraint that the outage probability
for each user is less than a certain level. Numerical results show
that the proposed beam design algorithm yields better sum outage
rate performance than conventional beam design algorithms
such as the `max-SINR' algorithm \cite{Gomadam&Jafar:11IT}
developed without the consideration of channel uncertainty.

\subsection{Related work}

The outage analysis for MIMO interference channels has been
performed by several other
researchers\cite{Ghosh&Rao&Zeidler:10SP,Makouei&Andrews&Heath:11SP}.
In \cite{Ghosh&Rao&Zeidler:10SP}, the outage probability for a
given rate tuple is computed under the assumption that the
knowledge of the channel mean and covariance matrix are available,
and transmit and receive beam vectors that minimize the power
consumption for a given outage constraint are obtained. However, it
is difficult to generalize this method of analysis to the case of
multiple data streams per user, whereas our analysis includes the
multiple data stream case. In \cite{Makouei&Andrews&Heath:11SP},
the outage probability and SINR distribution of each user in MIMO
interference channels with the knowledge of channel distribution
information are obtained under a particular transmit and receive
beam structure of IA transmit beams and zero-forcing (ZF)
receivers. On the other hand, our analysis can be applied to the
case of general transmit and receive beam structures beyond IA and
ZF.

The probability distribution of a quadratic form of Gaussian
random variables has been studied extensively in the statistics
field \cite{Gurland:55AMS, Kotz:67AMS-1, Kotz:67AMS-2,
Raphaeli:96IT} and in the communications area \cite{Nabar:05WC,
Hasna&Simon:01VTC, AlNaffouri&Hassibi:09ISIT}. The most
widely-used approach to obtain the probability distribution of a
Gaussian quadratic form is the series fitting method
\cite{Kotz:67AMS-1,Kotz:67AMS-2,Mathai&Provost:book,Nabar:05WC},
which typically converges to the probability distribution of a
Gaussian quadratic form from the lower tail first. However,  the
outage definition associated with robust beam design for MIMO
interference channels in this paper requires accurate computation
of upper tail probabilities. The series expansion for the
cumulative distribution function (CDF) obtained in this paper
based on the integral form for the CDF in
\cite{AlNaffouri&Hassibi:09ISIT} and the residue theorem
\cite{Raphaeli:96IT} is well suited to this purpose and converges
to the upper tail first. Thus, the obtained series in this paper
is more relevant for our outage analysis. For a detailed
explanation of the derived series, please see Appendices \ref{append:distribution}--\ref{append:computational}.

\subsection{Notation and organization}

We will make use of
standard notational conventions. Vectors and matrices are written
in boldface with matrices in capitals. All vectors are column
vectors. For a matrix $\Abf$, $\Abf^H$, $\|\Abf\|_F$ and $\Abf(i,j)$
indicate the Hermitian transpose, the Frobenius norm and the
element in row $i$ and column $j$ of $\Abf$, respectively, and
$\mbox{vec}(\Abf)$ and $\mbox{tr}(\Abf)$ denote the vector
composed of the columns of $\Abf$ and the trace of $\Abf$,
respectively. For vector $\abf$, $\| \abf\|$ and $[\abf]_i$
represent the 2-norm and the $i$-th element of $\abf$,
respectively. $\Ibf_n$ stands for the identity matrix of size $n$
(the subscript is included only when necessary), and
$\mbox{diag}(d_1,\cdots,d_n)$ means a diagonal matrix with
diagonal elements $d_1,\cdots,d_n$.
$\xbf\sim\Cc\Nc(\mubf,\Sigmabf)$ means that the random vector $\xbf$
has the circularly-symmetric complex Gaussian distribution with mean
vector $\mubf$ and covariance matrix $\Sigmabf$.
${\mathcal{K}}=\{1,2,\cdots,K\}$, $\iota = \sqrt{-1}$, and $|A|$
denotes the cardinality of the set $A$.

The paper is organized as follows. The system model and problem
formulation are described in Section \ref{sec:systemmodel}. In
Section \ref{sec:outage_prob}, closed-form expressions for the
outage probability are derived, and the behavior of the outage
probability as a function of several system parameters is examined
by using the Chernoff bound. In Section \ref{sec:beam_design}, an
outage-based beam design algorithm is proposed. Numerical results
are provided in Section \ref{sec:numerical}, followed by the
conclusion in Section \ref{sec:conclusion}.

\section{System Model and Problem Formulation}
\label{sec:systemmodel}

In this paper, we consider a $K$-user time-invariant MIMO
interference channel in which each transmitter equipped with $N_t$
antennas is paired with a receiver equipped with $N_r$ antennas,
and interferes with all receivers other than the desired
receiver. We assume that transmitter $k$ transmits $d\ (\le
\min(N_t,N_r))$ independent data streams to receiver $k$ paired
with transmitter $k$. Then, the received signal at receiver $k$ is
given by
\begin{equation}\label{eq:received_signal}
\ybf_k = \Hbf_{kk}\Vbf_k\sbf_k + \sum_{i=1, i\neq
k}^K\Hbf_{ki}\Vbf_i\sbf_i+\nbf_k,
\end{equation}
where $\Hbf_{ki}$ is the $N_r\times N_t$ channel matrix from
transmitter $i$ to receiver $k$;
$\Vbf_i=[\vbf_i^{(1)},\cdots,\vbf_i^{(d)}]$ is the $N_t\times d$
transmit beamforming matrix with normalized column vectors at
transmitter $i$, i.e., $||\vbf_i^{(m)}||=1$ for $m=1,\cdots, d$;
and $\sbf_i=[s_i^{(1)},\cdots,s_i^{(d)}]^T$ is the $d\times 1 $
symbol vector at transmitter $i$.  We assume that the transmit
symbol vector $\sbf_i$ is drawn from the zero-mean Gaussian
distribution with unit variance, i.e., $\sbf_i \sim
\Cc\Nc({\mathbf{0}},\Ibf)$, and the additive noise vector $\nbf_k$
is  zero-mean Gaussian distributed with variance $\sigma^2$, i.e.,
$\nbf_k\sim \Cc\Nc({\mathbf{0}},\sigma^2\Ibf)$. We assume that the
CSI available to the system is not perfect. That is, neither
the transmitters nor the receivers have perfect CSI. For the
imperfect CSI, we adopt the following model
\begin{equation} \label{eq:CSIerrormodel1}
\Hbf_{ki} = \hat{\Hbf}_{ki} + \Ebf_{ki}
\end{equation}
for each $(k,i) \in \Kc \times \Kc$, where $\Hbf_{ki}$ is the
unknown true channel, $\hat{\Hbf}_{ki}$ is the channel state
available to the transmitters and the receivers, and
$\Ebf_{ki}$ is the error between the true and available channel
information. For the CSI error $\Ebf_{ki}$  between the true and
available channel information, we adopt the Kronecker error model
which is widely used for MIMO systems to model the error
correlation that may be caused by the transmit and receive antenna
structure \cite{Biglieri:book}. Under this model,  the CSI error
$\Ebf_{ki}$ is given by
\begin{equation}  \label{eq:CSIerrormodel2}
\Ebf_{ki}=\Sigmabf_r^{1/2} {\Hbf}_{ki}^{(w)}\Sigmabf_t^{1/2},
~~\mbox{with}~~ \mbox{vec}( \Hbf_{ki}^{(w)})\sim
\Cc\Nc(0,\sigma_h^2\Ibf) ~~\mbox{for some}~~ \sigma_h^2 \ge 0,
\end{equation}
where $\Sigmabf_t$ and $\Sigmabf_r$ are transmit and receive
antenna correlation matrices, respectively, and the elements of
$\Hbf_{ki}^{(w)}$ are independent and identically
distributed (i.i.d.) and are drawn from a circularly-symmetric zero-mean complex Gaussian distribution.  The CSI uncertainty matrix $\Ebf_{ki}$ is
a circularly-symmetric\footnote{The circular symmetry of a random
matrix in form of $\Abf\Zbf\Bbf$ with constant matrices $\Abf$ and
$\Bbf$ and a circularly-symmetric complex Gaussian matrix $\Zbf$
can easily be shown by a similar technique to that used in the Appendix \ref{append:proof_covariance}.}
complex Gaussian random matrix with distribution
$\mbox{vec}(\Ebf_{ki})\sim \Cc\Nc ({\mathbf{0}},
\sigma_h^2(\Sigmabf_t^T \otimes \Sigmabf_r
))$\cite[p.90]{Biglieri:book}, and $\sigma_h^2$ is the parameter
capturing the uncertainty level in CSI.  We assume that
the $\Ebf_{ki}$'s are independent across transmitter-receiver pairs
$(k,i)$.  To specify the quality of CSI and signal reception,  we
define two parameters
\[
K_{ch}^{(ki)}:=\frac{\|\hat{\Hbf}_{ki}\|_F^2}{
\Ebb\{\|\Ebf_{ki}\|_F^2\}}=
\frac{\|\hat{\Hbf}_{ki}\|_F^2}{\sigma_h^2\mbox{tr}(\Sigmabf_t^T\otimes\Sigmabf_r)}
~~\mbox{and}~~
 \Gamma^{(k)}:=\frac{\|\hat{\Hbf}_{kk}\|_F^2}{\sigma^2}.
\]
 $K_{ch}^{(ki)}$ is the channel $K$ factor defined as the ratio of the power of the known channel part to that of the unknown channel part,  representing the quality of
CSI  \cite{Biglieri:book}, and $\Gamma^{(k)}$ is the
signal-to-noise ratio (SNR) at receiver $k$ since $\Vbf_k$ and
$\sbf_k$ are normalized in our formulation. Hereafter, we will use
$\hat{\mathcal{H}}$ to represent the collection of channel
information $\{\hat{\Hbf}_{ki}, \Sigmabf_t, \Sigmabf_r\}$ known to
the transmitters and receivers. By using the receiver filter
$\ubf_k^{(m)}$ ($||\ubf_k^{(m)}||=1$), receiver $k$ projects the
received signal $\ybf_k$ in \eqref{eq:received_signal}  to recover
the desired signal stream $m$:
\[
\hat{s}_k^{(m)} = (\ubf_k^{(m)})^H\ybf_k =
(\ubf_k^{(m)})^H\left((\hat{\Hbf}_{kk}+\Ebf_{kk})\Vbf_k\sbf_k +
\sum_{i=1, i\neq k}^K(\hat{\Hbf}_{ki} + \Ebf_{ki})
\Vbf_i\sbf_i+\nbf_k\right).
\]

We assume that the design of the transmit beamforming matrices
$\{\Vbf_k, k \in \Kc\}$ and receive filters
$\{\Ubf_k=[\ubf_k^{(1)},\cdots,\ubf_k^{(d)}], k\in \Kc\}$ is based
on the available CSI $\hat{\Hc}$. This model of beam design and
signal transmission and reception  captures many coherent linear
beamforming MIMO schemes including interference alignment and sum
rate maximizing beamforming schemes \cite{Gomadam&Jafar:11IT,
Santamaria&Gonzalez&Heath&Peter:10Globecom,
Sung&Park&Lee&Lee:10WCOM} in which transmit and receive
beamforming matrices are designed based on available CSI at
transmitters and receivers. Under this processing model, the SINR
for stream $m$ of user $k$ is given by {\small
\begin{eqnarray} \label{eq:SINR}
{\mathsf{SINR}}_k^{(m)}\big|_{\hat{\Hc}} = \hspace{1em} & &  \\
& & \hspace{-7em}
\frac{|(\ubf_k^{(m)})^H\hat\Hbf_{kk}\vbf_k^{(m)}|^2}
 {|(\ubf_k^{(m)})^H\Ebf_{kk}\vbf_k^{(m)}|^2+   \sum_{j\neq m}
  |(\ubf_k^{(m)})^H(\hat\Hbf_{kk}+\Ebf_{kk})\vbf_k^{(j)}|^2  +\sum_{i\neq k} \sum_{j=1}^d
  |(\ubf_k^{(m)})^H(\hat\Hbf_{ki}+\Ebf_{ki})\vbf_i^{(j)}|^2
+\sigma^2}, \nonumber
\end{eqnarray}}
where the numerator of the right-hand side (RHS) in
\eqref{eq:SINR} is the desired signal power, and the first,
second, third and fourth terms in the denominator of the RHS in
\eqref{eq:SINR} represent the interference purely by channel
uncertainty, inter-stream interference, other user interference
and thermal noise, respectively. (Here, the dependence of SINR on
$\hat{\Hc}$ is explicitly shown. Since the dependence is clear,
the notation $|_{\hat{\Hc}}$ will be omitted hereafter.) Because
the $\{\Ebf_{ki}\}$ are random, ${\mathsf{SINR}}_k^{(m)}$ is a random
variable for given $\hat{\mathcal{H}}$ and $\{\Vbf_k(\hat{\Hc}),
\Ubf_k(\hat{\Hc}), k\in \Kc\}$. Thus, an outage at stream $m$ of
user $k$ occurs if the supportable rate determined by the received
SINR \eqref{eq:SINR} is below the target rate $R_k^{(m)}$, and the
outage probability is  given by
\begin{equation} \label{eq:outageDef}
 {\mathrm{Pr}}\{\mbox{outage}\}
={\mathrm{Pr}} \left\{\log_2\left(1+{\mathsf{SINR}}_k^{(m)}
\right) \le R_k^{(m)}\right\}.
\end{equation}
By rearranging the terms in \eqref{eq:SINR}, the outage event can
be expressed as
\begin{equation} \label{eq:SINR_rearranged}
\sum_{i=1}^{K} \sum_{j=1}^{d} X_{ki}^{(mj)H}X_{ki}^{(mj)} \ge
\frac{|\ubf_k^{(m)H}\hat{\Hbf}_{kk}\vbf_k^{(m)}|^2}{2^{R_k^{(m)}}-1}-\sigma^2
=: \tau,
\end{equation}
where
\begin{equation}  \label{eq:Xklmj}
X_{ki}^{(mj)} := \left\{
\begin{array}{ll}
\ubf_k^{(m)H}\Ebf_{kk}\vbf_k^{(m)}, & i=k ~~ \mbox{and}~ j=m, \\
\ubf_k^{(m)H}(\hat{\Hbf}_{ki}+\Ebf_{ki})\vbf_i^{(j)}, &
\mbox{otherwise.}
\end{array}
\right.
\end{equation}
Since the $\{\Ebf_{ki}\}$ are  circularly-symmetric complex Gaussian
random matrices, $\{X_{ki}^{(mj)}, i=1,\cdots,K,j=1,\cdots,d\}$
are circularly-symmetric complex Gaussian random variables, and
the left-hand side (LHS) of \eqref{eq:SINR_rearranged} is a {\em
quadratic form of non-central Gaussian random variables}. To
simplify notation, we will use vector form from here on. In vector
form, \eqref{eq:SINR_rearranged} can be expressed as
\begin{equation}\label{eq:outage_rearranged}
\Xbf_k^{(m)H}\Xbf_k^{(m)} \ge \tau,
\end{equation}
where $\Xbf_k^{(m)}:=[X_{k1}^{(m1)}, \cdots, X_{k1}^{(md)},
X_{k2}^{(m1)},\cdots, X_{kK}^{(md)}]^T$. The elements of the mean
vector $\mubf_k^{(m)}(:=\Ebb\{\Xbf_{k}^{(m)}\})$ of
$\Xbf_{k}^{(m)}$ are given by
\begin{equation}  \label{eq:XkjMean}
[\mubf_k^{(m)}]_{(i-1)d+j}= \left\{
\begin{array}{ll}
 0,  & i=k,~ j=m, \\
\ubf_k^{(m)H}\hat{\Hbf}_{ki}\vbf_i^{(j)},  & \mbox{otherwise},
\end{array}
\right.
\end{equation}
for $i=1,\cdots,K$ and $j=1,\cdots, d$, and the covariance matrix
$\Sigmabf_k^{(m)}$ of $\Xbf_k^{(m)}$ is given by a block diagonal
matrix, since $\{\Ebf_{ki}, ~i=1,\cdots,K\}$ are independent for
different values of $i$, i.e.,
\begin{equation} \label{eq:Xcovariance1}
 \Sigmabf_k^{(m)} := \Ebb \{(\Xbf_{k}^{(m)}-
 {\mathbb{E}}\{\Xbf_{k}^{(m)}\})(\Xbf_{k}^{(m)}-
 {\mathbb{E}}\{\Xbf_{k}^{(m)}\})^H\} =
{\mathrm{diag}}(\Sigmabf_{k,1}^{(m)},\cdots,\Sigmabf_{k,K}^{(m)}),
\end{equation}
where the  $d\times d$  sub-block matrix $\Sigmabf_{k,i}^{(m)}$ is
given by
\begin{eqnarray} \label{eq:Xcovariance2}
\Sigmabf_{k,i}^{(m)} &=&
\sigma_h^2(\ubf_k^{(m)H}\Sigmabf_r\ubf_k^{(m)}) \left[
\begin{array}{cccc}
\vbf_i^{(1)H}\Sigmabf_t\vbf_i^{(1)} &
\vbf_i^{(2)H}\Sigmabf_t\vbf_i^{(1)} & \cdots &
\vbf_i^{(d)H}\Sigmabf_t\vbf_i^{(1)} \\
\vbf_i^{(1)H}\Sigmabf_t\vbf_i^{(2)} &
\vbf_i^{(2)H}\Sigmabf_t\vbf_i^{(2)} & \cdots &
\vbf_i^{(d)H}\Sigmabf_t\vbf_i^{(2)} \\
\vdots & \vdots & \ddots & \vdots \\
\vbf_i^{(1)H}\Sigmabf_t\vbf_i^{(d)} &
\vbf_i^{(2)H}\Sigmabf_t\vbf_i^{(d)} & \cdots &
\vbf_i^{(d)H}\Sigmabf_t\vbf_i^{(d)}
\end{array}
\right]
\end{eqnarray}
for each $i$. (The proof of \eqref{eq:Xcovariance2} is given in
Appendix \ref{append:proof_covariance}.) In the following sections, we will derive closed-form
expressions for \eqref{eq:outageDef}, investigate the behavior
 of the outage probability as a function of several parameters,
and propose an outage-based beam design algorithm.

\section{The Computation of the Outage Probability}
\label{sec:outage_prob}

In this section, we first derive a closed-form expression for the
outage probability in the general case of the Kronecker CSI error
model, and then consider special cases. After this, we examine the
behavior of the outage probability as a function of several
important system parameters based on the Chernoff bound.

\subsection{Closed-form expressions for the outage probability}

 For a Gaussian random vector $\Xbf
\sim \Cc\Nc(\mubf, \Sigmabf)$ with the eigendecomposition of its
covariance matrix $\Sigmabf = \Psibf \Lambdabf \Psibf^H$, the CDF
of $\Xbf^H\bar{\Qbf}\Xbf$ for some given $\bar{\Qbf}$ is given by
\cite{AlNaffouri&Hassibi:09ISIT}
\begin{equation} \label{eq:CDF_general14}
{\mathrm{Pr}}\{\Xbf^H\bar\Qbf\Xbf\le \tau \} =
\frac{1}{2\pi}\int_{-\infty}^{\infty} \frac{e^{\tau (\iota
\omega+\beta)}}{\iota \omega +\beta}\frac{e^{-c}}{\det(\Ibf+(
\iota \omega +\beta)\Qbf)} d\omega
\end{equation}
for some $\beta > 0$ such that $\Ibf+\beta\Qbf$ is positive
definite,  where
$\Qbf=\Lambdabf^{H/2}\Psibf^H\bar{\Qbf}\Psibf\Lambdabf^{1/2}$,
$\chibf=\Lambdabf^{-1/2}\Psibf^H\mubf$ and
$c=\chibf^H\left(\Ibf+\frac{1}{\iota \omega+\beta}\Qbf^{-1}
\right)^{\! -1}\!\!\!\!\chibf$.  From here on, we will derive
closed-form series expressions for the CDF of the outage
probability in several important cases by applying the residue
theorem used in \cite{Raphaeli:96IT} to the integral form
\eqref{eq:CDF_general14} for the CDF. First, we consider the most
general case of the Kronecker CSI error model. The outage
probability in this case
 is given by the following theorem.

\vspace{0.5em}

\begin{theorem}\label{theo:general_outage_thm}
For given transmit and receive beamforming matrices
$\{\Vbf_k=[\vbf_k^{(1)},\cdots,\vbf_k^{(d)}]\}$ and
$\{\Ubf_k=[\ubf_k^{(1)},\cdots,\ubf_k^{(d)}]\}$ designed based on
$\hat{\mathcal{H}}=\{\hat{\Hbf}_{ki}, \Sigmabf_t, \Sigmabf_r\}$,
the outage probability for stream $m$ of user $k$ with the target
rate $R_k^{(m)}$ under the CSI error model
(\ref{eq:CSIerrormodel1}) and (\ref{eq:CSIerrormodel2}) is given
by
\begin{eqnarray}
{\mathrm{Pr}}\{{\mbox{outage}}\}
&=&{\mathrm{Pr}}\{\log_2(1+{\mathsf{SINR}}_k^{(m)}) \le R_k^{(m)}\} \nonumber\\
&=& -\sum_{i=1}^{\kappa} \frac{e^{-(\frac{\tau}{\lambda_i} +
\sum_{j=1}^{\kappa_i}|\chi_i^{(j)}|^2)} }{\lambda_i^{\kappa_i}}
\sum_{n = \kappa_i-1}^{\infty} \frac{1}{n!} g_{i}^{(n)}(0)
\frac{1}{(n-\kappa_i+1)!}
\left(\frac{\sum_{j=1}^{\kappa_i}|\chi_i^{(j)}|^2}{\lambda_i}\right)^{\!\!
n-\kappa_i+1} \label{eq:general_outage}
\end{eqnarray}
where $\tau$ is given in \eqref{eq:SINR_rearranged};
$\{\lambda_i,i=1,\cdots,\kappa\}$ are all the distinct eigenvalues
of the $Kd \times Kd$ covariance matrix $\Sigmabf_k^{(m)}$ in
\eqref{eq:Xcovariance1}  with eigendecomposition
$\Sigmabf_{k}^{(m)}=\Psibf_k^{(m)}\Lambdabf_k^{(m)}\Psibf_k^{(m)H}$;
$\kappa_i$ is the multiplicity\footnote{Since  $\Sigmabf_k^{(m)}$
is a normal matrix, we have $Kd=\sum_{i=1}^\kappa \kappa_i$.} of
the eigenvalue $\lambda_i$; $\chi_{i}^{(j)}$ is the element
of vector
\begin{equation}\label{eq:theorem1chibf}
\chibf_k^{(m)}:=(\Lambdabf_{k}^{(m)})^{-\frac{1}{2}}\Psibf_{k}^{(m)H}\mubf_k^{(m)}
\end{equation}
corresponding to the $j$-th eigenvector of the eigenvalue
$\lambda_i$ ($1\le j \le \kappa_i$), i.e., it is  the $j$-th
element of
$(\lambda_i\Ibf_{\kappa_i})^{-\frac{1}{2}}\Psibf_{k,i}^{(m)H}
\mubf_k^{(m)}$.  ($\Psibf_{k,i}^{(m)}$ is a $Kd\times \kappa_i$
matrix composed of the eigenvectors of $\Sigmabf_{k}^{(m)}$
associated with $\lambda_i$.); {\small
\begin{equation} \label{eq:theorem1gis}
g_{i}(s)= \frac{e^{\tau s}}{s-{1}/{\lambda_i}} \cdot
\frac{\exp\left(-\sum_{ p\neq i }
\frac{(s-1/\lambda_i)\lambda_p}
        {1+(s-1/\lambda_i)\lambda_p }\sum_{q=1}^{\kappa_p}|\chi_p^{(q)}|^2 \right)}
{\prod_{p \neq i} \Big(1+(s-1/\lambda_i)\lambda_p\Big)^{\kappa_p}};
\end{equation}
} and $g_{i}^{(n)}(s)$ is the $n$-th derivative of $g_{i}(s)$
w.r.t. $s$.
\end{theorem}

\vspace{0.5em}
\begin{proof}
By using \eqref{eq:CDF_general14} and the facts $\bar{\Qbf}=\Ibf$
and $\Xbf_k^{(m)} \sim \Cc\Nc(\mubf_k^{(m)},\Sigmabf_k^{(m)})$ in
this case, we obtain the outage probability for stream $m$ of user
$k$ in an integral form as
\begin{equation} \label{eq:raw_integral}
  {\mathrm{Pr}}\{\Xbf_k^{(m)H}\Xbf_k^{(m)} \ge \tau \}
  = 1-\frac{1}{2\pi \iota }\int_{\beta-\iota \infty}^{\beta+ \iota \infty}
\frac{e^{s\tau}}{s}\cdot \frac{e^{-\sum_{i=1}^{\kappa}
\frac{s\lambda_i}{1+s\lambda_i}(\sum_{j=1}^{\kappa_i}|\chi_i^{(j)}|^2)}}
{\prod_{i=1}^{\kappa}(1+s\lambda_i)^{\kappa_i}}\ ds,
\end{equation}
where $s=\beta + \iota \omega$ ~($\beta>0$). The outage probability
\eqref{eq:raw_integral} can be expressed as a contour integral:
\begin{eqnarray} \label{eq:thm1_integral}
  {\mathrm{Pr}}\{\Xbf_k^{(m)H}\Xbf_k^{(m)} \ge \tau \}
  &=& 1-\frac{1}{2\pi\iota}\oint_{{\mathcal{C}}} \underbrace{\frac{e^{s\tau}}{s}\cdot\frac{e^{-\sum_{i=1}^{\kappa}
\frac{s\lambda_i}{1+s\lambda_i}(\sum_{j=1}^{\kappa_i}|\chi_i^{(j)}|^2)}}
{\prod_{i=1}^{\kappa}(1+s\lambda_i)^{\kappa_i}}}_{=: F(s)}\ ds,
\end{eqnarray}
where ${\mathcal{C}}$ is a contour of integration containing the
imaginary axis and the whole left half plane of the complex plane.
By the residue theorem, the sum of the residues at singular points
of $F(s)$ which do not have positive real parts yields the contour
integral in \eqref{eq:thm1_integral} times $2\pi \iota$. It is
easy to see that the singular points of $F(s)$ are $s=0$ and
$s=-1/\lambda_i$, $i=1,\cdots,\kappa$. Since
$\Sigmabf_{k,i}^{(m)}$ are all positive-definite,
$\Sigmabf_k^{(m)}$ is positive definite and $\lambda_i
> 0$ for all $i$. So, the outage probability is given by
\begin{equation}  \label{eq:theorem1PoutForm1}
 {\mathrm{Pr}}\{\mbox{outage}\}
=1-\Big( \underset{s=0}{\mbox{Res }} F(s)
     +\sum_{i=1}^{\kappa}  \underset{s=-1/\lambda_i}{\mbox{Res}}
      F(s)
\Big).
\end{equation}
It is also easy to see from \eqref{eq:thm1_integral} that the
residue of $F(s)$ at $s=0$ is $\underset{s=0}{\mbox{Res}}F(s)=1$.
To compute $\underset{s=-1/\lambda_i}{\mbox{Res }}F(s)$, for
each $i$ we introduce $G_i(s)$ defined as
{
\begin{eqnarray*}
G_{i}(s) &:=&  F\left(s-\frac{1}{\lambda_i}\right) =
\frac{e^{\tau(s-1/\lambda_i)}}{s-1/\lambda_i}\cdot
     \frac{e^{-\sum_{p=1}^{\kappa}
     \frac{\lambda_p(s-1/\lambda_i)}
         {1+\lambda_p(s-1/\lambda_i)}(\sum_{q=1}^{\kappa_p}|\chi_p^{(q)}|^2)}}       {\prod_{p=1}^{\kappa}(1+\lambda_p(s-1/\lambda_i))^{\kappa_p}}\\
&\ =&
\frac{e^{\tau(s-1/\lambda_i)}}{s-1/\lambda_i}\cdot
    \frac{e^{-\frac{\lambda_is-1}{\lambda_is}\sum_{j=1}^{\kappa_i}|\chi_i^{(j)}|^2}}
      {(\lambda_i s)^{\kappa_i}} \cdot \underbrace{\frac{e^{-\sum_{p \neq i}
    \frac{\lambda_p(s-1/\lambda_i)}
         {1+\lambda_p(s-1/\lambda_i)}\sum_{q=1}^{\kappa_p}|\chi_p^{(q)}|^2}}
         {\prod_{p \neq i}(1+\lambda_p(s-1/\lambda_i))^{\kappa_p}}}_{=:I_1} \\
&\ =& e^{-(\frac{\tau}{\lambda_i}+\sum_{j=1}^{\kappa_i}|\chi_i^{(j)}|^2)} \times
\underbrace{\frac{e^{\frac{1}{\lambda_is}\sum_{j=1}^{\kappa_i}|\chi_i^{(j)}|^2}}{(\lambda_i s)^{\kappa_i}}}_{=:f_i(s)} \times \underbrace{\left(\frac{e^{\tau
s}}{s-1/\lambda_i}\times I_1 \right)}_{=:g_{i}(s)}.
\end{eqnarray*}
} Now, the residue of $F(s)$ at $s=-1/\lambda_i$ is
transformed to that of $G_i(s)$ at  $s=0$. The  Laurent series
expansion  of $f_i(s)$  and the  Taylor series expansion of
$g_i(s)$ at $s=0$ are given respectively by
{\small
\begin{equation}
 f_i(s)=\frac{1}{(\lambda_is)^{\kappa_i}}\sum_{n=0}^{\infty}\frac{1}{n!}
 \left(\frac{\sum_{j=1}^{\kappa_i}|\chi_i^{(j)}|^2}{\lambda_is}\right)^{\!\!
 n} ~~~\mbox{and}~~~
 g_{i}(s)=\sum_{n=0}^{\infty}\frac{1}{n!}g_{i}^{(n)}(0)s^n.
\end{equation}
}
By multiplying the two series and computing the coefficient of
$1/s$, we obtain the residue of $G_{i}(s)$ at $s=0$ as
\begin{equation}
\underset{s=0}{\mbox{Res\ }}G_i(s) =
\frac{e^{-(\frac{\tau}{\lambda_i}
+\sum_{j=1}^{\kappa_i}|\chi_i^{(j)}|^2)}}{\lambda_i^{\kappa_i}}
\sum_{n=\kappa_i-1}^{\infty}
\frac{1}{n!}
g_{i}^{(n)}(0)
\frac{1}{(n-\kappa_i+1)!}
\left(\frac{\sum_{j=1}^{\kappa_i}|\chi_i^{(j)}|^2}{\lambda_i}\right)^{\!\! n-\kappa_i+1}
\end{equation}
for each $i$. Finally,  substituting the residues into
\eqref{eq:theorem1PoutForm1} yields \eqref{eq:general_outage}.
\end{proof}

\vspace{0.5em} \noindent To compute
\eqref{eq:general_outage}, we need to compute $\{\lambda_i\}$,
$\{\chi_i^{(j)}\}$ and the higher order derivatives of $g_i(s)$.
The first two terms are easy to compute since they are related
with the mean vector of size $Kd$ and the covariance matrix of
size  $Kd \times Kd$. Furthermore, the higher order derivatives of
$g_i(s)$ can also be computed efficiently based on recursion.
(Please see Appendix \ref{append:derivative}.) Note that in the
case that the elements $\Hbf_{ki}^{(w)}$ in
\eqref{eq:CSIerrormodel2} have difference variances,
\eqref{eq:general_outage} is still valid since the difference
variances only change the covariance matrix
\eqref{eq:Xcovariance1} and the outage expression depends on the
covariance matrix \eqref{eq:Xcovariance1} through $\{\lambda_i\}$
and $\{\chi_i^{(j)}\}$.

Next, we provide some useful corollaries to Theorem
\ref{theo:general_outage_thm} regarding the outage probability in
meaningful special cases. First, we consider the case in which a
subset of channels are perfectly known at receiver $k$, i.e.,
$\Hbf_{ki}=\hat{\Hbf}_{ki}$ and $\Ebf_{ki}= {\mathbf {0}}$ for
some $i\in\Kc$. This corresponds to the case in which channel
estimation or CSI feedback for some links is easier than that for
other links. For example, the desired link channel may be easier
to estimate than others. The outage probability in this case is
given by the following corollary. \vspace{0.5em}

\begin{corollary} \label{corr:known_H_outage_prob}
When perfect CSI for some channel links including the desired link
is available  at receiver $k$, i.e., $\hat{\Hbf}_{ki}=\Hbf_{ki}$
for $i\in\Upsilon_k \subset \Kc$, the outage probability for
stream $m$ of user $k$ is given by
\begin{eqnarray} \label{eq:general_outage_known_H}
{\mathrm{Pr}}\{{\mbox{outage}}\}
&=&{\mathrm{Pr}}\{\log_2(1+{\mathsf{SINR}}_k^{(m)}) \le R_k^{(m)}\} \nonumber\\
&=& -\sum_{i=1}^{\kappa^\prime}
\frac{e^{-(\frac{\tau^\prime}{\lambda_i} +
\sum_{j=1}^{\kappa_i}|\chi_i^{(j)}|^2)} }{\lambda_i^{\kappa_i}}
\sum_{n = \kappa_i-1}^{\infty} \frac{1}{n!} g_{1,i}^{(n)}(0)
\frac{1}{(n-\kappa_i+1)!}
\left(\frac{\sum_{j=1}^{\kappa_i}|\chi_i^{(j)}|^2}{\lambda_i}\right)^{\!\!
n-\kappa_i+1}
\end{eqnarray}
where $\tau^\prime$ is defined below; $\{\lambda_i,
i=1,\cdots,\kappa^\prime\}$ is the set of all the distinct
eigenvalues of the covariance matrix \eqref{eq:Xcovariance1};
$\kappa_i$ is the multiplicity of $\lambda_i$, satisfying
$(K-|\Upsilon_k|)d = \sum_{i=1}^{\kappa^\prime} \kappa_i$;
$\chi_{i}^{(j)}$ is given in \eqref{eq:theorem1chibf}; and {\small
\begin{equation}
g_{1,i}(s)= \frac{e^{\tau^\prime s}}{s-{1}/{\lambda_i}} \cdot
\frac{\exp\left(-\sum_{ p\neq i } \frac{(s-1/\lambda_i)\lambda_p}
        {1+(s-1/\lambda_i)\lambda_p }\sum_{q=1}^{\kappa_p}|\chi_p^{(q)}|^2 \right)}
{\prod_{p \neq i} \Big(1+(s-1/\lambda_i)\lambda_p\Big)^{\kappa_p}}.
\end{equation}
}

\begin{proof}
When CSI for some links including the desired link is perfect, the
outage event at stream $m$ of user $k$ is given by {\scriptsize
\begin{eqnarray*}
\log_2\left(1+\frac{|\ubf_k^{(m)H}\hat\Hbf_{kk}\vbf_k^{(m)}|^2}
{\sum_{i\in\Upsilon_k}\sum_{\substack{j=1,\\ j\neq m}}^{d}
|\ubf_k^{(m)H}{\hat\Hbf}_{ki}\vbf_i^{(j)}|^2
+\sum_{{\substack{i\in\Upsilon_k,\\ i\neq k}}}
|\ubf_k^{(m)H}{\hat\Hbf}_{ki}\vbf_i^{(m)}|^2 +\sum_{i\in
\Upsilon_k^c}\sum_{j=1}^{d}
|\ubf_k^{(m)H}(\hat{\Hbf}_{ki}+\Ebf_{ki}) {\vbf_i^{(j)}}|^2
+\sigma^2}\right) \le R_k^{(m)}
\end{eqnarray*}
} since  $\Ebf_{ki}=  \mathbf{0}$ for $i\in\Upsilon_k$. Thus, in
this case the outage event is expressed in a quadratic form as
follows: {\small
\begin{equation} \label{eq:cor1_outage_rearranged}
\sum_{i\in\Upsilon_k^c}\sum_{j=1}^{d} X_{ki}^{(mj)H} X_{ki}^{(mj)}
 \ge
\frac{|\ubf_k^{(m)H}\hat{\Hbf}_{kk}\vbf_k^{(m)}|^2}{2^{R_k^{(m)}}-1}
-\sum_{i\in\Upsilon_k}\sum_{\substack{j=1, \\ j\neq m}}^d
|\ubf_k^{(m)H}\hat\Hbf_{ki}\vbf_i^{(j)}|^2
-\sum_{\substack{i\in\Upsilon_k,\\ i\neq k}}
 |\ubf_k^{(m)H}\hat\Hbf_{ki}\vbf_i^{(m)}|^2
-\sigma^2 =:\tau^\prime,
\end{equation}
} and we have $X_{ki}^{(mj)}\equiv 0$ for all $i\in\Upsilon_k$
(See \eqref{eq:Xklmj}). The size of $\Xbf_k^{(m)}$ now reduces to
$(K-|\Upsilon_k|)d$, and  the size of the covariance matrix
$\Sigmabf_k^{(m)}$ is $(K-|\Upsilon_k|)d\times(K-|\Upsilon_k|)d$.
With the new threshold $\tau^\prime$, the same argument as that in
Theorem \ref{theo:general_outage_thm} can be applied to yield the
result.
\end{proof}
\end{corollary}
\vspace{0.5em}

\noindent Thus, when perfect CSI is available for some links, the
order of the distribution is reduced under the same structure.
Next, consider the specific beam design method of interference
alignment and the corresponding outage probability, which can be
obtained by Corollary \ref{corr:known_H_outage_prob} and is given
in the following corollary.

\begin{corollary} \label{corr:IA_outage}
When the desired channel link is perfectly known (i.e.
$k\in\Upsilon_k$) and $\{\Vbf_k\}$ and $\{\Ubf_k\}$ are designed
under IA based on $\hat{\mathcal{H}}$, the outage probability for
stream $m$ of user $k$ is given by
\begin{eqnarray}
{\mathrm{Pr}}\{\mbox{outage}\} &=&
-\sum_{i=1}^{\kappa^\prime}\frac{1}{\lambda_i^{\kappa_i}}e^{-\frac{\tau^\prime}{\lambda_i}}
\frac{1}{(\kappa_i-1)!}g_{1,i}^{(\kappa_i-1)}(0).
\end{eqnarray}

\begin{proof}
First, express the random term in
\eqref{eq:cor1_outage_rearranged} as
$\sum_{i\in\Upsilon_k^c}\sum_{j=1}^{d} X_{ki}^{(mj)H}
X_{ki}^{(mj)} = (\Xbf_k^{(m)})^H \Xbf_k^{(m)}$. When the beam is
designed under IA based on $\hat{\Hc}$, we have
${\mathbb{E}}\{\Xbf_k^{(m)}\} = {\mathbf{0}}$  since
$\ubf_k^{(m)H}\hat\Hbf_{ki}\vbf_i^{(j)}$ $=0$ for all $i \in
\Kc\backslash \{k\} \supset \Upsilon_k^c$, $j=1,\cdots,d$. (See
\eqref{eq:XkjMean}.) Hence, $\chibf_k^{(m)}={\mathbf 0}$ and thus
$\chi_i^{(j)}= 0$ for all $i$ and $j$. (See
\eqref{eq:theorem1chibf}.) Then, the terms in the infinite
series in \eqref{eq:general_outage_known_H} are zero  for all $n >
\kappa_i-1$ from the fact that $0^0=1$ and $0!=1$, and the result
follows.
\end{proof}
\end{corollary}

\vspace{0.5em}

\noindent The outage probability for single stream
communication is given in Corollary \ref{coro:singleStream}.

\begin{corollary} \label{coro:singleStream}
When $d=1$ and all eigenvalues of $\Sigmabf_{k}^{(m)}$ are
distinct, the outage probability for user $k$ is given by
{\small
 \begin{eqnarray}\label{eq:cor_singlestream}
{\mathrm{Pr}}\{{\mbox{outage}}\}
&=&{\mathrm{Pr}}\{\log_2(1+{\mathsf{SINR}}_k) \le R_k\} \nonumber \\
&=& -\sum_{i=1}^{K} \frac{e^{-(|\chi_i|^2+\tau /\lambda_i)}
}{\lambda_i} \sum_{n=0}^{\infty}\left(\frac{1}{n!}\right)^{\!\!2}
\left(\frac{|\chi_i|^2}{\lambda_i}\right)^{\!\!n} g_{i}^{(n)}(0),
\end{eqnarray}
}
where $g_i(s)$ in \eqref{eq:theorem1gis} reduces to $g_{i}(s)=
\frac{e^{\tau s}}{s-{1}/{\lambda_i}} \cdot \frac{e^{ -\sum_{p \neq
i} \frac{\lambda_p(s-1/\lambda_i)}
        {1+\lambda_p(s-1/\lambda_i)}|\chi_p|^2} }
{\prod_{p \neq i} \Big(1+\lambda_p (s-1/\lambda_i)\Big)}$.
 (Here, we have omitted the stream superscripts since the stream
index is unique.)

\begin{proof}
Since all eigenvalues are assumed to be distinct, there are
$\kappa=K$ eigenvalues with $\kappa_i=1$ for all $i$. Substituting
these into Theorem \ref{theo:general_outage_thm} yields the
result.
\end{proof}
\end{corollary}
\vspace{0.5em} \noindent Now, let us consider a simpler case for
$d=1$ with no antenna correlation. In this case, the outage
probability is given as an explicit function of the channel
uncertainty level $\sigma_h^2$, and it is given by the following
corollary to Theorem \ref{theo:general_outage_thm}.

\vspace{0.5em}
\begin{corollary}\label{cor:d=1&R=I}
 When $d=1$ and there is  no antenna correlation,  the outage probability is given by
\begin{equation} \label{eq:CDF_definite_d1}
{\mathrm{Pr}}\{\mbox{outage}\} =
-\frac{1}{(\sigma_h^2)^K}e^{-(\frac{\tau}{\sigma_h^2}+\|\chibf_k\|^2)}
\sum_{n=K-1}^{\infty}\frac{1}{n!}g^{(n)}(0)\frac{1}{(n-K+1)!}
\left(\frac{\|\chibf_k\|^2}{\sigma_h^2}\right)^{\!\!n-K+1},
\end{equation}
where $\chibf_k ={\mathbb{E}}\{\Xbf_k\}/\sigma_h$ and $g(s) =
\frac{e^{\tau s}}{s-1/\sigma_h^2}$.

\begin{proof}
In this case, an outage at user $k$ occurs if and only if
$\Xbf_k^H\Xbf_k \ge
\frac{|\ubf_k^H\hat{\Hbf}_{kk}\vbf_k|^2}{2^{R_k}-1} -\sigma^2$.
Now, the covariance matrix $\Sigmabf_k$ of $\Xbf_k$ is
$\sigma_h^2\Ibf_K$ (see \eqref{eq:Xcovariance1} and
\eqref{eq:Xcovariance2}), and thus there is only one eigenvalue
$\sigma_h^2$ with multiplicity $K$. Moreover, $\chibf_k
={\mathbb{E}}\{\Xbf_k\}/\sigma_h$ from \eqref{eq:theorem1chibf}
since $\Psibf_k=\Ibf$ and $\Lambdabf_k=\sigma_h^2\Ibf$. By
substituting these into Theorem \ref{theo:general_outage_thm}, the
outage probability \eqref{eq:CDF_definite_d1} is obtained.
\end{proof}
\end{corollary}

\subsection{The behavior analysis of the outage probability based on the Chernoff bound}
\label{subsec:Chernoff}


 The obtained exact expressions for the outage probability
 in the previous subsection  can easily be computed numerically, and will be used for the robust beam design based on the outage probability in
Section \ref{sec:beam_design}. Before we address the outage-based
robust beam design problem, let us investigate the behavior of the
outage probability as a function of several parameters. Suppose
that transmit and receive beam vectors $\{\vbf_k^{(m)},
\ubf_k^{(m)}\}$ are designed by  some known method
based on $\hat{\bf{\Hc}}$. For the
given beam vectors, as seen in the obtained expressions, the
outage probability is a function of other system parameters such
as the known channel mean $\{\hat{\Hbf}_{ki}\}$, the noise
variance $\sigma^2$, the channel uncertainty level $\sigma_h^2$,
the antenna correlation $\Sigmabf_t$ and $\Sigmabf_r$, and the
target rate $R_{k}^{(m)}$. Here, the dependence on
$\hat{\Hbf}_{kk}$, $\sigma^2$ and $R_k^{(m)}$ is via the threshold
$\tau(\hat{\Hbf}_{kk},\sigma^2,R_k^{(m)})$, and the dependence on
$\sigma_h^2$, $\Sigmabf_t$, $\Sigmabf_r$ and
$\{\hat{\Hbf}_{ki},i\ne k\}$ is via
$\chibf_k^{(m)}(\Sigmabf_k^{(m)}(\sigma_h^2, \Sigmabf_t,
\Sigmabf_r), \Ebb\{\Xbf_k^{(m)}\}(\hat\Hbf_{ki}))$ and the
eigenvalues of $\Sigmabf_{k,i}^{(m)}(\sigma_h^2,\Sigmabf_t,\Sigmabf_r)$. This
complicated dependence structure makes it difficult to analyze the
properties of the outage probability as a function of the system
parameters. Thus, in this subsection we  apply the Chernoff
bounding technique \cite{Poor:book} to the tractable\footnote{In
certain cases of $d >1$, Chernoff bound can still be obtained when
each element in $\Xbf_k^{(m)}$ is independent of the others. Such
cases include the case that there is no antenna correlation and
the transmit beam vectors are orthogonal as in the IA beam case.
In this case, similar results to the case of $d=1$ are obtained.}
case of $d=1$ to obtain insights into the outage probability as a function of
several important parameters. When $d=1$, the outage event is expressed
as
\begin{equation} \label{eq:ChernoffBound1}
{\mathrm{Pr}}\bigg\{\Xbf_k^H\Xbf_k\ge \tau =
\frac{|\ubf_k^H\hat{\Hbf}_{kk}\vbf_k|^2}{(2^{R_k}-1)}-\sigma^2\bigg\} =
{\mathrm{Pr}}\bigg\{\sum_{i=1}^K X_{ki}^H X_{ki}\ge\tau \bigg\}.
\end{equation}
Since $\Ebf_{k1},\cdots, \Ebf_{kK}$ are independent and
circularly-symmetric complex Gaussian random matrices,
$X_{k1},\cdots,$ $X_{kK}$ are independent and circularly-symmetric
complex Gaussian random variables. (See \eqref{eq:Xklmj}.) Thus,
the term on the LHS in the second bracket in
\eqref{eq:ChernoffBound1} is a sum of independent random
variables, and the Chernoff bound can be applied to yield
\begin{equation} \textstyle
{\mathrm{Pr}}\{\Xbf_k^H\Xbf_k\ge \tau \} \le e^{- \tau s}
\prod_{i=1}^K {\mathbb{E}}\left\{e^{s|X_{ki}|^2}\right\}
\end{equation}
for any $s>0$. The moment generating function (m.g.f.) of
$|X_{ki}|^2$ ($X_{ki} \sim \Cc\Nc(\mu_{ki},\sigma_{ki}^2)$) is
given by $ {\mathbb{E}}\{e^{s|X_{ki}|^2}\} =
\frac{1}{1-\sigma_{ki}^2s}
\exp\left(\frac{|\mu_{ki}|^2s}{1-\sigma_{ki}^2s}\right)$ for
$s<1/\sigma_{ki}^2$, where $\mu_{kk}=0$,
$\mu_{ki}=\ubf_k^H\hat{\Hbf}_{ki}\vbf_i$ for $i\ne k$, and
$\sigma_{ki}^2=$
$\sigma_h^2(\ubf_k^H\Sigmabf_r\ubf_k)(\vbf_i^H\Sigmabf_t\vbf_i)$.
(See (\ref{eq:Xklmj},\ref{eq:XkjMean},\ref{eq:Xcovariance2}).)
Therefore, the Chernoff bound on the outage probability is given
by
\begin{align}
{\mathrm{Pr}}\{\Xbf_k^H\Xbf_k\ge \tau \}
&\le e^{-\tau s}
\prod_{i=1}^K \frac{1}{1-\sigma_{ki}^2s}
\exp\left(\frac{|\mu_{ki}|^2s}{1-\sigma_{ki}^2s}\right) \nonumber\\
&= \exp\left\{ -\left[\tau s+\sum_{i=1}^K\log(1-\sigma_{ki}^2s)
+\sum_{i=1}^K\frac{|\mu_{ki}|^2s}{\sigma_{ki}^2s-1}\right]\right\}
\label{eq:upper_bound}
\end{align}
for $0< s < \min_{i}\{1/\sigma_{ki}^2\}$. Now,
\eqref{eq:upper_bound} provides a tool to analyze the behavior of
the outage probability as a function of several important
parameters. The most desired property is the behavior of the
outage probability as a function of the channel uncertainty level.
This behavior is explained in the following theorem.

\vspace{0.5em}

\begin{theorem}  \label{theo:scaling1}
When $d=1$, as $\sigma_h^2 \rightarrow 0$, the outage probability
decreases to zero, and the decay rate is given  by
\begin{equation} \label{eq:scaling1}
{\mathrm{Pr}}\{\mbox{outage}\} \le
e^{-c_1}\cdot\exp(-c_2/{\sigma_h^2})
\end{equation}
for some $c_1$ and $c_2 > 0$ not depending on $\sigma_h^2$, if the
target rate $R_k$ and the designed transmit and receive beam
vectors $\{\vbf_k, \ubf_k\}$ satisfy
\begin{equation}  \label{eq:Theorem2d=1Decaysigmah2}
R_k <
\bar{R}_k=\log_2
\left(1+\frac{|\ubf_k^H\hat\Hbf_{kk}\vbf_k|^2}
{\sum_{i=1}^K \frac{|\mu_{ki}|^2}
{1-\frac{(\ubf_k^H\Sigmabf_r\ubf_k)(\vbf_i^H\Sigmabf_t\vbf_i)}{tr(\Sigmabf_r)tr(\Sigmabf_t)}
}+\sigma^2}\right).
\end{equation}
\end{theorem}

\vspace{0.5em}
\begin{proof} \eqref{eq:upper_bound} is valid for any $s \in (0, \min_{i}\{1/\sigma_{ki}^2\})$. So,
let $s={1}/{\sigma_h^2
\mbox{tr}(\Sigmabf_t)\mbox{tr}(\Sigmabf_r)}~~ (<
\min_{i}\{1/\sigma_{ki}^2\}$ since $||\vbf_k||=||\ubf_k||=1$ and
$\sigma_{ki}^2=\sigma_h^2(\ubf_k^H\Sigmabf_r\ubf_k)(\vbf_i^H\Sigmabf_t\vbf_i)
\le \sigma_h^2\mbox{tr}(\Sigmabf_t)\mbox{tr}(\Sigmabf_r)$ for all
$i$). Then, the exponent in \eqref{eq:upper_bound} is given by
{\small
\begin{eqnarray*}
& & -\frac{\tau}{\sigma_h^2 \mbox{tr}(\Sigmabf_t)
\mbox{tr}(\Sigmabf_r)}
-\sum_{i=1}^K\log\left[1-\frac{(\ubf_k^H\Sigmabf_r\ubf_k)(\vbf_i^H\Sigmabf_t\vbf_i)}{\mbox{tr}(\Sigmabf_t)\mbox{tr}(\Sigmabf_r)}\right]
-\sum_{i=1}^K\frac{|\mu_{ki}|^2}{\sigma_h^2(\ubf_k^H\Sigmabf_r\ubf_k)(\vbf_i^H\Sigmabf_t\vbf_i)
-\sigma_h^2 \mbox{tr}(\Sigmabf_t) \mbox{tr}(\Sigmabf_r)} \\
&=& -\frac{1}{\sigma_h^2} \underbrace{
\left\{\frac{\tau}{\mbox{tr}(\Sigmabf_t)\mbox{tr}(\Sigmabf_r)}+\sum_{i=1}^K
\frac{|\mu_{ki}|^2}{(\ubf_k^H\Sigmabf_r\ubf_k)(\vbf_i^H\Sigmabf_t\vbf_i)-\mbox{tr}(\Sigmabf_t)\mbox{tr}(\Sigmabf_r)}
\right\} }_{(=: c_2)} -
\underbrace{\sum_{i=1}^K\log\left[1-\frac{(\ubf_k^H\Sigmabf_r\ubf_k)(\vbf_i^H\Sigmabf_t\vbf_i)}{\mbox{tr}(\Sigmabf_t)\mbox{tr}(\Sigmabf_r)}\right]}_{(=:c_1)}.
\end{eqnarray*}
} Now, substituting $\tau = |\ubf_k^H
\hat{\Hbf}_{kk}\vbf_k|^2/(2^{R_k}-1)-\sigma^2$ into the inequality
$c_2
>0$ yields \eqref{eq:Theorem2d=1Decaysigmah2}.
\end{proof}

\vspace{0.5em} \noindent Theorem \ref{theo:scaling1} states that
 the outage probability decays to zero as the CSI quality
improves, more precisely, it decays exponentially w.r.t. the
inverse of channel estimation MSE (or equivalently w.r.t. the
channel $K$ factor), if the target rate is below $\bar{R}_k$. In
the Fisherian inference framework, the inverse of estimation MSE
is information. Thus, another way we can view the above is that
{\em the outage probability decays exponentially as the Fisher
information for channel state increases, if the target rate is
below a certain value.} So, the outage probability due to channel
uncertainty is another case in which information is the error
exponent as in many other inference problems.  In certain cases,
the condition \eqref{eq:Theorem2d=1Decaysigmah2} can be simplified
considerably. For example, when interference-aligning beam vectors
based on $\hat{\Hc}$ are used at the transmitters and receivers,
we have $\mu_{ki} = \ubf_k^H \hat{\Hbf}_{ki} \vbf_i=0$ for $i\ne
k$ in addition to $\mu_{kk}=0$, and the condition is simplified to
$R_k <
\log_2\left(1+\frac{|\ubf_k^H\hat{\Hbf}_{kk}\vbf_k|^2}{\sigma^2}
\right)$.  Thus, in the case of interference alignment the outage
probability can be made arbitrarily small by improving the CSI
quality if the target rate is strictly less than the rate obtained
by using $\hat{\Hbf}_{kk}$ as the nominal channel. Next, consider
the outage behavior as the effective SNR, $\Gamma_{eff}:=|\ubf_k^H
\hat{\Hbf}_{kk} \vbf_k|^2/\sigma^2$, increases. Since the two
terms determining the effective SNR are contained only in $\tau$,
it is straightforward to see from \eqref{eq:upper_bound} that
\begin{equation}
\mbox{Pr}\{\mbox{outage}\} \le c_3 \exp\left( -c_4 \Gamma_{eff}
\right),
\end{equation}
for some $c_3$ and $c_4=s\sigma^2/(2^{R_k}-1) >0$ not depending on
$\Gamma_{eff}$.  Finally, consider the case in which the target rate
$R_k$ decreases. One can  expect that the outage probability
decays to zero if the target rate decreases to zero. The decaying
behavior in this case is given in the following theorem.

\vspace{0.5em}
\begin{theorem}
When $d=1$, as $R_k \rightarrow 0$, the outage probability
decreases to zero, and the decay rate is given  by
\begin{equation}
{\mathrm{Pr}}\{\mbox{outage}\} \le c_6
\exp\left(-\frac{c_7}{2^{R_k}-1} \right)=c_6
\exp\left(-\frac{c_7^\prime}{R_k+o(R_k)} \right)
\end{equation}
for some $c_7, ~c_7^\prime >0$ not depending on $R_k$. The last
equality is when $R_k$ is near zero.
\end{theorem}

\begin{proof}
Let $s$ be any positive constant contained in an interval $(0, ~1/\max_i
\{\sigma_h^2(\ubf_k^H\Sigmabf_r\ubf_k)(\vbf_i^H\Sigmabf_t\vbf_i)\})$.
Then,  the exponent in \eqref{eq:upper_bound} becomes
{\small
\begin{eqnarray*}
& & -\tau s- \underbrace{
\sum_{i=1}^K\log[1-\sigma_h^2(\ubf_k^H\Sigmabf_r\ubf_k)(\vbf_i^H\Sigmabf_t\vbf_i)s]
-\sum_{i=1}^K\frac{|\mu_{ki}|^2s}{s\sigma_h^2(\ubf_k^H\Sigmabf_r\ubf_k)(\vbf_i^H\Sigmabf_t\vbf_i)-1}}_{(=:c_5)} \\
&=&
-\left(\frac{|\ubf_k^H\hat\Hbf_{kk}\vbf_k|^2}{2^{R_k}-1}-\sigma^2\right)s
- c_5 = -\frac{|\ubf_k^H\hat\Hbf_{kk}\vbf_k|^2}{2^{R_k}-1}s -
c_5^\prime.
\end{eqnarray*}
} Hence, the Chernoff bound is given by
${\mathrm{Pr}}\{\mbox{outage}\} \le c_6
\exp\left(-\frac{s|\ubf_k^H\hat{\Hbf}_{kk}\vbf_k|^2}{2^{R_k}-1}
\right)=c_6 \exp\left(-\frac{c_7^\prime}{R_k+o(R_k)} \right)$ for
some $c_7^\prime >0$. The last equality is when $R_k$ is near
zero. In this case, we have $2^{R_k}-1=(\log 2)R_k+o(R_k)$ by
Taylor's expansion.
\end{proof}

\section{Outage-Based Robust Beam Design}
\label{sec:beam_design}

  In this section, we propose an outage-based  beam design algorithm
based on  the closed-form expressions for the outage probability
derived in the previous section. Our assumption is that
$\hat{\Hc}$ is given for the beam design, as mentioned earlier.
Suppose that transmit and receive beamforming matrices $\{\Vbf_k,
\Ubf_k\}$ are designed by using any available beam design method
based on $\hat{\Hc}$. Based on the designed $\{\Vbf_k, \Ubf_k\}$
and known $\{\hat{\Hc}, \sigma^2\}$, one can compute and use a
nominal rate for transmission. Since $\hat{\Hc}$ is not perfect,
however, an outage may occur depending on the CSI error if the
nominal rate is used for transmission. Of course, the outage
probability can be made small by making the transmission rate low
or by improving the CSI quality, as seen in Section
\ref{subsec:Chernoff}. However, these methods are inefficient
sometimes since we may have limitations in the CSI quality or need
as high rate as possible for given $\hat{\Hc}$. Further, in many
wireless systems the target outage probability for transmission is
determined and the data transmission is performed under such an
outage constraint. Thus,  we here consider the beam design
problem when the outage probability is given as a system
parameter. In particular, we consider the following
per-stream based beam design problem  to maximize the
 sum $\epsilon$-outage rate for given $\hat{\Hc}$:
\begin{eqnarray} %
&\mathop{\mbox{maximize}}\limits_{ \{\vbf_k^{(m)}\},
{\{\ubf_k^{(m)}\}} } &
         \sum_{k=1}^{K}\sum_{m=1}^{d} R_k^{(m)}  \label{eq:weighted_sum}\\
&\mbox{subject to}& {\mathrm{Pr}}\{\log_2(1 +
{\mathrm{SINR}}_k^{(m)} \big|_{\hat{\bf{\mathcal{H}}}}) \le
R_k^{(m)} \} \le \epsilon
\label{eq:outage_const}\\
& & \|\ubf_k^{(m)}\|=\|\vbf_k^{(m)}\|=1, \quad \forall k \in
{\mathcal{K}},\  m=1,\cdots,d, \label{eq:beam_const}
\end{eqnarray}
where the $\epsilon$-outage rate for stream $m$ of user $k$ is the
maximum rate satisfying \eqref{eq:outage_const}. Like other beam
design problems in MIMO interference channels, the simultaneous
joint optimal design for all transmit and receive beam vectors for
this problem also seems difficult. Hence, we propose an iterative
approach to the above sum $\epsilon$-outage rate
maximization problem. The proposed method is explained as follows.
In the first step, we initialize $\{\vbf_k^{(m)}\}$ and
$\{\ubf_k^{(m)}\}$ properly (here a known beam design algorithm
for the MIMO interference channel can be used), and then find
optimal rate-tuple
$(R_1^{(1)},\cdots,R_1^{(d)},R_2^{(1)},\cdots,R_K^{(d)})$ that
maximizes the sum for given $\{\vbf_k^{(m)},
\ubf_k^{(m)}\}$ under  the outage constraint. This step is
performed based on the derived outage probability expressions in
the previous section. Since designing each $R_k^{(m)}$ does not
affect others, this step can be done separately for each
$R_k^{(m)}$. Since the outage probability for stream $m$ of user
$k$ increases monotonically w.r.t. $R_k^{(m)}$, the optimal
$R_k^{(m)}$ in this step is the rate with the outage probability
$\epsilon$. In the second step, for the obtained rate-tuple and
receive beam vectors $\{\ubf_k^{(m)}\}$ in the first step, we
update the transmit beam vectors $\{\vbf_k^{(m)}\}$ to  minimize
the maximum of the outage probabilities of all streams and all
users. (Since the outage probabilities of all streams of all users
are $\epsilon$ at the end of the first step, this means that the
outage probability
 decreases for all streams and all users.)  Here, we apply the alternating minimization technique \cite{Csiszar&Shields:book} to circumvent the difficulty in the joint transmit beam design. (The change in one transmit beam
vector affects the outage probabilities of other users.) That is,
we optimize one transmit beam vector while fixing all the others
at a time. We iterate this procedure from the first stream of
transmitter $1$ to the last stream of user $K$ until this step
converge. In the third step, we design the receive beam vector
$\ubf_k^{(m)}$ to minimize the outage probability at stream $m$ of
user $k$ with the rate-tuple determined in the first step and
$\{\vbf_k^{(m)}\}$ determined in the second step for each $(k,m)$.
This optimization can also be performed separately for each stream
of each user since the receiver filter for one stream does not
affect the performance of other streams. Finally, we go back to
the first step with the updated transmit and receive beam vectors
(in the revisited first step, the rate for each stream will be
increased by increasing the outage probability upto to $\epsilon$
again), and iterate the procedure until the sum $\epsilon$-outage
rate does not change. We have summarized the sum outage rate maximizing beam design algorithm in Table \ref{table:algorithm}.

\begin{table}[!ht]
\centering
\begin{tabular}{p{400pt}}
\hline \normalsize{ \bfseries The Proposed Algorithm} \\
\hline \vspace{-0.4em}
\begin{enumerate}
\small{ \item[]\noindent Input: channel state estimate
$\hat{\bf\mathcal{H}}$
and allowed outage probability $\epsilon$. %

\item[0.] Initialize $\{\vbf_k^{(m)}\}$ and  $\{\ubf_k^{(m)}\}$ as
sets of unit-norm vectors  properly.

\item[1.] For given $\{\Vbf_k\}$ and $\{\Ubf_k\}$, find
$(R_1^{(1)},\cdots,R_K^{(d)})$ that maximizes $\sum_{k=1}^{K}\sum_{m=1}^{d} R_k^{(m)}$ while the outage constraint is satisfied.  %

\item[2.] Update $\{\Vbf_k = [
\vbf_k^{(1)},\cdots,\vbf_k^{(d)}]$\} for $\{R_k^{(m)}\}$ and
$\{\Ubf_k^{(m)}\}$ given from step 1.

\textbullet For pair $(i,j)$, fix $\{\vbf_k^{(m)},k=1,\cdots,K,
~m=1,\cdots,d\}\backslash \{\vbf_i^{(j)}\}$ and $\{\Ubf_k\}$ and
solve
\begin{equation} \label{eq:TablePropAlog}
\vbf_i^{(j)} = \mathop{\arg\min}_{\vbf \in {\mathbb{C}}^{N_t}}
\max_{k,m} {\mathrm{Pr}}\{\mbox{outage}_k^{(m)}\}.
\end{equation}
(Here, a commercial tool such as the matlab fminimax function
can be used to solve \eqref{eq:TablePropAlog} together with the derived outage expression.)

\textbullet Iterate the above step from the first stream of
transmitter $1$ to the last stream of transmitter $K$ until
$\{\Vbf_1,\cdots,\Vbf_K\}$ converges. %

\item[3.] For receiver $1$ to $K$, obtain the receive filter
$\ubf_k^{(m)}$ that minimize the outage probability of stream $m$
of receiver $k$ for given $\{\Vbf_k\}$ from step 2 and given
$R_k^{(m)}$ from step 1. (Here, again a commercial tool such
as the matlab fmincon function can be used together with the
derived outage expression.)%

\item[4.] Go to step 1 and repeat the whole procedure until the
algorithm converges. } \vspace{-1.5em}
\end{enumerate}
\\
\hline
\end{tabular}
\caption{The proposed algorithm for sum $\epsilon$-outage rate
 maximization with channel uncertainty  }
\label{table:algorithm}
\end{table}

\begin{theorem}
The proposed beam design algorithm converges.
\end{theorem}

\begin{proof} It is straightforward to see that the sum
$\epsilon$-outage rate increases monotonically for each iteration
of the three steps of the proposed algorithm. Also, the maximum
sum rate is bounded by the rate with perfect CSI. Hence, the algorithm converges by the monotone convergence theorem for real sequences.
\end{proof}

\section{Numerical Results}
\label{sec:numerical}

In this section, we provide some numerical results to validate our
series derivation,  to examine the outage probability as a
function of several system parameters and to evaluate the
performance of the proposed beam design algorithm.
 For given $\Sigmabf_t$, $\Sigmabf_r$, $K_{ch}^{(ki)}$ and $\Gamma^{(k)}$, we first generated $\{\hat{\Hbf}_{ki}\}$ randomly according to zero-mean
Gaussian distribution, and then scaled $\hat{\Hbf}_{ki}$ to yield
$\|\hat\Hbf_{ki}\|_F^2=N_tN_r$ for all $(k,i)$.  In this way, the
channel $K$ factor and the SNR were simply controlled by
$\sigma_h^2$ and $\sigma^2$, respectively. After
$\{\hat{\Hbf}_{ki}\}$ were generated as such, we generated
$\{\Ebf_{ki}\}$ according to \eqref{eq:CSIerrormodel2} and the
true channel was determined by \eqref{eq:CSIerrormodel1} if
necessary\footnote{The computation of the closed-form outage
probability requires only the channel statistics and
$\{\hat{\Hbf}_{ki}\}$ regarding the channel information, but for
Monte Carlo runs we need to generate $\{\Ebf_{ki}\}$.}. For
simplicity, we used $K_{ch}^{(ki)}=K_{ch}$ for all $(k,i)$ and
$\Gamma^{(k)}=\Gamma$ for all $k$.


\begin{figure}[t]
\centering \scalefig{0.6}\epsfbox{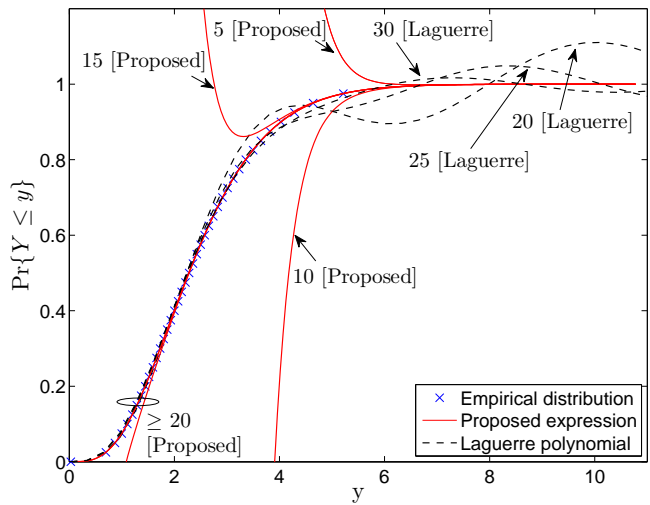}
\caption{Comparison of two series expressions for the CDF of
quadratic form of Gaussian random variables.
$\Xbf\sim{\mathcal{CN}}([0.5,0.5,0.5, 0.5]^T, 0.3\Ibf_4)$,
$\bar{\mathbf{Q}}=[1, 0.5, 0, 0; 0.5, 1, 0, 0; 0, 0, 1, 0; 0, 0,
0, 1]$, and $\beta=2$ for Laguerre series expansion.}
\label{fig:cdf_comparison}
\end{figure}

First, Fig. \ref{fig:cdf_comparison} compares the convergence
behavior of the derived series in this paper with that of the
series fitting method
\cite{Kotz:67AMS-1,Kotz:67AMS-2,Mathai&Provost:book,Nabar:05WC}
based on the Laguerre basis functions for a given set of
parameters shown in the label of the figure. It is seen that
indeed our series converges from the upper tail first whereas the
series fitting method converges from the lower tail first. (For a
proof of this in the identity covariance matrix case, please refer
to Appendix \ref{append:convergence}.) Note that the series fitting
method yields large error at the upper tail distribution even with
a reasonably large number of terms.  With this verification, next
consider  the outage behavior as a function of several system
parameters.
\begin{figure}[H]
\centering
\scalefig{0.6}\epsfbox{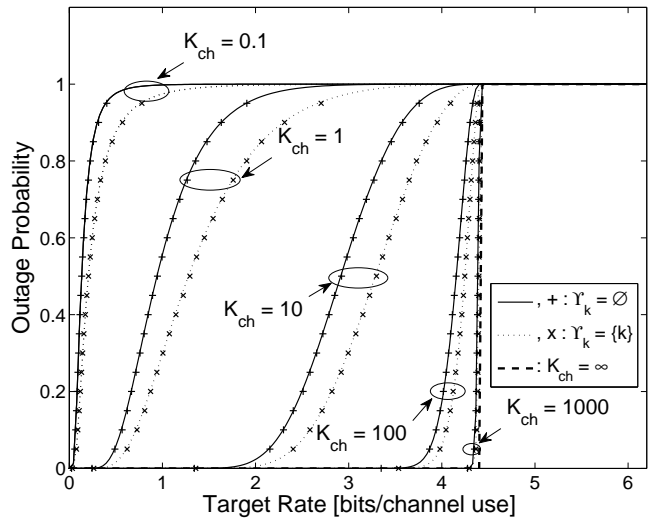}
\caption{Outage probability versus the target rate $R_k$ ($K=3$,
$N_t=N_r=2d=2$, $\Sigmabf_t=\Sigmabf_r=\Ibf$, $\Gamma = 15$ dB.
Transmit and receive beam vectors are obtained by the IIA
algorithm in \cite{Gomadam&Jafar:11IT}.)}
\label{fig:outage_prob_rate}
\end{figure}
Fig. \ref{fig:outage_prob_rate} shows the outage
probability w.r.t. the target rate $R_k$ for a given set
$\{\hat{\Hbf}_{ki}\}$ (randomly generated as above) with several
different channel $K$ factors, when $K=3, N_t=N_r=2d=2$,
$\Sigmabf_t=\Sigmabf_r=\Ibf$, $\Gamma=15$ dB and the transmit and
receive beam vectors were designed by the iterative interference
alignment (IIA) algorithm \cite{Gomadam&Jafar:11IT}. The solid and
dotted lines represent the result of our analysis, and the markers
$+$ and $\times$ indicate the result of Monte Carlo runs for the
outage probability. {The theoretical outage curves in Fig.
\ref{fig:outage_prob_rate} were obtained by using
\eqref{eq:general_outage_known_H} with the first 38 terms in the
infinite series.} It is seen that our analysis  matches the result
of Monte Carlo runs very well. The dashed line shows the outage
performance  when $K_{ch}=\infty$, i.e., all transmitters and
receivers have perfect CSI. In the case of $K_{ch}=\infty$, we
have a sharp transition behavior across $R_{limit}$ determined by
the SINR \eqref{eq:SINR} with $\Ebf_{ki}={\mathbf{0}}$ for all
$(k, i)$. It is seen that the outage performance deteriorates from
the ideal step curve of $K_{ch}=\infty$, as the CSI quality
degrades. The solid lines correspond to the outage performance for
the finite values of $K_{ch}$, when the CSI for all channel links
is imperfect. It is seen that $K_{ch}= 100$ ~(20 dB) yields
reasonable outage performance compared with the perfect CSI case
in this setup. Note that the gain in the outage probability by
knowing the desired link perfectly is not negligible. (See the
dotted lines.) Fig. \ref{fig:outage_prob_rated=2} show the outage
probability w.r.t. the target rate $R_k$ for a given set
$\{\hat{\Hbf}_{ki}\}$ with several different $K_{ch}$, when $K=3,
N_t=N_r=2d=4$, $\Sigmabf_t=\Sigmabf_r=\Ibf$, $\Gamma=25$ dB and
the transmit and receive beam vectors were designed by the IIA
algorithm. Similar behavior is seen as in the single stream case,
i.e., the outage performance generally deteriorates as $K_{ch}$
decreases. However, it is interesting to observe in the multiple
stream case that sufficiently good but not perfect CSI quality
yields better outage performance than the perfect CSI in the high
outage probability regime. (See Fig. \ref{fig:outage_prob_rated=2}
(b).) This implies that  in the multiple stream case  the second
term (i.e., the self inter-stream interference term) in the
denominator of the SINR formula \eqref{eq:SINR} is made smaller by
$\Ebf_{kk}$'s being negatively aligned with $\Hbf_{kk}$ than in
the case of $\Ebf_{kk}\equiv 0$. However, this is not useful in
system operation since the system is operated in the low outage
probability regime. {All the theoretical curves in Figures
\ref{fig:outage_prob_rated=2} (a) and (b) were obtained by
\eqref{eq:general_outage_known_H} with the first 45 terms in the
infinite series.} { Fig. \ref{fig:outage_prob_svd} shows the
outage probability curves  when the transmit and receive
beamforming vectors are respectively chosen  as the right and left
singular vectors corresponding to the largest singular value of
the desired channel and the other parameters are identical to the
case in Fig. \ref{fig:outage_prob_rate}. A similar outage
probability behavior to the previous case is observed. }
\begin{figure}[htbp]
\centerline{ \SetLabels
\L(0.23*-0.1) (a)\ CDF \\
\L(0.68*-0.1) (b)\ Residual error \\
\endSetLabels
\leavevmode
\strut\AffixLabels{
\scalefig{0.52}\epsfbox{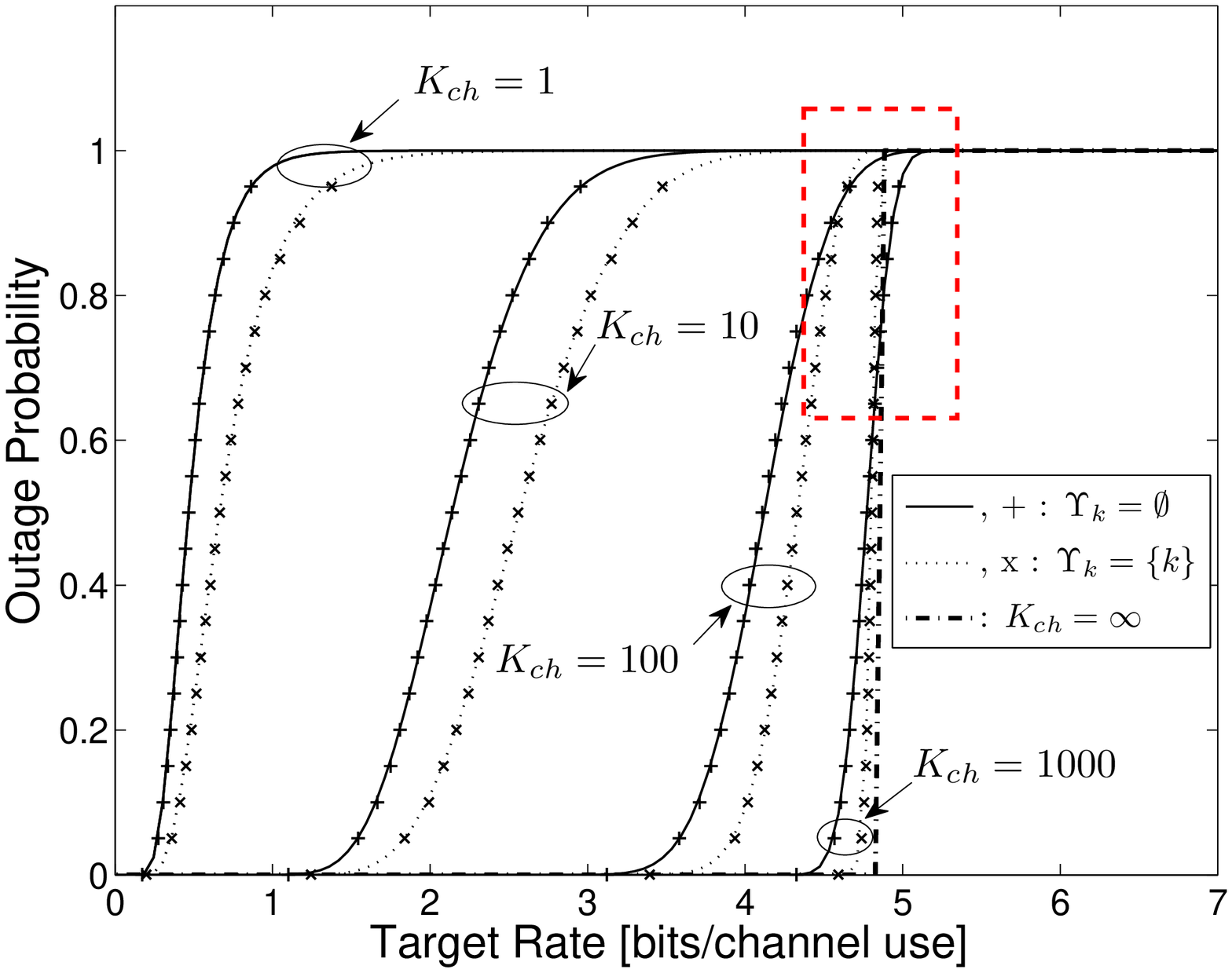}
\scalefig{0.52}\epsfbox{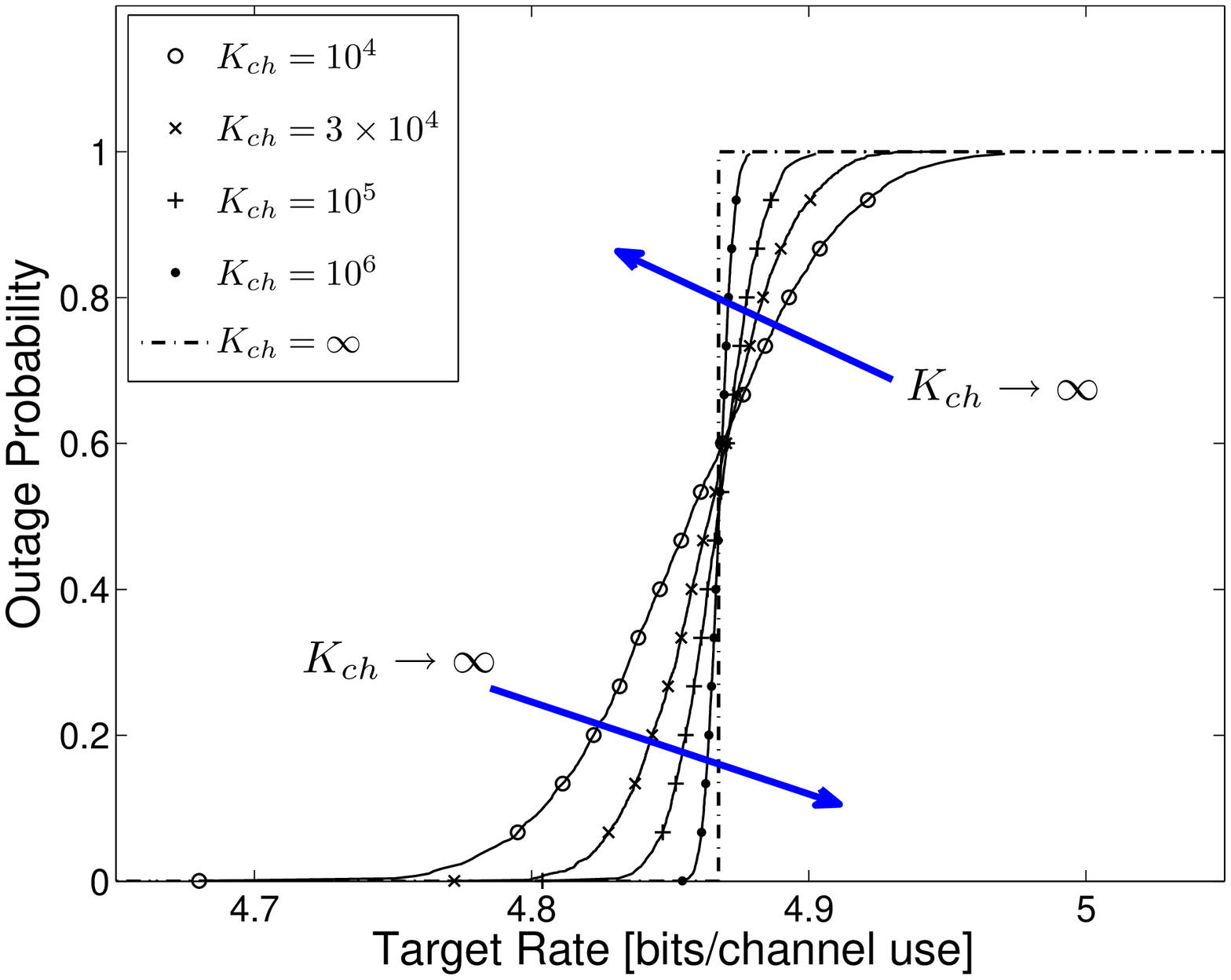} } }
\vspace{1em} \caption{Outage probability versus the target rate $R_k$ ($K=3$,
$N_t=N_r=2d=4$, $\Sigmabf_t=\Sigmabf_r=\Ibf$, $\Gamma = 25$ dB.
Transmit and receive beam vectors are designed by the IIA
algorithm in \cite{Gomadam&Jafar:11IT}.)} \label{fig:outage_prob_rated=2}
\end{figure}

Next, the outage probability w.r.t. the channel $K$ factor for a
given set $\{\hat{\Hbf}\}$ for several values of the target rate
$R_k$ is shown in Fig. \ref{fig:outage_prob_rician}, where the
outage probability along the $y$-axis is drawn in log scale. (The
same setup as for Fig. \ref{fig:outage_prob_rate} was used
and the IIA algorithm is used for the transmit and receive
beam design.  Here, \eqref{eq:general_outage_known_H} with the
first 38 terms in the infinite series was used to compute the
analytic curves.) As predicted by Theorem \ref{theo:scaling1},
the outage probability indeed decays exponentially w.r.t. the
channel $K$ factor (equivalently, w.r.t. the inverse of
$\sigma_h^2$). The exponent depends on the target rate $R_k$; the
higher the target rate is, the smaller the exponent is. This
decaying behavior is also predicted in Theorem
\ref{theo:scaling1}; the exponent $c_2$ in \eqref{eq:scaling1} is
proportional to $\tau$, and $\tau$ is inversely proportional to
the target rate $R_k$. It is seen that the outage probability does
not decay as $K_{ch}$ increases, if $R_k$ is larger than
$R_{limit}$. In addition to the exact outage probability, the
Chernoff bound in this case is shown in Fig.
\ref{fig:outage_prob_rician} as the lines with dots and dashes. It
is seen that the Chernoff bound is not very tight but the decaying
slope is the same as that of the exact outage probability.
\begin{figure}[H]
\centerline{
\scalefig{0.6}\epsfbox{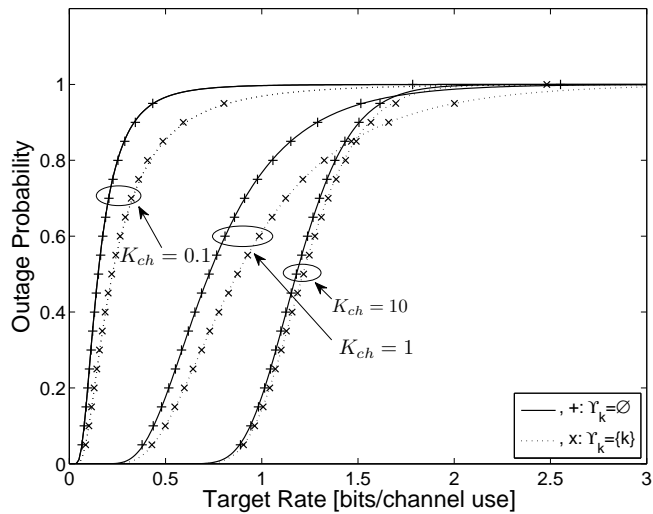}
} \caption{Outage probability versus  the target rate $R_k$
($K=3$, $N_t=N_r=2d=2$, $\Sigmabf_t=\Sigmabf_r=\Ibf$, $\Gamma=15$
dB. Transmit and receive beam vectors are respectively chosen as
the right and left singular vectors corresponding to the largest
singular value of the desired channel matrix.)}
\label{fig:outage_prob_svd}
\end{figure}
Figures  \ref{fig:outage_prob_correlation1} and
\ref{fig:outage_prob_correlation2} show the impact of antenna
correlation on the outage probability. We adopted the exponential
antenna correlation profile considered in
\cite{Belmega&Lasaulce&Debbah:09WCOM, Hjorungnes&Gesbert:07WCOM}.
Under this model, the $(i,j)$-th element of the antenna
correlation matrix $\Sigmabf_t$ (or $\Sigmabf_r$) in
\eqref{eq:CSIerrormodel2} is given by $\rho^{|i-j|}$, where $\rho
\in [0, ~1]$ is a parameter determining the correlation strength.
Since $\mbox{tr}(\Sigmabf_t)=N_t$ and $\mbox{tr}(\Sigmabf_r)=N_r$
for this exponential antenna correlation model, we have the same
transmit and receive powers as in the case of no antenna
correlation, i.e., $\Sigmabf_t=\Ibf$ and $\Sigmabf_r = \Ibf$.
Since the outage probability depends on $\{\hat{\Hbf}_{ki}\}$ as
well as on $\Sigmabf_t$ and $\Sigmabf_r$, we generated one hundred
$\{\hat{\Hbf}_{ki}\}$ randomly in the way that we explained
already, and averaged the corresponding 100 outage probabilities
to see the impact of the error correlation only. Other aspects of
the system configuration  were the same as those for Figures
\ref{fig:outage_prob_rate} and \ref{fig:outage_prob_rician}. It is
seen that the error correlation decreases the outage probability
especially when the CSI quality is very bad, but the gain becomes
negligible when the CSI quality is good.

\begin{figure}[H]
\centering
\scalefig{0.6}\epsfbox{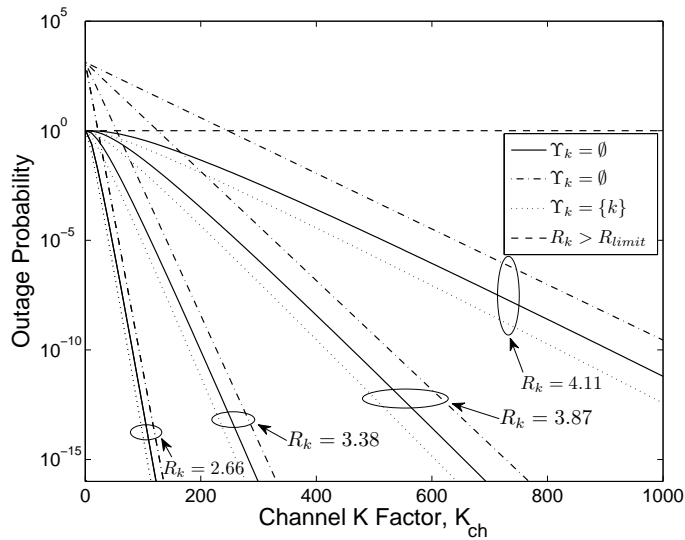}
\caption{Outage probability versus $K_{ch}$ ($K=3$,
$N_t=N_r=2d=2$, $\Sigmabf_t=\Sigmabf_r=\Ibf$, $\Gamma = 15$ dB.
Transmit and receive beam vectors are designed by the IIA
algorithm in \cite{Gomadam&Jafar:11IT}.)}
\label{fig:outage_prob_rician}
\end{figure}

\begin{figure}[H]
\centering
\scalefig{0.6}\epsfbox{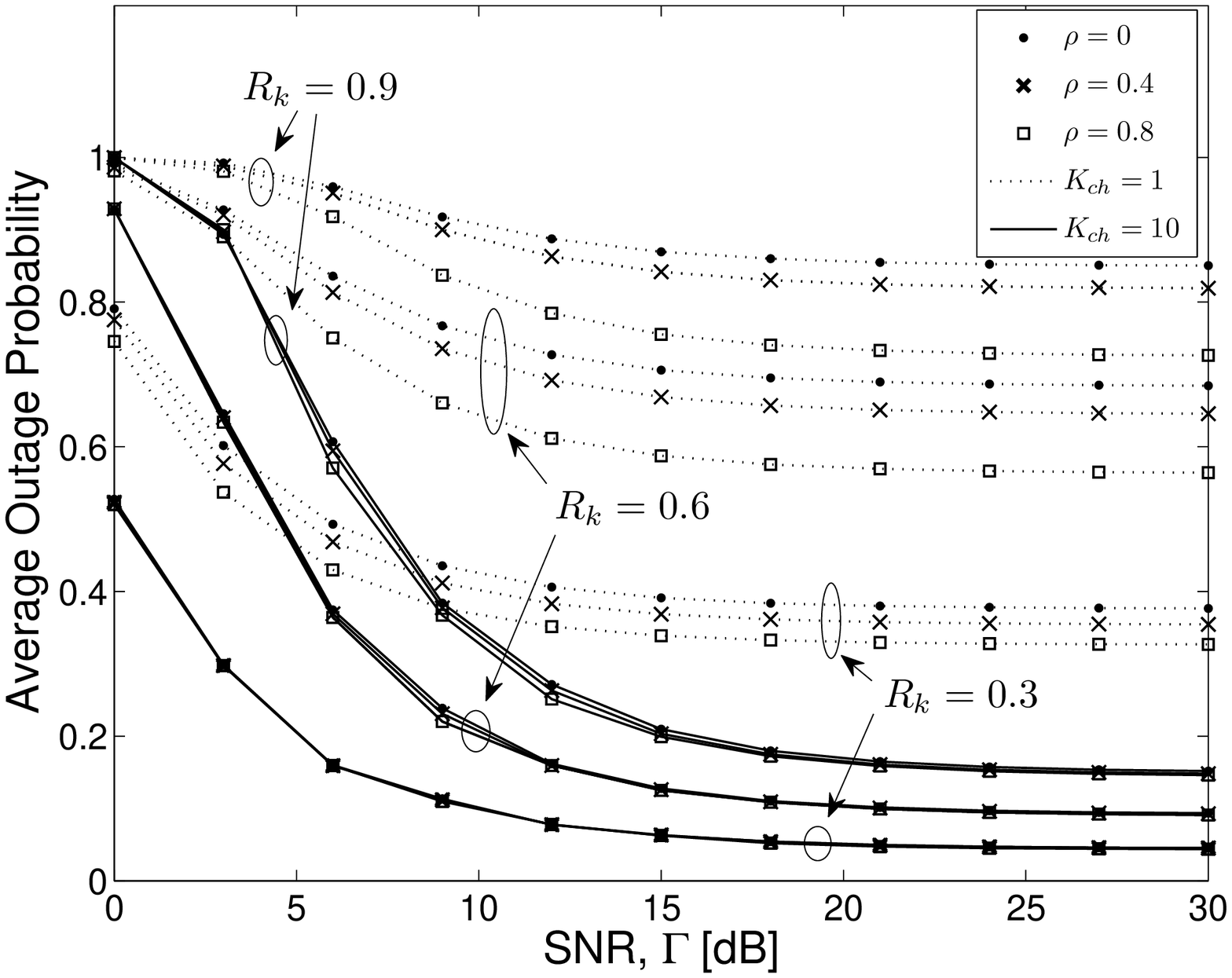}
\caption{Average outage probability versus $\Gamma$ ($K=3$,
$N_t=N_r=2d=2$. Transmit and receive beam vectors designed by
the IIA algorithm in \cite{Gomadam&Jafar:11IT}.)}
\label{fig:outage_prob_correlation1}
\end{figure}

\begin{figure}[H]
\centering
\scalefig{0.55}\epsfbox{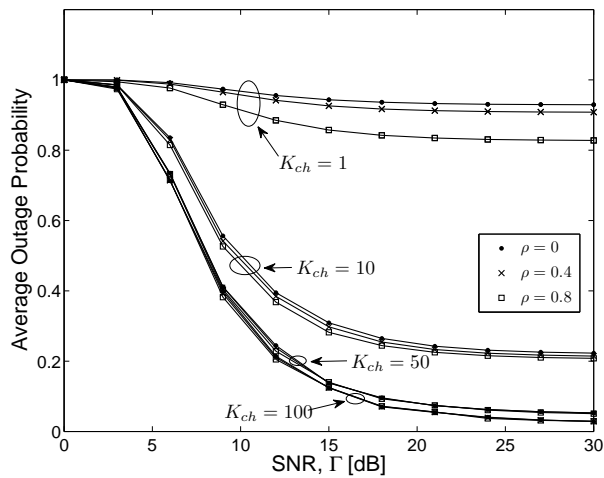}
\caption{Average outage probability versus $\Gamma$ ($K=3$,
$N_t=N_r=2d=2$, $R_k=1.2$. Transmit and receive beam vectors
designed by the IIA algorithm in \cite{Gomadam&Jafar:11IT}.)}
\label{fig:outage_prob_correlation2}
\end{figure}

Finally, the performance of the proposed beam design
algorithm maximizing the sum $\epsilon$-outage rate was evaluated.
As reference, we adopted the max-SINR algorithm and IIA algorithm
in \cite{Gomadam&Jafar:11IT}. Although the max-SINR and IIA
algorithms were originally proposed to design beam vectors with
perfect channel information, we applied the algorithms to design
beam vectors by treating the imperfect channel
$\hat{\bf{\mathcal{H}}}$ as the true channel. The
$\epsilon$-outage rate of the max-SINR algorithm (or the IIA
algorithm) is defined as the maximum rate that can be achieved
under the outage constraint of $\epsilon$ using the beam vectors
designed by the max-SINR algorithm (or the IIA algorithm). Once
$\{\Vbf_k\}$ and $\{\Ubf_k\}$ are designed by any design method
for given $\Sigmabf_t$, $\Sigmabf_r$ and $\{\hat{\Hbf}_{ki}\}$,
the outage probability corresponding to the designed beam vectors
is easily computed as a function of the target rate $R_k$ from
Theorem \ref{theo:general_outage_thm}. Thus, for the beam vectors
designed by the max-SINR and IIA algorithms as well as for those
designed by the proposed design algorithm in Section
\ref{sec:beam_design}, the $\epsilon$-outage rate $R_k$ can easily
be obtained. Figures \ref{fig:outage_capacity-e10} and
\ref{fig:outage_capacity-e20} show the sum $\epsilon$-outage rate
of the proposed beam design method averaged over thirty different
sets of $\{\hat{\Hbf}_{ki}\}$ for $\epsilon=0.1$ and
$\epsilon=0.2$, respectively, when $K=3, N_t=N_r=2d=2$ and
$\Sigmabf_t=\Sigmabf_r=\Ibf$ for different $K_{ch}$'s. (The
outage probability expression \eqref{eq:CDF_definite_d1}  with the
first 40 terms was used to compute the outage probability.) It is
seen that the proposed algorithm outperforms the IIA and max-SINR
algorithms in all SNR, and the max-SINR algorithm shows good
performance almost comparable to the proposed algorithm at low
SNR. However, as SNR increases, the performance of the max-SINR
algorithm degrades to that of the IIA algorithm  (the two
algorithm themselves converge as SNR increases) and there is a
considerable gain by exploiting the channel uncertainty.

\begin{figure}[H]
\centering
\scalefig{0.6}\epsfbox{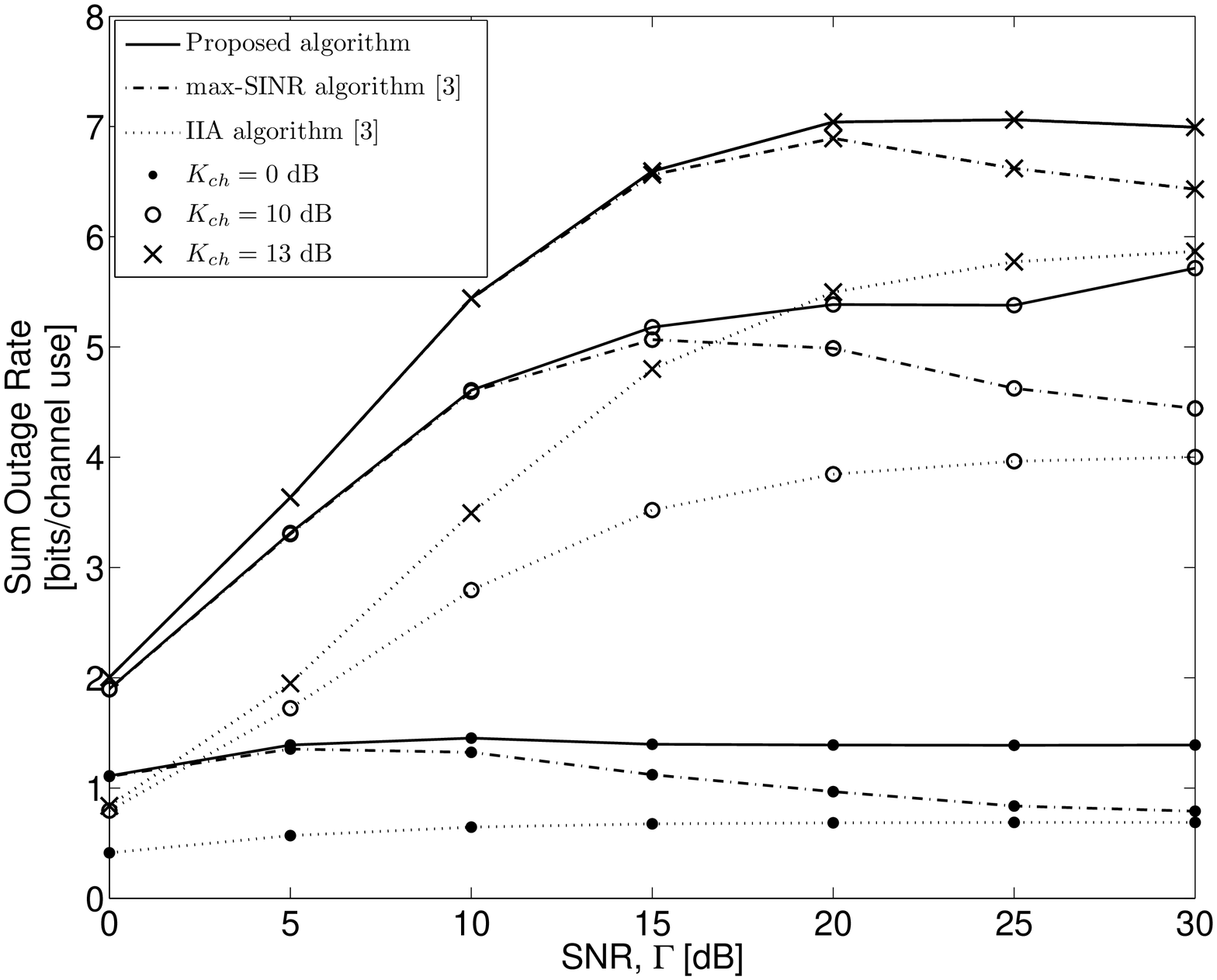}
\caption{Sum $\epsilon$-outage rate for $\epsilon=0.1$ ($K=3$,
$N_t=N_r=2d=2$, $\Sigmabf_t=\Sigmabf_r=\Ibf$) }
\label{fig:outage_capacity-e10}
\end{figure}

\begin{figure}[H]
\centering
\scalefig{0.6}\epsfbox{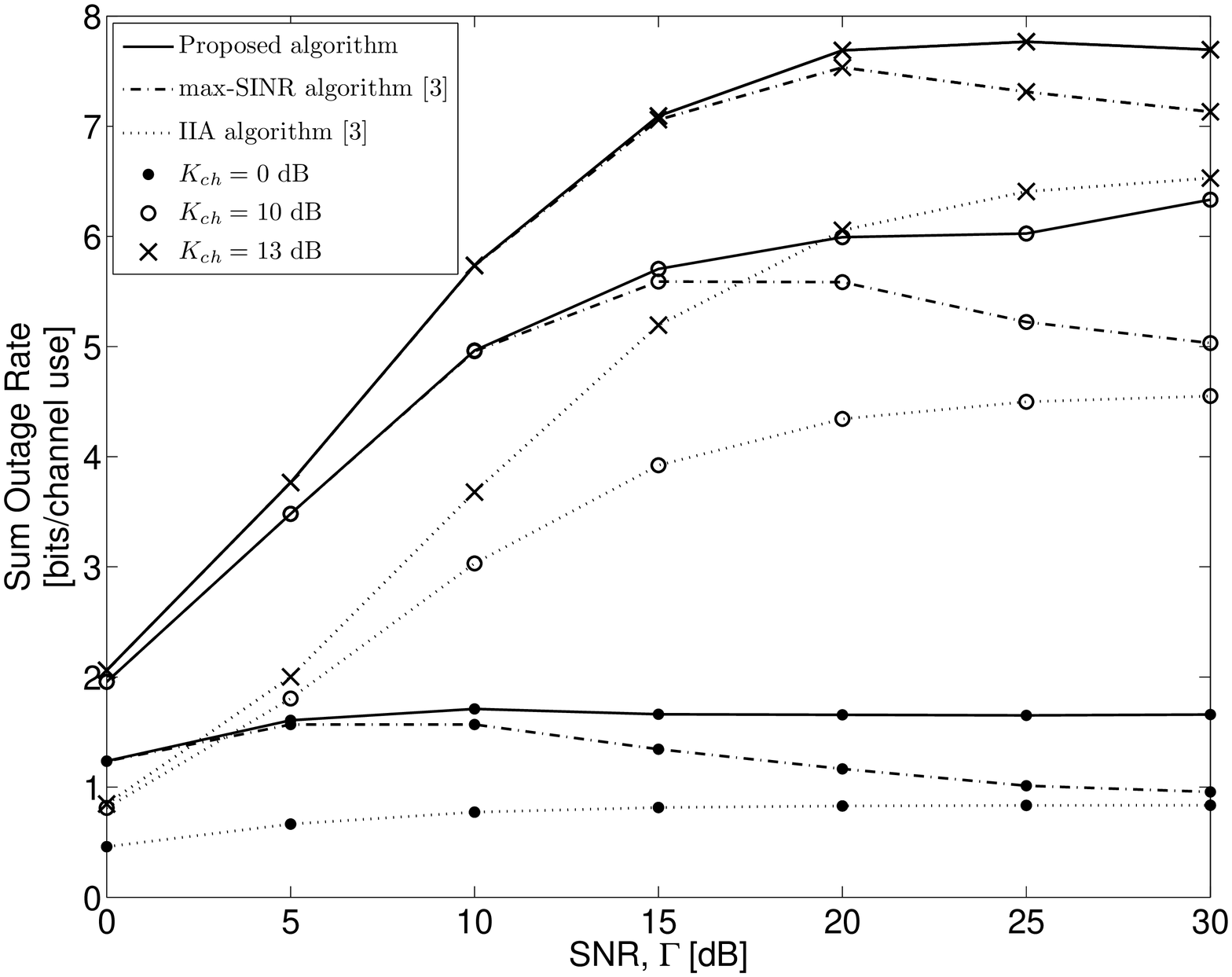}
\caption{Sum $\epsilon$-outage rate for $\epsilon=0.2$ ($K=3$,
$N_t=N_r=2d=2$, $\Sigmabf_t=\Sigmabf_r=\Ibf$)}
\label{fig:outage_capacity-e20}
\end{figure}

\section{Conclusion}
\label{sec:conclusion}

 In this paper, we have considered the outate probability and the outage-based beam design for  MIMO interference channels. We have derived
closed-form expressions for the outage probability in MIMO
interference channels under the assumption of Gaussian-distributed
CSI error, and have derived the asymptotic behavior of the outage
probability as a function of several system parameters based on
the Chernoff bound. We have shown that the outage probability
decreases exponentially w.r.t. the channel  $K$ factor defined as
the ratio of the power of the known channel part and that of the
unknown channel part. We have also provided an iterative beam
design algorithm for maximizing the sum outage rate based on the
derived outage probability expressions. Numerical results show
that the proposed beam design method significantly outperforms
conventional methods assuming perfect CSI in the sum outage rate
performance.


\appendices
\section{Proof of \eqref{eq:Xcovariance2}}
\label{append:proof_covariance}

 The $(p,q)$-th element of $\Sigmabf_{k,i}^{(m)}$ is given by
 \vspace{-0.5em}
\begin{eqnarray*}
& & \Ebb\{(X_{ki}^{(mp)}-\Ebb\{X_{ki}^{(mp)}\})
(X_{ki}^{(mq)}-\Ebb\{X_{ki}^{(mq)}\})^H\} \\
&=& \Ebb \{(\ubf_k^{(m)H}\Ebf_{ki}\vbf_i^{(p)})
(\ubf_k^{(m)H}\Ebf_{ki}\vbf_i^{(q)})^H\}  \\
&\stackrel{(a)}{=}&
\Ebb\{(\vbf_i^{(p)T}\otimes\ubf_k^{(m)H})\mbox{vec}(\Ebf_{ki})\mbox{vec}(\Ebf_{ki})^H
  (\vbf_i^{(q)T}\otimes\ubf_k^{(m)H})^H\} \\
&\stackrel{(b)}{=}& \sigma_h^2(\vbf_i^{(p)T}\otimes\ubf_k^{(m)H})
(\Sigmabf_t^T \otimes \Sigmabf_r)
(\vbf_i^{(q)T}\otimes\ubf_k^{(m)H})^H \\
&\stackrel{(c)}{=}&
\sigma_h^2(\vbf_i^{(p)T}\Sigmabf_t^T \otimes \ubf_k^{(m)H}\Sigmabf_r)(\vbf_i^{(q)*}\otimes\ubf_k^{(m)}), \qquad \mbox{where}~ \vbf_i^{(q)*}=(\vbf_i^{(q)T})^H\\
&\stackrel{(d)}{=}&
\sigma_h^2(\vbf_i^{(p)T}\Sigmabf_t^T\vbf_i^{(q)*} \otimes
\ubf_k^{(m)H}\Sigmabf_r\ubf_k^{(m)}) \stackrel{(e)}{=}
\sigma_h^2(\vbf_i^{(q)H}\Sigmabf_t\vbf_i^{(p)})(\ubf_k^{(m)H}\Sigmabf_r\ubf_k^{(m)}).
\end{eqnarray*}
Here, (a) is obtained by applying
$\mbox{vec}(\Abf\Bbf\Cbf)=(\Cbf^T\otimes\Abf)\mbox{vec}(\Bbf)$ to
each of the two terms in the expectation, (b) is by
$\Ebb\{\mbox{vec}(\Ebf_{ki})\mbox{vec}(\Ebf_{ki})^H\}=\sigma_h^2(\Sigmabf_t^T\otimes\Sigmabf_r)$,
(c) and (d) are by
$(\Abf\otimes\Bbf)(\Cbf\otimes\Dbf)=(\Abf\Cbf\otimes\Bbf\Dbf)$,
and finally (e) is because
$\vbf_i^{(p)T}\Sigmabf_t^T\vbf_i^{(q)*}$ and
$\ubf_k^{(m)H}\Sigmabf_r\ubf_k^{(m)}$ are scalars.
\hfill{$\blacksquare$}


\section{Distribution of a Non-Central Gaussian Quadratic Form}
\label{append:distribution}

The contents in Appendices B and C are from the technical report  WISRL-2012-APR-1, KAIST, "A Study on the Series Expansion of Gaussian Quadratic Forms".

\subsection{Previous work and literature survey}

There exist extensive literature
 about the probability distribution and statistical properties
of a quadratic form of non-central (complex) Gaussian random
variables in the communications area and the probability and
statistics community. Through a literature survey, we found that
the main technique to compute the distribution of a central (or a
non-central) Gaussian quadratic form is based on series fitting,
which was concretely unified and developed by S. Kotz \cite{Kotz:67AMS-1,
Kotz:67AMS-2}, and most of other works are its variants, e.g.,
\cite{Nabar:05WC}. First, we briefly explain this series fitting method
here.

Consider a  Gaussian quadratic form  $\xbf^H \Bar{\Qbf} \xbf$,
where $\xbf\sim{\mathcal{CN}}(\mubf, \Sigmabf)$ with size $n$ and
$\Bar{\Qbf}=\Bar{\Qbf}^H$. The first step of the series fitting
method is to convert the non-central Gaussian quadratic form into
a linear combination of chi-square random variables:
\begin{equation}\label{eq:equivalent_form}
\xbf^H\Bar{\Qbf}\xbf = \sum_{i=1}^n\lambda_i |z_i+\delta_i|^2 =
\sum_{i=1}^n\lambda_i [{\mathrm{Re}}(z_i+\delta_i)^2 +
{\mathrm{Im}}(z_i+\delta_i)^2],
\end{equation}
where $z_i  \stackrel{independent}{\sim} {\mathcal{CN}}(0,2)$ for
$i=1,\cdots,n$,  and $\{\delta_i, \lambda_i\}$ are constants
determined by $\Bar{\Qbf}$, $\mubf$ and $\Sigmabf$. Note that
${\mathrm{Re}}(z_i)\sim{\mathcal{N}}(0,1)$ and
${\mathrm{Re}}(z_i)\sim{\mathcal{N}}(0,1)$.  Thus, the non-central
Gaussian quadratic form is equivalent to a weighted sum of
non-central Chi-square random variables of which moment generating
function (MGF) is {\it known}. The MGF of a weighted sum of $n$
independent non-central $\chi^2$ random variables with degrees of
freedom $2m_i$ and non-centrality parameter $\mu_i^2$ is given by
\begin{equation}\label{eq:mgf}
\Phi(s) = \exp\Big\{-\frac{1}{2}\sum_{i=1}^{n}\mu_i^2 +
\frac{1}{2}\sum_{i=1}^n \frac{\mu_i^2}{1-2\lambda_i s}\Big\} \cdot
\prod_{i=1}^n \frac{1}{(1-2\lambda_i s)^{m_i}}.
\end{equation}
\begin{figure}[!ht]
\centerline{ \scalefig{0.6}\epsfbox{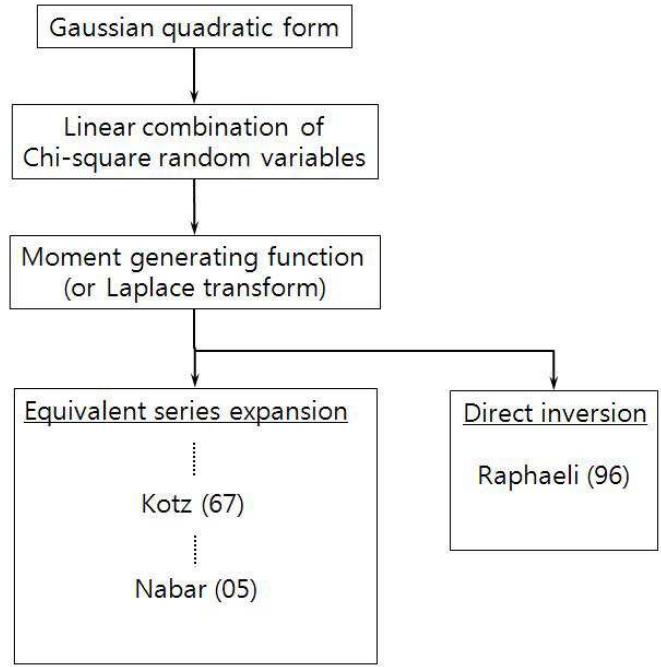} }
\caption{Computation of the distribution of a Gaussian quadratic
form} \label{fig:QuadFormTree}
\end{figure}
  Note here
that $\Phi(-s)$ is nothing but the Laplace transform of the {\it
probability density function (PDF)} of $\xbf^H\Bar{\Qbf}\xbf$ or
equivalently $\sum_{i=1}^n\lambda_i |z_i+\delta_i|^2$. Now, the
series fitting method expresses the PDF as an infinite series
composed of a set of known basis functions and tries to find the
linear combination coefficients so that the Laplace transform of
this series is the same as the known $\Phi(-s)$. Specifically, let
the PDF be
\begin{equation}
g_n(\Bar{\Qbf}, \mubf, \Sigmabf; y) = \sum_{k=0}^\infty c_k
h_k(y),
\end{equation}
where $\{h_k(y), k=0, 1, \cdots\}$ is the set of known basis
functions and  $\{c_k, k=0,1, \cdots\}$ is the set of linear
combination coefficients to be determined. Here, to make the
problem tractable, in most cases, the following conditions are
imposed. First, the sequence $\{h_k(y)\}$  of basis functions is
chosen among measurable complex-valued functions on $[0,\infty]$
such that
\begin{equation}
\sum_{k=0}^{\infty}|c_k||h_k(y)| \le Ae^{by}, \qquad   y \in
[0,\infty]  ~~\mbox{almost everywhere},
\end{equation}
where $A$ and $b$ are real constants. Second, the Laplace
transform  $\hat{h}_k(s)$ of $h_k(y)$ has a special form:
\begin{equation}
\hat{h}_k(s) = \xi(s) \eta^k(s),
\end{equation}
where $\xi(s)$ is a non-vanishing, analytic function for
${\mathrm{Re}}(s)>b$, and $\eta(s)$ is analytic for
${\mathrm{Re}}(s)>b$ and has an inverse function. The first
condition is for the existence of Laplace transform and the second
condition is to make the problem  tractable.  Finally, with the
pre-determined $\{h_k(y)\}$ with the conditions, the coefficients
$\{c_k\}$ are computed so that
\begin{equation}
{\mathcal{L}}({g}_n(\Bar{\Qbf},\mubf,\Sigmabf;y))
=\sum_{k=0}^\infty c_k \hat{h}_k(s)= \Phi(-s),
\end{equation}
where ${\mathcal{L}}(\cdot)$ denote the Laplace transform of a
function.

Widely used $\{h_k(y)\}$ for the series expansion of the {PDF} of
a quadratic form of non-central Gaussian random variables is as
follows \cite{Kotz:67AMS-1, Kotz:67AMS-2}.
\begin{enumerate}
\item[] 1. (Power series): $h_k(y) = (-1)^k
\frac{(y/2)^{n/2+k-1}}{2\Gamma(n/2+k)}$. %
\item[] 2. (Laguerre polynomials):
\begin{equation} \label{eq:LagReply}
h_k(y) = g(n;y/\beta)[k!\frac{\Gamma(n/2)}{\beta\Gamma(n/2+k)}]
L_k^{(n/2-1)}(y/2\beta),
\end{equation}
 where $g(n;y)$ is the central $\chi^2$ density with $n$ degrees
of freedom and $L_k^{(n/2-1)}(x)$ is the generalized Laguerre
polynomial defined by Rodriges' formula
\[
L_k^{(n/2-1)}(x) = \frac{1}{k!}e^x
x^{-(n/2-1)}\frac{d^k}{dx^k}e^{-x} x^{k+1}
\]
for $a>1$ and a positive control parameter $\beta$.
\end{enumerate}
For the detail computation of $\{c_k\}$, please refer to
\cite{Kotz:67AMS-1, Kotz:67AMS-2, Mathai&Provost:book}. The whole procedure is summarized in Fig. \ref{fig:QuadFormTree}.

\textbf{Reference group 1}

\begin{enumerate}
\item[] [Kotz-67a] S. Kotz, N. L. Johnson, and D. W. Boyd,
``Series representation of distributions of quadratic forms in
normal variables. I. Central Case,'' \textit{Ann. Math. Statist.,}
vol, 38, pp. 823 -- 837, Jun. 1967.

\item[] [Kotz-67b] S. Kotz, N.
L. Johnson, and D. W. Boyd, ``Series representation of
distributions of quadratic forms in normal variables. II.
Non-central Case,'' \textit{Ann. Math. Statist.,} vol. 38, pp. 838
-- 848, Jun. 1967.
\item[] [Mathai-92] A. M. Mathai and S. B. Provost,
\textit{Quadratic forms in random variables: Theory and
applications}, New York:M. Dekker, 1992.
\item[] [Nabar-05] R. Nabar, H. Bolcskei, and A. Paulraj,
``Diversity and Outage Performance of Space-Time Block Coded
Ricean MIMO Channels'', \textit{IEEE Trans. on Wireless Commun.},
vol. 4, no. 5, Sept. 2005.
\end{enumerate}

\textbf{Reference group 2}

\begin{enumerate}
\item[] [Pachares-55]  J. Pachares, ``Note on the distribution of
a definite quadratic form,'' \textit{Ann. Math. Statist.,} vol.
26, pp. 128 -- 131, Mar. 1955. $\Rightarrow$ Power series
representation of quadratic form of central Gaussian random
variables.
%
%
\item[] [Shah-61]  B. K. Shah and C. G. Khatri, ``Distribution of
a definite quadratic form for non-central normal variates,''
\textit{Ann. Math. Statist.,} vol. 32, pp. 883 -- 887, Sep. 1961.
$\Rightarrow$ Power series representation of quadratic form of
non-central Gaussian random variables.
%
\item[] [Shah-63] B. K. Shah, ``Distribution of definite and of
indefinite quadratic forms from a non-central normal
distribution,'' \textit{Ann. Math. Statis.,} vol. 34, pp. 186 --
190, Mar. 1963. $\Rightarrow$ Extends [Gurland-55] to derive a
representation of quadratic form of non-central Gaussian random
vector with Laguerre polynomial. Double series of Laguerre
polynomials is required.
%
%
\item[] [Gurland-55]  J. Gurland, ``Distribution of definite and
indefinite quadratic forms,'' \textit{Ann. Math. Statist.,} vol.
26, pp. 122 -- 127, Jan. 1955. $\Rightarrow$ Provides a simple
representation of quadratic form of central Gaussian random vector
in Laguerre polynomial.
%
%
\item[] [Gurland-56] J. Gurland, ``Quadratic forms in normally
distributed random variables,'' \textit{Sankhya: The Indian
Journal of Statistics} vol. 17, pp. 37 -- 50, Jan. 1956.
$\Rightarrow$ CDF for the indefinite quadratic form of central
random variable.
%
%
%
\item[] [Ruben-63] H. Ruben, ``A new result on the distribution of
quadratic forms,'' \textit{Ann. Math. Statist.,} vol. 34, pp. 1582
-- 1584, Dec. 1963. $\Rightarrow$ Represents the CDF of quadratic
form of central and non-central Gaussian random vector with
central/non-central $\chi^2$ distribution function.
%
\item[] [Tiku-65] M. L. Tiku, ``Laguerre series forms of non
central $\chi^2$ and $F$ distributions,'' \textit{Biometrika,}
vol. 52, pp. 415 -- 427, Dec. 1965. $\Rightarrow$ Another series
representaion with Laguerre polynomials.
%
%
%

\item[] [Davis-77] A. W. Davis, ``A differential equation approach
to linear combinations of independent chi-squares,'' \textit{J. of
the Ame. Statist. Assoc.} vol. 72, pp. 212 -- 214, Mar. 1977.
$\Rightarrow$ Provides another series representation with power series.
%
%
\item[] [Imhof-61] J. P. Imhof, ``Computing the distribution of
quadratic forms in normal variables,'' \textit{Biometrika} vol.
48, pp. 419 -- 426, Dec. 1961. $\Rightarrow$ Provides a numerical
method of computing the distribution
%
%
\item[] [Rice-80] S. O. Rice, ``Distribution of quadratic forms in
normal variables - Evaluation by numerical integration,''
\textit{SIAM J. Scient. Statist. Comput.,} vol. 1, no. 4, pp. 438
-- 448, 1980. $\Rightarrow$ Another numerical method of computing
distribution.
%
%
\item[] [Biyari-93] K. H. Biyari and W. C. Lindsey, ``Statistical
distribution of Hermitian quadratic forms in complex Gaussian
variables,'' \textit{IEEE Trans. Inform. Theory,} vol. 39, pp.
1076 -- 1082, Mar. 1993. $\Rightarrow$ Series expansion of
multi-variate complex Gaussian random variables. This paper deals
with the case that the Hermitian matrix in the quadratic form is a
special block-diagonal matrix.
\end{enumerate}

 \textbf{Reference group 3}
\begin{enumerate}
\item[] [Raphaeli-96] D. Raphaeli, ``Distribution of noncentral
indefinite quadratic forms in complex normal variables,''
\textit{IEEE Trans. Inf. Theory,} vol. 42, pp. 1002 -- 1007, May
1996.
\item[] [Al-Naffouri-09] T. Al-Naffouri and B. Hassibi, ``On the
distribution of indefinite quadratic forms in Gaussian random
variables,'' in \textit{Proc. of IEEE Int. Symp. Inf. Theory},
(Seoul, Korea), Jun.--Jul. 2009.

\end{enumerate}

\subsection{The difference of our work from the previous
works}

First, let us remind our outage event in MIMO interference
channels. From equations \eqref{eq:outageDef}, \eqref{eq:SINR_rearranged} and \eqref{eq:Xklmj}, we
have
\begin{equation}
\mbox{Pr}\{\mbox{outage}\}=\mbox{Pr}  \left\{\sum_{i=1}^{K}
\sum_{j=1}^{d} X_{ki}^{(mj)H}X_{ki}^{(mj)} \ge
\frac{|\ubf_k^{(m)H}\hat{\Hbf}_{kk}\vbf_k^{(m)}|^2}{2^{R_k^{(m)}}-1}-\sigma^2
=: \tau\right\},
\end{equation}
where $X_{ki}^{(mj)}$ is a non zero-mean Gaussian random variable.
Note that the outage probability is an {\it upper} tail
probability of the distribution of the Gaussian quadratic form
$\sum_{j=1}^{d} X_{ki}^{(mj)H}X_{ki}^{(mj)}$.   However, as seen
in Fig.  \ref{fig:cdf+mse+1},  {\it the most widely-used series
fitting method explained in the previous subsection yields a good
approximation of the distribution at the lower tail not at the
upper tail.}  The discrepancy between the series and the true
{PDF} is large at the upper\footnote{In the case of the problem
considered in \cite{Nabar:05WC}, the outage defined in \cite{Nabar:05WC} is
associated with the lower tail of the distribution and thus the
series fitting method is well suited to that case. However, our
system setup and considered problem are different from those in
\cite{Nabar:05WC}.} tail for a truncated series.
 {\it On the other hand, our approach
yields a good approximation to the true distribution at the upper
tail.} Thus,  the proposed series is more relevant to our problem
than the series fitting method.
\begin{figure}[!ht]
\centerline{ \SetLabels
\L(0.25*-0.1) (a) \\
\L(0.76*-0.1) (b) \\
\endSetLabels
\leavevmode
\strut\AffixLabels{ \scalefig{0.52}\epsfbox{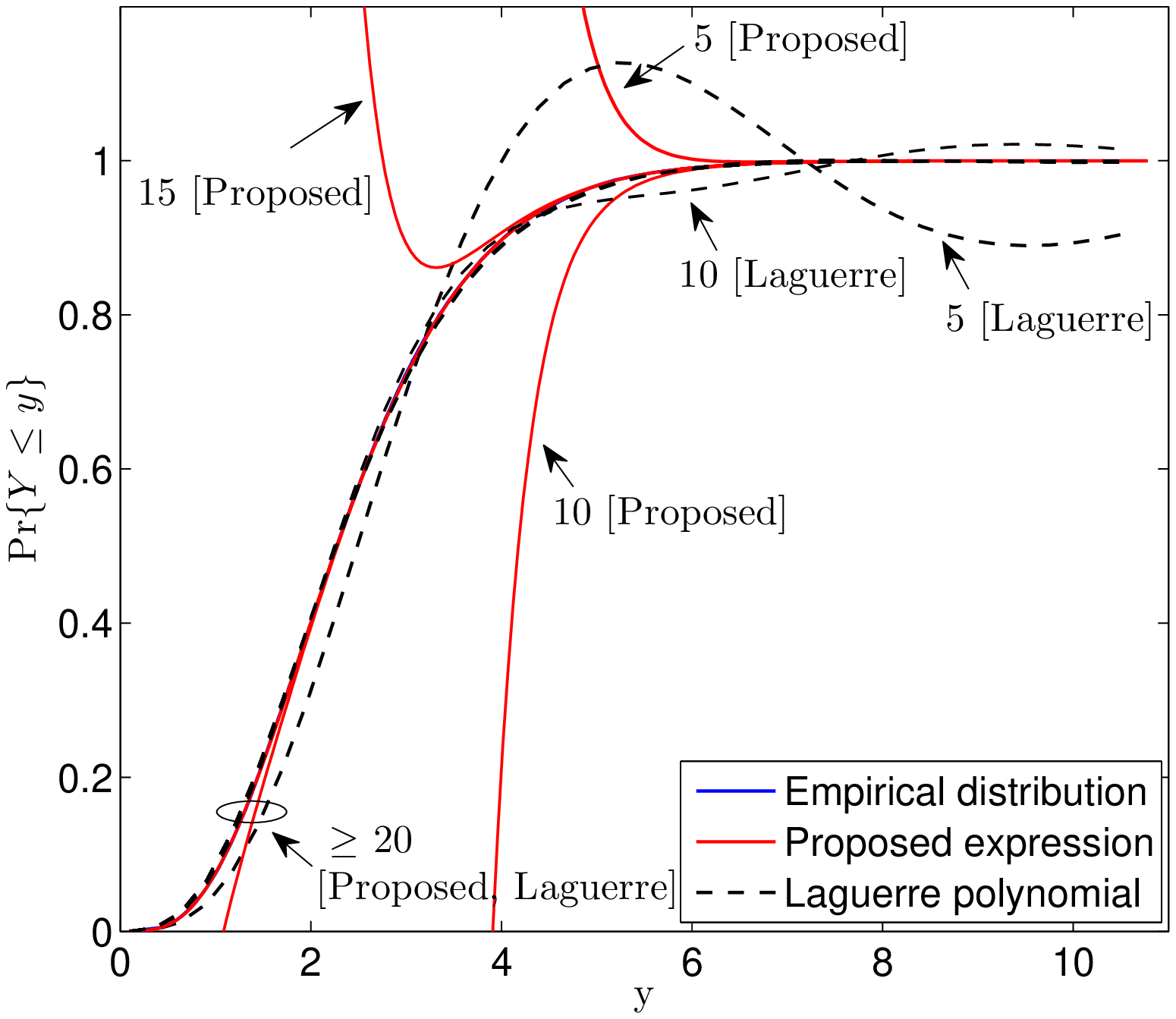}
\scalefig{0.52}\epsfbox{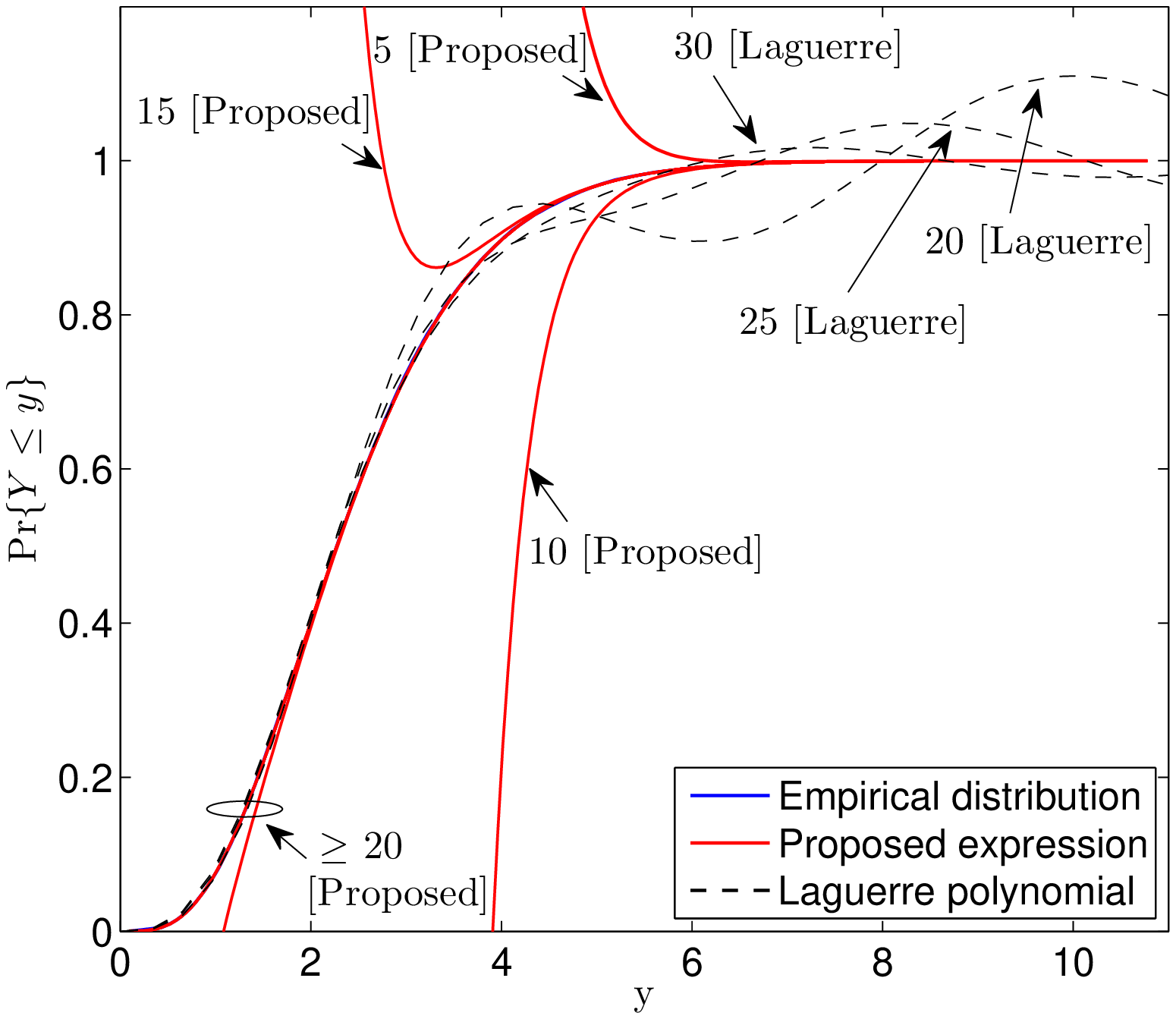} } } \vspace{1em}
\caption{Series fitting method versus direct inverse Laplace
transform method: number of variables = 4,
{$\mubf=0.5\mathbf{1}$}, $\Bar{\Qbf}=[1, 0.5, 0, 0; 0.5, 1, 0, 0;
0, 0, 1, 0; 0, 0, 0, 1]$, and $\Sigmabf=0.3\Ibf$.  (a) $\beta=1$
and (b) $\beta=2$. ($\beta$ is the control parameter for the
Laguerre polynomials in \eqref{eq:LagReply}.) Note that the
convergent speed of the series fitting method based on the
Laguerre polynomials depends much on $\beta$. In the case of
$\beta=2$, the series fitting method based on the Laguerre
polynomials yields large errors at the upper tail. It is not
simple how to choose $\beta$ and an efficient method is not known.
(One cannot run simulations for empirical distributions for all
cases.) The series fitting method based on the power series shows
bad performance, and it cannot be used in practice.}
\label{fig:cdf+mse+1}
\end{figure}

Our approach to the upper tail approximation is based on the
recent works   by Raphaeli \cite{Raphaeli:96IT} and by Al-Naffouri and
Hassibi \cite{AlNaffouri&Hassibi:09ISIT}. First, let us explain Raphaeli's method. The procedure in Fig. \ref{fig:QuadFormTree} up to obtaining the
MGF of the Gaussian quadratic form is common to both the sequence
fitting method and Raphaeli's method. However, Raphaeli's method
obtains the PDF by direct inverse Laplace transform of the MGF
$\Phi(s)$. Typically, the inverse Laplace transform of the MGF is
represented as a complex contour integral and then the complex
contour integral is computed as an infinite series by the residue
theorem. However, to obtain the cumulative distribution function
(CDF), which is actually necessary to compute the tail
probability, Raphaeli's method requires one more step, the
integration of the PDF, to obtain the CDF since the MGF $\Phi(s)$
is the Laplace transform of the {\em PDF}.

To obtain the CDF of  a general Gaussian
quadratic form, we did not use the MGF $\Phi(s)$, which is a bit
complicated and requires an additional step, like Raphaeli, but
instead we directly used a simple contour integral for the CDF
\eqref{eq:CDF_general14}, obtained by Al-Naffouri and Hassibi
\cite{AlNaffouri&Hassibi:09ISIT}.\footnote{In \cite{AlNaffouri&Hassibi:09ISIT}, Al-Naffouri and
Hassibi obtained the contour integral, \eqref{eq:CDF_general14}
for the CDF of a Gaussian quadratic form. However, they did not
obtain closed-form series expressions for the contour integral in
general cases except a few simple cases.
 The main goal of \cite{AlNaffouri&Hassibi:09ISIT} was to derive a nice and simple contour integral form for the CDF. } Then, the contour integral was computed as an infinite series by the residue theorem. (Using the residue theorem is borrowed from Raphaeli's work.) Thus, our result is simpler
than Raphaeli's approach and does not require the integration of a
PDF for the CDF.

As mentioned already, the series expansion in this paper has a
particular advantage over the series fitting method considered in
\cite{Nabar:05WC} for the outage event defined in this paper; The
series in this paper fits the upper tail of the distribution
well with a few number of terms.  We shall provide a detailed
proof for this in a special case in the next subsection. Thus, our
series expressions for outage probability in MIMO interference
channels are meaningful and relevant.

\section{Computational Issues and Convergence of the Obtained Series}
\label{append:computational}

\subsection{Computing higher order derivatives}
\label{append:derivative}

The general outage expression in Theorem 1 is given by
\begin{eqnarray}
{\mathrm{Pr}}\{{\mbox{outage}}\}
&=&{\mathrm{Pr}}\{\log_2(1+{\mathsf{SINR}}_k^{(m)}) \le R_k^{(m)}\} \nonumber\\
&=& -\sum_{i=1}^{\kappa} \frac{e^{-(\frac{\tau}{\lambda_i} +
\sum_{j=1}^{\kappa_i}|\chi_i^{(j)}|^2)} }{\lambda_i^{\kappa_i}}
\sum_{n = \kappa_i-1}^{\infty} \frac{1}{n!} g_{i}^{(n)}(0)
\frac{1}{(n-\kappa_i+1)!}
\left(\frac{\sum_{j=1}^{\kappa_i}|\chi_i^{(j)}|^2}{\lambda_i}\right)^{\!\!
n-\kappa_i+1},  \label{eq:appendix_general_outage}
\end{eqnarray}
where
\begin{equation} \label{eq:general_outagegis}
g_i(s) = \frac{e^{\tau s}}{s-1/\lambda_i} \cdot
\frac{\exp\left(-\sum_{p\neq i}
\frac{(s-1/\lambda_i)\lambda_p}{1+(s-1/\lambda_i)\lambda_p}
\sum_{q=1}^{\kappa_p}|\chi_p^{(q)}|^2\right)} {\prod_{p\neq i}
\Big(1+(s-1/\lambda_i)\lambda_p \Big)^{\kappa_p}}.
\end{equation}
To compute \eqref{eq:appendix_general_outage}, we need to compute
\begin{itemize}
\item $\{\lambda_i\}$ (the eigenvalues of the $Kd \times Kd$
covariance matrix $\Sigmabf = \Psibf \Lambdabf \Psibf^H$),

\item $\{\chi_i^{(j)}\}$ (the elements of $Kd$ vector $\chibf =
\Lambdabf^{-1/2}\Psibf^H \mubf$, where $\mubf$ is the mean vector
of the Gaussian distribution),

\item and the higher order derivatives of $g_i(s)$.

\end{itemize}
The computation of $\{\lambda_i\}$ and $\{\chi_i^{(j)}\}$ is
simple since the sizes of the mean vector and the covariance
matrix are $Kd$ and $Kd \times Kd$, respectively.  Furthermore,
the higher order derivatives of $g_i(s)$ can also be computed
efficiently based on recursion \cite{Mathai&Provost:book},\cite{Raphaeli:96IT}. Note
that $g_i(s) = e^{\log g_i(s)}$. Thus, the derivative of $g_i(s)$
can be written as
\begin{eqnarray} \label{eq:recursive_dev}
g_i^{(1)}(s) &=& g_i(s)[\log g_i(s)]^{(1)}, \nonumber \\
g_i^{(2)}(s) &=&
g_i^{(1)}(s)[\log g_i(s)]^{(1)} + g_i(s)[\log g_i(s)]^{(2)},  \nonumber \\
&\vdots&  \nonumber \\
g_i^{(n)}(s) &=& \sum_{l=0}^{n-1} \binom{n-1}{l} g_i^{(l)}(s)[\log
g_i(s)]^{(n-l)}, \quad\quad n\ge 1
\end{eqnarray}
where $g_i^{(l)}(s)$ and $[\log g_i(s)]^{(l)}$ denote the $l$-th
derivatives of $g_i(s)$ and $\log g_i(s)$, respectively. Here,
$[\log g_i(s)]^{(n)}$ can be computed from
\eqref{eq:general_outagegis} as
\[
[\log g_i(s)]^{(n)} = \tau\delta_{1n}
-\frac{(n-1)!(-1)^{n-1}}{(s-1/\lambda_i)^n} -\sum_{p\neq i}
\frac{n!(-1)^{n-1}\lambda_p^n}{(1+\lambda_p(s-1/\lambda_i))^{n+1}}\sum_{q=1}^{\kappa_p}|\chi_p^{(q)}|^2
-\sum_{p\neq i}
\frac{(n-1)!(-1)^{n-1}\kappa_p\lambda_p^n}{(1+\lambda_p(s-1/\lambda_i))^n}
\]
where $\delta_{1n}$ is Kronecker delta function. Thus, for given
$g_i(s)$ and $[\log g_i(s)]^{(l)}$, we can compute $g_i^{(l)}(s)$
efficiently in a recursive way, as shown in
\eqref{eq:recursive_dev}.

\subsection{Convergence analysis}
\label{append:convergence}

In this subsection, we provide some convergence analysis on the
derived series expansion in Sec. \ref{sec:outage_prob}. Consider the general
result in Theorem 1 for the CDF of a Gaussian quadratic form:
\begin{equation}\label{eq:cdf}
{\mathrm{Pr}}\{Y\le y\} = 1+\sum_{i=1}^\kappa
\frac{e^{-(\frac{y}{\lambda_i}+\sum_{j=1}^{\kappa_i}|\chi_i^{(j)}|^2)}}{\lambda_i^{\kappa_i}}
\sum_{n=\kappa_i-1}^{\infty} \frac{1}{n!} g_i^{(n)}(0, y)
\frac{1}{(n-\kappa_i+1)!} \left(
\frac{\sum_{j=1}^{\kappa_i}|\chi_i^{(j)}|^2}{\lambda_i}
\right)^{n-\kappa_i+1}
\end{equation}
where
\[
g_i(s, y) = \frac{e^{sy}}{s-\lambda_i^{-1}} \cdot
\frac{\exp\left(-\sum_{ p\neq i } \frac{(s-1/\lambda_i)\lambda_p}
        {1+(s-1/\lambda_i)\lambda_p }\sum_{q=1}^{\kappa_p}|\chi_p^{(q)}|^2 \right)}
{\prod_{p \neq i}
\Big(1+(s-1/\lambda_i)\lambda_p\Big)^{\kappa_p}}.
\]
Here, we explicitly use the variable $y$ as an input parameter of
the function $g_i(s)$ for later explanation. $g_i^{(n)}(s,y)$
denotes the $n$-th partial derivative of $g_i(s,y)$ with respect
to $s$. (Here, $\kappa$ is the number of distinct eigenvalues of
the $Kd \times Kd$ covariance matrix $\Sigmabf$ and $\kappa_i$ is
the geometric order of eigenvalue $\lambda_i$. $\sum_{i=1}^\kappa
\kappa_i = Kd$.) The residual error caused by truncating the
infinite series after the first $N$ terms is given by
\begin{equation}
R_N(y) = \sum_{i=1}^\kappa
\frac{e^{-(\frac{y}{\lambda_i}+\sum_{j=1}^{\kappa_i}|\chi_i^{(j)}|^2)}}{\lambda_i^{\kappa_i}}
\sum_{n=N+1}^{\infty} \frac{1}{n!} g_i^{(n)}(0, y)
\frac{1}{(n-\kappa_i+1)!} \left(
\frac{\sum_{j=1}^{\kappa_i}|\chi_i^{(j)}|^2}{\lambda_i}
\right)^{n-\kappa_i+1},
\end{equation}
and  we have
\[
{\mathrm Pr}\{Y \le y; \mbox{infinite sum}\} = {\mathrm Pr}\{Y\le
y; \mbox{truncation at $N$}\}+R_N(y).
\]
The truncation error $R_N(y)$ can be expressed as
\begin{equation}
R_N(y) = \sum_{i=1}^\kappa R_N^{i}(y),
\end{equation}
where
\begin{equation}
R_N^{i}(y) =
\frac{e^{-(\frac{y}{\lambda_i}+\sum_{j=1}^{\kappa_i}|\chi_i^{(j)}|^2)}}{\lambda_i^{\kappa_i}}
\sum_{n=N+1}^{\infty} \frac{1}{n!} g_i^{(n)}(0, y)
\frac{1}{(n-\kappa_i+1)!} \left(
\frac{\sum_{j=1}^{\kappa_i}|\chi_i^{(j)}|^2}{\lambda_i}
\right)^{n-\kappa_i+1}
\end{equation}
for each $1\le i \le\kappa$.  Then, the magnitude of each term
$|R_N^i(y)|$ in the truncation error is bounded as
\begin{align}
|R_N^i(y)| &\le \frac{1}{\lambda_i^{\kappa_i}}
\exp\bigg\{-\bigg(\frac{y}{\lambda_i}+\sum_{j=1}^{\kappa_i}|\chi_i^{(j)}|^2\bigg)\bigg\}
\cdot \sum_{n=N+1}^{\infty} \frac{1}{n!} \left| g_i^{(n)}(0,y)
\right| \cdot \frac{1}{(n-\kappa_i+1)!} \left(
\frac{\sum_{j=1}^{\kappa_i}|\chi_i^{(j)}|^2}{\lambda_i}
\right)^{n-\kappa_i+1}.
\end{align}

As seen in Fig. \ref{fig:cdf+mse+1}, our series expansion fits the
upper tail distribution first. Now, to assess the overall
convergence speed of our series, for the same step as in Fig.
\ref{fig:cdf+mse+1}, we ran some simulations to obtain an
empirical distribution, and computed the overall mean square error
(MSE) between the truncated series and the empirical distribution
over {$0 \le y \le 10$} as
\[
\mbox{CDF MSE} = \frac{1}{200}\sum_{i=1}^{200} \Big|
{\mathrm{Pr}}\{Y\le y_i; N, \mbox{type of series} \} -
{\mathrm{Pr}}\{Y\le y_i; \mbox{empirical}\} \Big|^2,
\]
where $\{y_i\}$ are the uniform samples of {$[0, 10]$}. Fig.
\ref{fig:LagMSEfuncN} shows the CDF MSE of the three methods in
Fig. \ref{fig:cdf+mse+1}: the proposed series, the series fitting
method with $\beta=1$ and the series fitting method with $\beta
=2$.
\begin{figure}[!ht]
\centerline{ \scalefig{0.52}\epsfbox{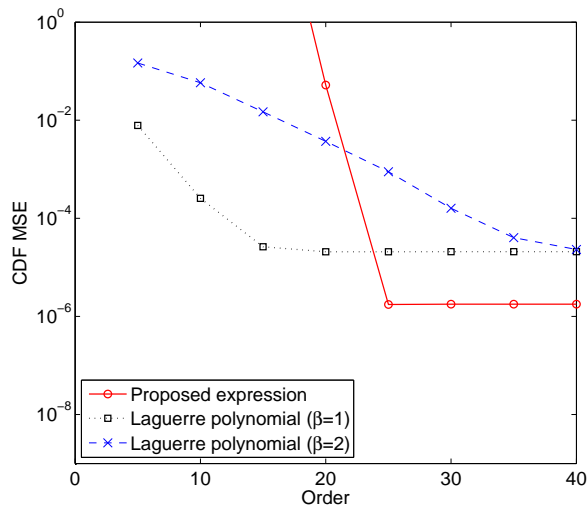} }
\caption{{CDF MSE of the CDFs in Fig. \ref{fig:cdf+mse+1}} }
\label{fig:LagMSEfuncN}
\end{figure}
\begin{figure}[!htbp]
\centerline{ \SetLabels
\L(0.25*-0.1) (a) \\
\L(0.76*-0.1) (b) \\
\endSetLabels
\leavevmode
\strut\AffixLabels{
\scalefig{0.52}\epsfbox{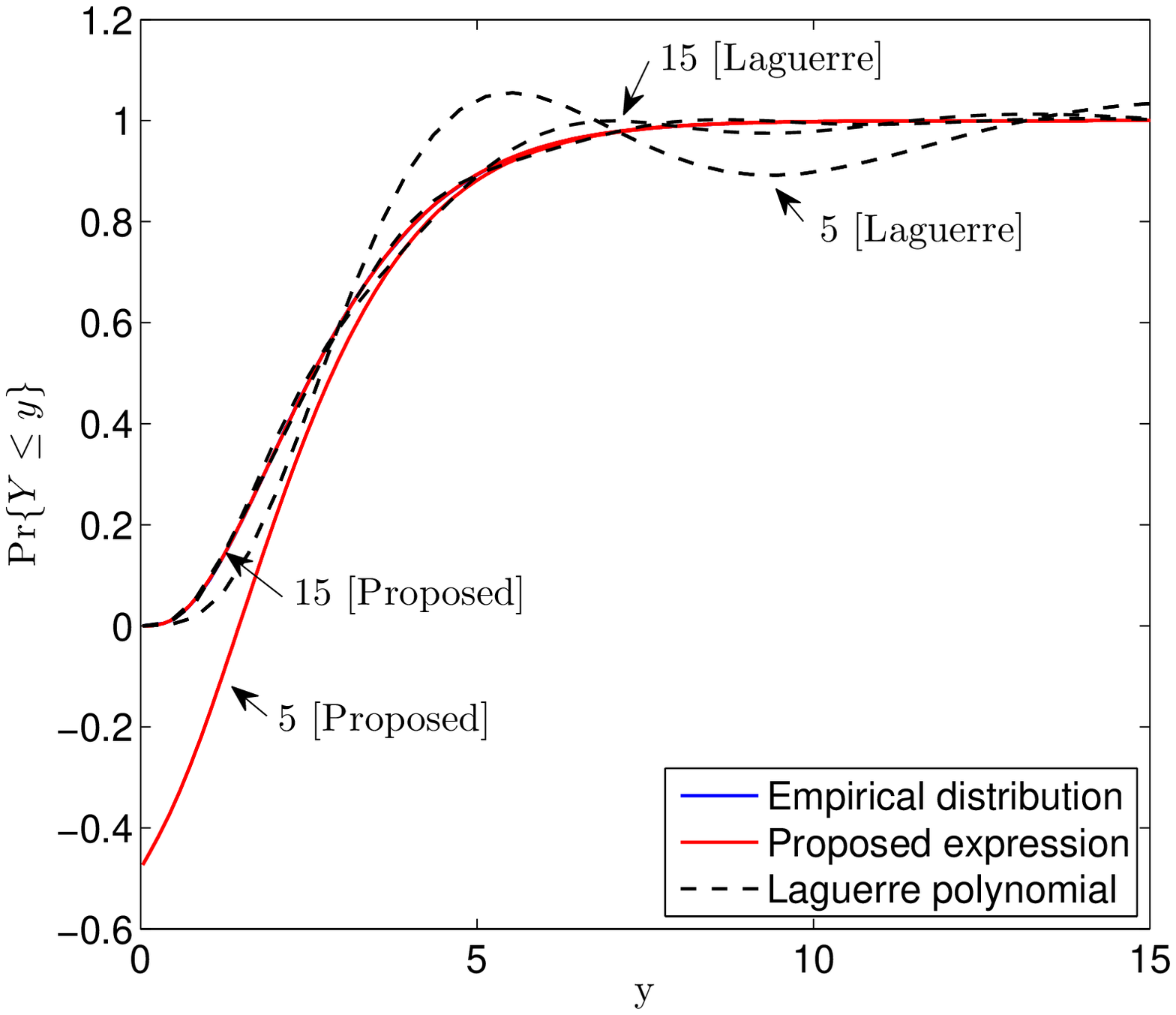}
\scalefig{0.52}\epsfbox{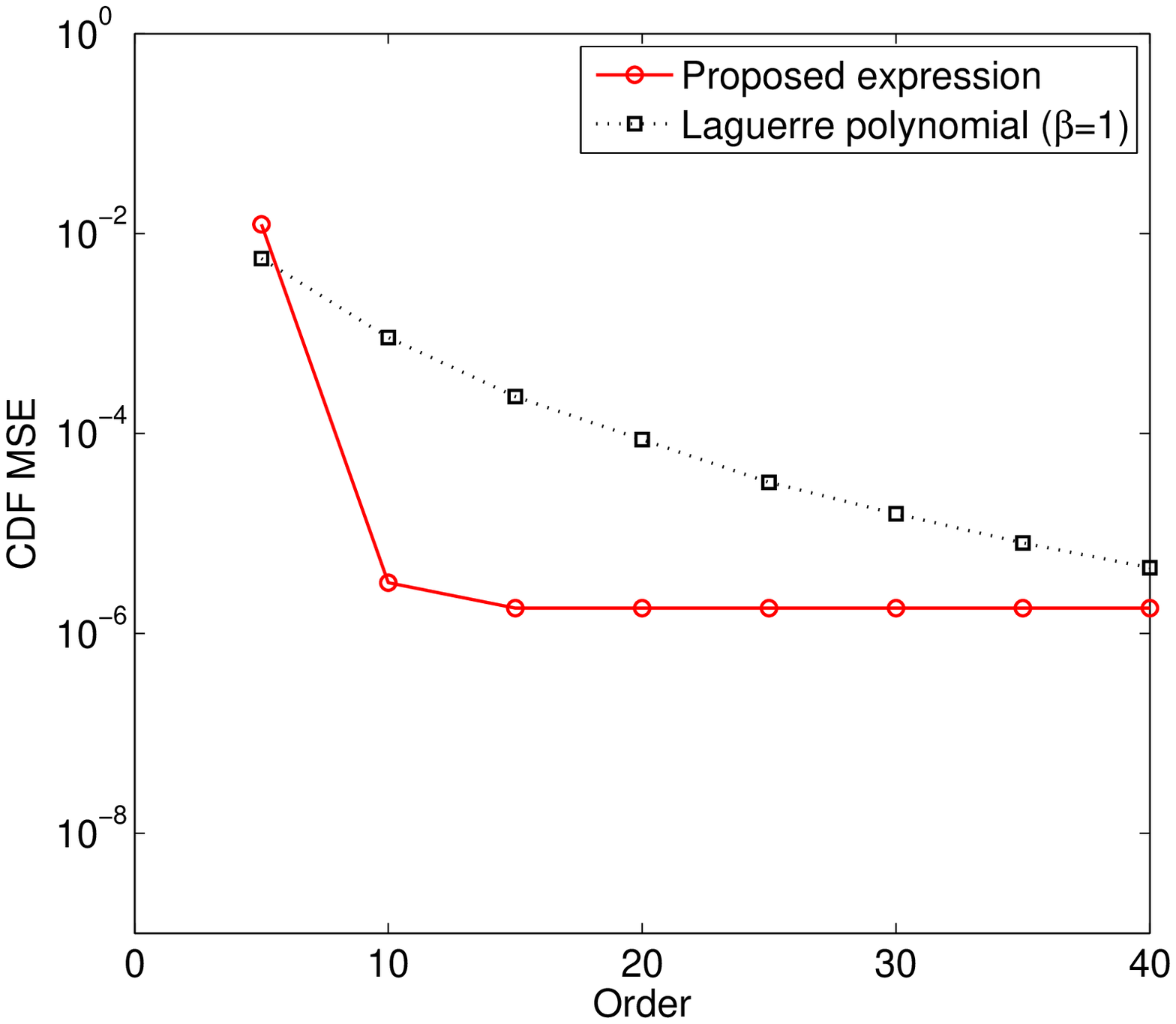} } } \vspace{1em}
\caption{number of variables = 4, {$\mubf=0.5\mathbf{1}$},
$\Bar{\Qbf}=\Ibf$, and $\Sigmabf= [0.2641\   0.0328\   0.1963\
0.1140;\
       0.0328\   0.6097\  -0.1739\    0.1708;\
       0.1963\  -0.1739\   0.8746\   -0.0022;\
       0.1140\   0.1708\  -0.0022\    0.1250] $. In this case eigenvalues are $1.0000$, $0.6318$, $0.2158$, and $0.0259$ with $\beta=1$. (a) CDF, (b) {CDF MSE. Uniform sample of $y$ is taken over $[0, 15.9]$.} }
\label{fig:cdf+mse+2}
\end{figure}
It is seen in Fig. \ref{fig:LagMSEfuncN} that the overall
convergence of the proposed series can be worse than the series
fitting method at the small values for the number of summation
terms for the setting in Fig. \ref{fig:cdf+mse+1}. The bad overall
convergence is due to  worse fitting at the lower tail of the
distribution, but the bad lower tail approximation is not
important to our outage computation. (Please see Fig.
\ref{fig:cdf+mse+1}.) Fig. \ref{fig:cdf+mse+2} shows another case.
In this case, the proposed series outperforms the series fitting
method both in the overall convergence and in the upper tail
convergence.  It is seen numerically that the proposed series fits
the upper tail distribution first. Now, we shall prove this
property of the proposed series. However, it is a difficult
problem to prove this property in general cases.  Thus, in the
next subsection, we provide a proof of this property when the
number of distinct eigenvalues of the covariance matrix $\Sigmabf$
is one, e.g., in the i.i.d. case.

\subsubsection{The identity covariance matrix case}

 Suppose that there is only one eigenvalue, $\lambda$ ($> 0$), with
multiplicity $\kappa$ for the covariance matrix $\Sigmabf$. This
case corresponds to Corollary 4, and the outage probability is given by
\begin{equation} \label{eq:replyLastCDF}
{\mathrm{Pr}}\{Y\le y\} = 1 + \frac{\exp(-\eta^2)}{\lambda^\kappa}
\exp\left(-\frac{y}{\lambda}\right) \sum_{n=\kappa-1}^{\infty}
g^{(n)}(0, y)
\frac{(\eta^2/\lambda)^{n-\kappa+1}}{n!(n-\kappa+1)!},
\end{equation}
where
\begin{equation} \label{eq:replyLastgsy}
g(s, y)=\frac{e^{ys}}{s-\lambda^{-1}}
\end{equation}
and $\eta^2=\sum_{j=1}^\kappa|\chi^{(j)}|^2$. The residual error
caused by truncating the infinite series after the first $N$ terms
is given by
\begin{equation} \label{eq:residual_e1}
R_N(y) = \frac{\exp(-\eta^2)}{\lambda^\kappa}
\exp\left(-\frac{y}{\lambda}\right) \sum_{n=N+1}^{\infty}
g^{(n)}(0, y)
\frac{(\eta^2/\lambda)^{n-\kappa+1}}{n!(n-\kappa+1)!}.
\end{equation}

\noindent Before we proceed, we first obtain the $n$-th derivative
of $g(s,y)$ at $s=0$, which is given in the following lemma.

\vspace{0.5em}
\begin{lemma}
For $n\ge 0$,
\begin{eqnarray}\label{eq:general_g}
g^{(n)}(0, y) = -\lambda \sum_{k=0}^n \frac{n!}{(n-k)!} \lambda^k
y^{n-k}.
\end{eqnarray}
\end{lemma}

\vspace{1em} \noindent \textit{Proof:} Proof is given by
induction. The validity of the claim for $n=0, 1$ and $2$ is shown
by direction computation:
\begin{align*}
g^{(0)}(0, y) =& \frac{ye^{ys}}{s-1/\lambda}\bigg|_{s=0} =
-\lambda
= -\lambda \sum_{k=0}^0 \frac{0!}{(0-k)!} \lambda^k y^{0-k},  \\
g^{(1)}(0, y) =& \frac{ye^{ys}(s-1/\lambda)-e^{ys}}{(s-1/\lambda)^2}\bigg|_{s=0} =
-\lambda (y + \lambda)
= -\lambda \sum_{k=0}^1 \frac{1!}{(1-k)!} \lambda^k y^{1-k},  \\
g^{(2)}(0, y) =&
\frac{\left(ye^{ys}(ys-y/\lambda-1)+e^{ys}y\right)(s-\frac{1}{\lambda})^2
 -2e^{ys}(ys-y/\lambda-1)(s-\frac{1}{\lambda})}{(s-1/\lambda)^4}\bigg|_{s=0} \\
=&  -\lambda (y^2 + 2\lambda y + 2\lambda^2)
 =  -\lambda \sum_{k=0}^2 \frac{2!}{(2-k)!} \lambda^k y^{2-k}.
\end{align*}
Now, suppose that \eqref{eq:general_g} holds up to the $(n-1)$-th
derivative of $g(s, y)$. From the recursive formula in
\eqref{eq:recursive_dev}, $g^{(n)}(0,y)$ is obtained as
\begin{align}
& g^{(n)}(0, y) \nonumber \\
&= \sum_{k=0}^{n-1}
{n-1 \choose k}g^{(k)}(0, y)(\log g(0, y))^{(n-k)} \nonumber \\
&=
 {n-1 \choose 0}g^{(0)}(0, y)(\log g(0, y))^{(n)}
+{n-1 \choose 1}g^{(1)}(0, y)(\log g(0, y))^{(n-1)} +\cdots \nonumber \\
&~~~~ +{n-1 \choose n-1}g^{(n-1)}(0, y)(\log g(0, y))^{(1)}. \label{eq:g_n}
\end{align}
Since $[\log g(s)] = ys - \log(s-1/\lambda)$, we can easily see
that $[\log g(0)]^{(1)} = y+\lambda$ and $[\log g(0)]^{(n)} =
(n-1)!\lambda^n$ for $n\ge 2$. Therefore, \eqref{eq:g_n} can be
rewritten as
\begin{align}
g^{(n)}(0, y) =& (n-1)!g(0, y)\lambda^n + (n-1)g^{(1)}(0,
y)(n-2)!\lambda^{n-1}
+ {n-1 \choose 2}g^{(2)}(0, y)(n-3)!\lambda^{n-2}  + \cdots \nonumber \\
&{  + (n-1) g^{(n-2)}(0, y)\lambda^2 + g^{(n-1)}(0, y)(y+\lambda)} \nonumber \\
=& (n-1)!g(0, y)\lambda^n + (n-1)!g^{(1)}(0, y)\lambda^{n-1}
+ \frac{(n-1)!}{2} g^{(2)}(0, y) \lambda^{n-2}  + \cdots  \nonumber\\
&  + (n-1) g^{(n-2)}(0, y)\lambda^2 + \lambda g^{(n-1)}(0, y)
   + y g^{(n-1)}(0, y)  \nonumber\\
\stackrel{(a)}{=}& -\lambda\left[
\sum_{l=0}^{n-1}\frac{(n-1)!}{l!} \left( \sum_{k=0}^l
\frac{l!}{(l-k)!}\lambda^k y^{l-k} \right)\lambda^{n-l} {+
y\sum_{m=0}^{n-1}\frac{(n-1)!}{(n-m-1)!}\lambda^m y^{n-m-1}}
\right]  \nonumber\\
=& -\lambda\left[ \sum_{l=0}^{n-1}\frac{(n-1)!}{l!} \left(
\sum_{k=0}^l \frac{l!}{(l-k)!}\lambda^k y^{l-k}
\right)\lambda^{n-l} +
\sum_{m=0}^{n-1}\frac{(n-1)!}{(n-m-1)!}\lambda^m y^{n-m} \right]
\label{eq:candidate}
\end{align}
where (a) holds since \eqref{eq:general_g} holds for all
$g^{(0)}(0, y),\cdots, g^{(n-1)}(0, y)$ by the induction
assumption.

Here, consider the coefficient of each $y^i$ in
\eqref{eq:candidate} for $i=0,\cdots, n$.
\begin{enumerate}
\item[] i) $y^n$ is obtained only when $m=0$. The coefficient of
$y^{n}$ from \eqref{eq:candidate} is therefore given by
$-\lambda$. It corresponds to the coefficient of $y^n$ in
\eqref{eq:general_g}.
\item[] ii) For  $0 < p \le n$, the coefficient of $y^{n-p}$ is
obtained by considering  all $(l,k)$ that satisfies $l-k=n-p$ due
to the first term in the right-hand side (RHS) of
\eqref{eq:candidate}, and $m=p$ due to the second term of the RHS
of \eqref{eq:candidate}. In the first case, we obtain $y^{n-p}$
with the following pairs  $(l,k) = (n-1, p-1)$, $(n-2,p-2)$,
$\cdots$, $(n-p, 0)$. For these $(l,k)$ pairs, we have
\begin{equation}\label{eq:c1}
\hspace{-1em}
-\lambda\sum_{l=n-p}^{n-1} \frac{(n-1)!}{l!} \cdot
\left(\frac{l!}{(n-p)!} \lambda^{l-n+p}
y^{n-p}\right)\cdot\lambda^{n-l} = -\lambda\sum_{l=n-p}^{n-1}
\frac{(n-1)!}{(n-p)!} \lambda^{p}  y^{n-p} = -\lambda p
\frac{(n-1)!}{(n-p)!} \lambda^p  y^{n-p}.
\end{equation}
In the second case of $m=p$, we have
\begin{equation}\label{eq:c2}
-\lambda\frac{(n-1)!}{(n-p-1)!} \lambda^p y^{n-p}.
\end{equation}
Finally, the coefficient of $y^{n-q}$ is given by adding
\eqref{eq:c1} and \eqref{eq:c2}:
\begin{eqnarray*}
& & -\lambda\Bigg(
 \frac{(n-1)!}{(n-p-1)!}
+p \frac{(n-1)!}{(n-p)!} \Bigg) \lambda^{p} y^{n-p} \\
&=& -\lambda\frac{(n-1)!}{(n-p-1)!}
\Bigg(1+\frac{p}{n-p} \Bigg) \lambda^{p} y^{n-p} \\
&=& -\lambda\frac{n!}{(n-p)!}
 \lambda^p y^{n-p},
\end{eqnarray*}
which is equivalent to the  coefficient for $y^{n-p}$ in
\eqref{eq:general_g} ($0<p\le n$). Thus, \eqref{eq:general_g}
holds for $g^{(n)}(0,y)$.

\hfill$\blacksquare$
\end{enumerate}

\vspace{0.5em} \noindent Note that $g^{(n)}(0,y) < 0$ for all
$n\ge 0$ from \eqref{eq:general_g}. Therefore, $R_N(y)\le 0$ for
all $N$ and $y$ and $|g^{(n)}(0,y)|=-g^{(n)}(0,y)$.

Now, consider the residual error term $R_N(y)$ in
\eqref{eq:residual_e1}. The magnitude of the residual error can be
upper bounded as follows:
\begin{eqnarray} \label{eq:UB_single}
|R_N(y)| &=& \frac{\exp(-\eta^2)}{\lambda^\kappa} \cdot
\exp\left(-\frac{y}{\lambda}\right) \sum_{n=N+1}^{\infty}
|g^{(n)}(0, y)|
\frac{(\eta^2/\lambda)^{n-\kappa+1}}{n!(n-\kappa+1)!} \nonumber  \\
&=& \frac{\exp(-\eta^2)}{\lambda^\kappa} \cdot
\exp\left(-\frac{y}{\lambda}\right) \sum_{n=N+1}^{\infty}
(-g^{(n)}(0, y))
\frac{(\eta^2/\lambda)^{n-\kappa+1}}{n!(n-\kappa+1)!} \nonumber  \\
&=& -\frac{\exp(-\mu^2)}{\lambda^\kappa} \cdot
\exp\left(-\frac{y}{\lambda}\right) \sum_{n=N+1}^{\infty}
g^{(n)}(0, y)
\frac{(\eta^2/\lambda)^{n-\kappa+1}}{n!(n-\kappa+1)!} \nonumber  \\
&=& -\frac{\exp(-\eta^2)}{\lambda^\kappa} \cdot
\exp\left(-\frac{y}{\lambda}\right) \sum_{n=N+1}^{\infty}
\frac{1}{n!}g^{(n)}(0, y) \Big(\frac{1}{2\lambda}\Big)^n
\frac{(2\eta^2)^{n-\kappa+1}(2\lambda)^{\kappa-1} }{(n-\kappa+1)!} \nonumber  \\
&=& -(2\lambda)^{\kappa-1} \cdot
\frac{\exp(-\eta^2)}{\lambda^\kappa} \cdot
\exp\left(-\frac{y}{\lambda}\right) \sum_{n=N+1}^{\infty}
\frac{1}{n!}g^{(n)}(0, y) \Big(\frac{1}{2\lambda}\Big)^n
\frac{(2\eta^2)^{n-\kappa+1} }{(n-\kappa+1)!} \nonumber  \\
&\stackrel{(a)}{\le}& -(2\lambda)^{\kappa-1} \cdot
\frac{\exp(-\eta^2)}{\lambda^\kappa} \cdot
\exp\left(-\frac{y}{\lambda}\right) \sum_{n=N+1}^{\infty}
\frac{1}{n!}g^{(n)}(0, y)
\Big(\frac{1}{2\lambda}\Big)^n \exp(2\eta^2)  \nonumber  \\
&=& -\frac{2^{\kappa-1}}{\lambda} \exp(\eta^2) \cdot
\exp\left(-\frac{y}{\lambda}\right) \sum_{n=N+1}^{\infty}
\frac{1}{n!}g^{(n)}(0, y)
\Big(\frac{1}{2\lambda}\Big)^n  \nonumber  \\
&\stackrel{(b)}{\le}& -\frac{2^{\kappa-1}}{\lambda} \exp(\eta^2)
\cdot \exp\left(-\frac{y}{\lambda}\right)\ \cdot
\sum_{n=0}^{\infty} \frac{1}{n!}g^{(n)}(0, y)
\Big(\frac{1}{2\lambda}\Big)^n  \nonumber  \\
&\stackrel{(c)}{=}& -\frac{2^{\kappa-1}}{\lambda} \exp(\eta^2)
\cdot \exp\left(-\frac{y}{\lambda}\right)
\cdot g\left(\frac{1}{2\lambda}, y\right) \nonumber \\
&\stackrel{(d)}{=}& -\frac{2^{\kappa-1}}{\lambda} \exp(\eta^2)
\cdot \exp\left(-\frac{y}{\lambda}\right)
\cdot \frac{\exp(y/2\lambda)}{-1/2\lambda} \nonumber \\
&=& 2^{\kappa} \exp(\eta^2) \cdot
\exp\left(-\frac{y}{2\lambda}\right)
\end{eqnarray}
where $(a)$ is from $\frac{\gamma^k}{k!}\le
\exp(\gamma)=\sum_{p=0}^\infty \gamma^p/p!$ for any $\gamma>0$,
$(b)$ is from the fact that summand is {negative}, (c) is by using
the Taylor series expansion, and (d) is from
\eqref{eq:replyLastgsy}. Since $\eta$ is a fixed constant, from
\eqref{eq:UB_single},  for any $N\ge 0$
\begin{equation}
\lim_{y\to\infty} |R_N(y)| = 0.
\end{equation}
Thus, it is clear that the proposed series converges from the
upper tail distribution!

Now, let us consider the residual error magnitude as a function of
$y$ for given $N$.  From \eqref{eq:general_g}, we have
\begin{equation}
\frac{\partial g^{(n)}(0, y)}{\partial y} = n g^{(n-1)}(0, y).
\end{equation}
Differentiating $R_N(y)$ with respect to $y$ yields
\begin{align}
\frac{\partial R_N(y)}{\partial y} &=
\frac{\exp(-\eta^2)}{\lambda^\kappa}\left(-\frac{1}{\lambda}\right)
\exp\left(-\frac{y}{\lambda}\right) \sum_{n=N+1}^{\infty}
g^{(n)}(0, y)
\frac{(\eta^2/\lambda)^{n-\kappa+1}}{n!(n-\kappa+1)!} \nonumber \\
& \ \ + \frac{\exp(-\eta^2)}{\lambda^\kappa}
\exp\left(-\frac{y}{\lambda}\right) \sum_{n=N+1}^{\infty}
\frac{\partial g^{(n)}(0, y)}{\partial y} \cdot
\frac{(\eta^2/\lambda)^{n-\kappa+1}}{n!(n-\kappa+1)!} \nonumber \\
&= \frac{\exp(-\eta^2)}{\lambda^\kappa}
\exp\left(-\frac{y}{\lambda}\right) \sum_{n=N+1}^{\infty}
\frac{(\eta^2/\lambda)^{n-\kappa+1}}{n!(n-\kappa+1)!} \left(
-\frac{1}{\lambda} g^{(n)}(0, y) + n g^{(n-1)}(0, y) \right).
\label{eq:replyFinallambdagn0yabove}
\end{align}
Furthermore, from \eqref{eq:general_g} we have
\begin{equation} \label{eq:replyFinallambdagn0y}
-\frac{1}{\lambda}g^{(n)}(0,y)+n g^{(n-1)}(0,y) = y^n.
\end{equation}
By substituting \eqref{eq:replyFinallambdagn0y} into
\eqref{eq:replyFinallambdagn0yabove}, we have
\begin{eqnarray} \label{eq:diff_d1}
\frac{\partial R_N(y)}{\partial y} &=&
\frac{\exp(-\eta^2)}{\lambda^\kappa}
\exp\left(-\frac{y}{\lambda}\right) \sum_{n=N+1}^{\infty}
\frac{(\eta^2/\lambda)^{n-\kappa+1}y^n}{n!(n-\kappa+1)!},
\end{eqnarray}
which is positive.  Since $R_N(y) \le 0$,
$\lim_{y\rightarrow\infty} R_N(y) = 0$ and $\frac{\partial
R_N(y)}{\partial y} >0$, the residual error magnitude
monotonically decreases as $y$ increases and the maximum error
occurs at $y=0$ for any given $N$.

Now, let us compute the worst truncation error $R_N(0)$, which is
given by
\begin{equation}
R_N(0) =    \frac{\exp(-\eta^2)}{\lambda^\kappa}
\sum_{n=N+1}^{\infty} g^{(n)}(0,0)
\frac{(\eta^2/\lambda)^{n-\kappa+1}}{n!(n-\kappa+1)!}.
\end{equation}
From \eqref{eq:general_g}, we have $g^{(n)}(0,0)= -n!
\lambda^{n+1}$. Therefore,
\begin{align}
R_N(0) &= \frac{\exp(-\eta^2)}{\lambda^\kappa}
\sum_{n=N+1}^{\infty} (-n! \lambda^{n+1})
\frac{(\eta^2/\lambda)^{n-\kappa+1}}{n!(n-\kappa+1)!} \nonumber \\
&= -\frac{\exp(-\eta^2)}{\lambda^\kappa} \sum_{n=N+1}^{\infty}
\lambda^{n+1}
\frac{(\eta^2/\lambda)^{n-\kappa+1}}{(n-\kappa+1)!} \nonumber \\
&= -\frac{\exp(-\eta^2)}{\lambda^\kappa} \sum_{n=N+1}^{\infty}
\frac{(\eta^2)^{n-\kappa+1}}{(n-\kappa+1)!}\cdot\lambda^{\kappa} \nonumber \\
&= -\exp(-\eta^2) \sum_{n=N+1}^{\infty}
\frac{(\eta^2)^{n-\kappa+1}}{(n-\kappa+1)!}.
\end{align}
From \eqref{eq:replyLastCDF},  $N \ge \kappa-2$.   For general $N
\ge \kappa-2$, let $m=n-\kappa+1$. Then,
\[
R_N(0) = -\exp(-\eta^2) \sum_{m=N-\kappa+2}^{\infty}
\frac{(\eta^2)^m}{m!}.
\]
Note that $\sum_{m=N-\kappa+2}^{\infty} \frac{(\eta^2)^m}{m!}$ is
the residual error of the Taylor series expansion of $\exp(x)$
after the first $(N-\kappa+1)$ terms. By the Taylor theorem,
\begin{equation}
\sum_{m=N-\kappa+2}^{\infty}\frac{(\eta^2)^m}{m!} =
\frac{(\eta^2)^{N-\kappa+2}}{(N-\kappa+2)!} \exp(\alpha\eta^2)
\end{equation}
where some $\alpha \in [0, 1]$. Therefore, the worst truncation
error is given by
\begin{equation} \label{eq:bound_1}
|R_N(0)| = \exp\Big((\alpha-1)\eta^2\Big)\times
\frac{(\eta^2)^{N-\kappa+2}}{(N-\kappa+2)!} \le
\frac{(\eta^2)^{N-\kappa+2}}{(N-\kappa+2)!},
\end{equation}
where the inequality holds since $\exp((\alpha-1)\eta^2) \le 1$
for  $0\le \alpha \le 1$.  Furthermore, the residual error
magnitude is a strictly decreasing function of $N$ for any $y$,
\begin{equation} \label{eq:bound_2}
|R_N(y)| > |R_{N+1}(y)|.
\end{equation}
This can be shown easily as follows.
\begin{align*}
R_N(y) =& \frac{\exp(-\eta^2)}{\lambda^\kappa}
\exp\left(-\frac{y}{\lambda}\right) \sum_{n=N+1}^{\infty}
g^{(n)}(0, y)
\frac{(\eta^2/\lambda)^{n-\kappa+1}}{n!(n-\kappa+1)!} \\
=& \frac{\exp(-\eta^2)}{\lambda^\kappa}
\exp\left(-\frac{y}{\lambda}\right) \left\{ \sum_{n=N+2}^{\infty}
g^{(n)}(0, y)
\frac{(\eta^2/\lambda)^{n-\kappa+1}}{n!(n-\kappa+1)!}
+g^{(N+1)}(0, y)
\frac{(\eta^2/\lambda)^{N-\kappa+2}}{(N+1)!(N-\kappa+2)!}
\right\} \\
=& R_{N+1}(y) + \frac{\exp(-\eta^2)}{\lambda^\kappa}
\exp\left(-\frac{y}{\lambda}\right) \cdot g^{(N+1)}(0, y)
\frac{(\eta^2/\lambda)^{N-\kappa+2}}{(N+1)!(N-\kappa+2)!}.
\end{align*}
Since $R_N(y)<0$ and $g^{(N+1)}(y)<0$ for all $y\ge 0$ and $N$, we
have \eqref{eq:bound_2}. Now, based on \eqref{eq:bound_1} and
\eqref{eq:bound_2}, with given $\chibf_k$ and $\sigma_h^2$, we can
compute the required number $N$ of terms in the series to achieve
the desired level of accuracy since $\eta^2$ is known.

\begin{figure}[htbp]
\centerline{ \SetLabels
\L(0.23*-0.1) (a)\ CDF \\
\L(0.68*-0.1) (b)\ Residual error \\
\endSetLabels
\leavevmode
\strut\AffixLabels{
\scalefig{0.52}\epsfbox{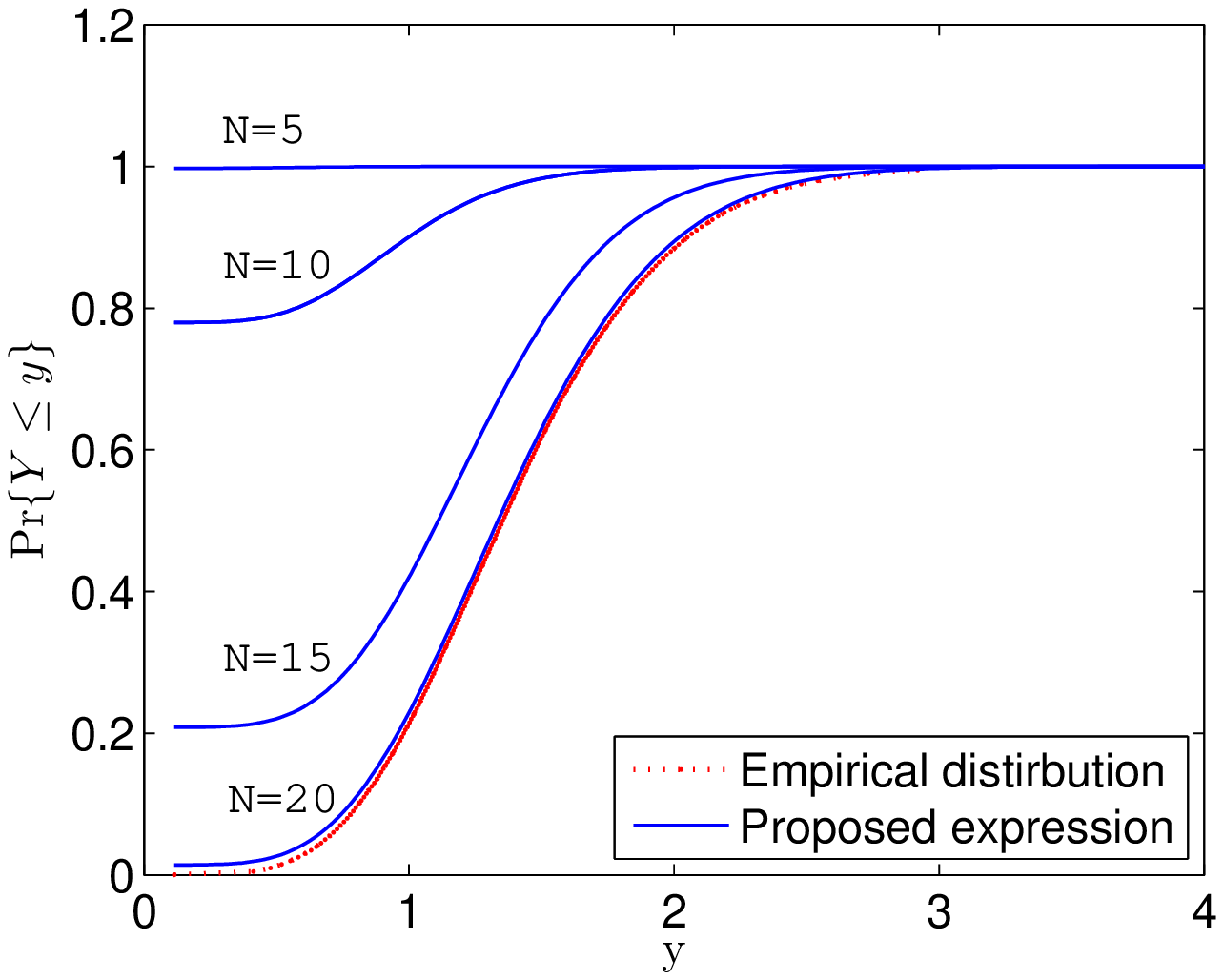}
\scalefig{0.52}\epsfbox{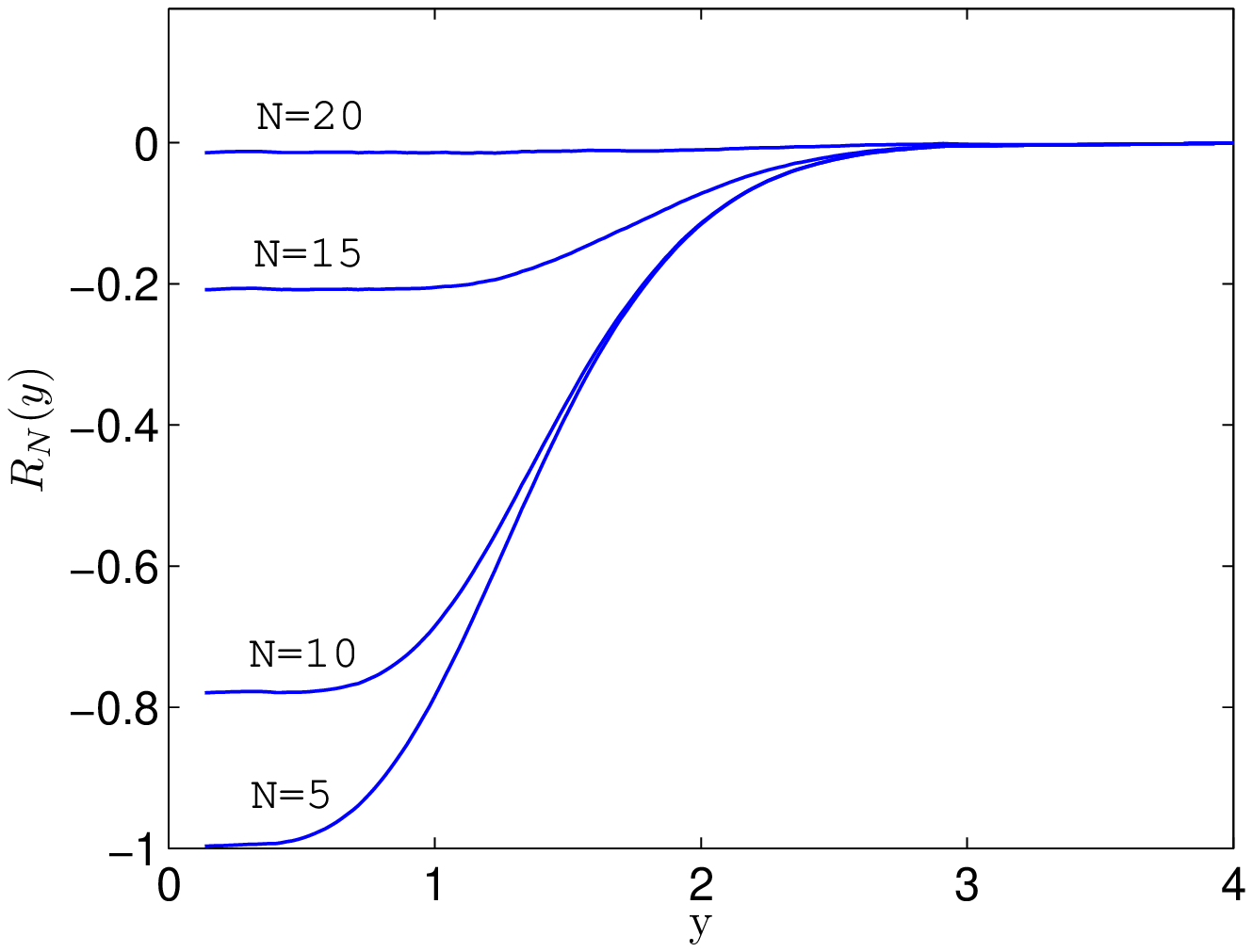} } } \vspace{1em}
\caption{number of variables = 4, {$\mubf=0.5\mathbf{1}$},
$\Bar{\Qbf}=\Ibf$, and $\Sigmabf=0.1\Ibf$.}
\label{fig:cdf+mse+d1+replyFinal}
\end{figure}

    Finally,
consider the worst case of $N = \kappa-2$ and $y=0$:
\[
R_{\kappa-2}(0) = -\exp(-\eta^2) \sum_{n=\kappa-1}^{\infty}
\frac{(\eta^2)^{n-\kappa+1}}{(n-\kappa+1)!} = -\exp(-\eta^2)
\sum_{m=0}^{\infty}\frac{(\eta^2)^m}{m!} = -1,
\]
where the second equality is by replacing $m=n-\kappa +1$. It is
easy to see that the worst case error is -1 in the identity
covariance matrix case. Fig. \ref{fig:cdf+mse+d1+replyFinal} shows
the performance of the proposed series expansion in the case of
the identity covariance matrix. The numerical results well match
our theoretical analysis in this subsection. From the figure, it
seems reasonable to choose $N \ge 20 \sim 30$ for accurate outage
probability computation.



\end{document}